\documentclass[a4paper,11pt]{article}
\usepackage[english]{babel}

\usepackage[utf8x]{inputenc}
\usepackage[T1]{fontenc}
\usepackage[a4paper,top=2cm,bottom=2cm,left=2cm,right=2cm,marginparwidth=2cm]{geometry}
\usepackage{multirow}
\usepackage{graphicx}
\usepackage{amsmath}
\usepackage{bbm}
\usepackage{indentfirst}  
\usepackage{makecell}
\usepackage{threeparttable}
\usepackage{setspace}
\usepackage{adjustbox}
\usepackage{array}
\usepackage{subfigure}
\usepackage{caption}
\usepackage{color}
\usepackage[table,xcdraw]{xcolor}
\usepackage{appendix}
\usepackage{natbib}

\usepackage{authblk}
\usepackage{amssymb,amsmath,bm,mathrsfs,makeidx,amsfonts,graphicx,amsthm,multirow}
\usepackage[figuresright]{rotating}
\usepackage{longtable}
\usepackage{algorithm,algorithmic} 
\newcolumntype{?}{!{\vrule width 1pt}}
\numberwithin{equation}{section}

\usepackage[pdftex,colorlinks=true,hypertexnames=false]{hyperref}
\definecolor{darkblue}{rgb}{0,0,.6}
\hypersetup{citecolor=darkblue,linkcolor=darkblue,urlcolor=darkblue}

\DeclareMathOperator{\sgn}{sgn}

\DeclareMathOperator{\argmin}{arg\,min}

\newcommand{\beqs}{\begin{eqnarray}}
\newcommand{\eeqs}{\end{eqnarray}}
\newcommand{\beqsn}{\begin{eqnarray*}}
	\newcommand{\eeqsn}{\end{eqnarray*}}

\newcommand{\non}{\nonumber \\}
\newtheorem{remark}{Remark}

\newtheorem{theorem}{Theorem}

\newtheorem{lemma}{Lemma}
\newtheorem{assumption}{Assumption}

\newcommand{\probP}{\text{I\kern-0.15em P}}

\begin{document}
	\title{Shocks-adaptive Robust Minimum Variance Portfolio for a Large Universe of Assets
	}
	\author[a]{Qingliang Fan}
	\author[b]{Ruike Wu}
	\author[c]{Yanrong Yang}

	\affil[a]{Department of Economics, The Chinese University of Hong Kong}
	\affil[b]{School of Economics, Xiamen University}
	\affil[c]{College of Business \& Economics, The Australian National University}
	
	\maketitle
	\begin{center}
		\textbf{Abstract}
	\end{center}
	This paper proposes a robust, shocks-adaptive portfolio in a large-dimensional assets universe where the number of assets could be comparable to or even larger than the sample size. It is well documented that portfolios based on optimizations are sensitive to outliers in return data.
We deal with outliers by proposing a robust factor model, contributing methodologically through the development of a robust principal component analysis (PCA) for factor model estimation and a shrinkage estimation for the random error covariance matrix. This approach extends the well-regarded Principal Orthogonal Complement Thresholding (POET) method \citep{fan2013large}, enabling it to effectively handle heavy tails and sudden shocks in data.  
 The novelty of the proposed robust method is its adaptiveness to both global and idiosyncratic shocks, without the need to distinguish them, which is useful in forming portfolio weights when facing outliers. We develop the theoretical results of the robust factor model and the robust minimum variance portfolio. Numerical and empirical results show the superior performance of the new portfolio.
	\onehalfspacing

	 \smallskip\\
	
	\onehalfspacing
\newpage
\section{Introduction}
Diversification is widely recognized in academia and industry as the predominant investment strategy. The mean-variance model \citep{markowitz1952portfolio} provides an elegant analytical solution given the population mean and (co)variances. Real data frequently display characteristics of heavy tails and high dimensionality, presenting challenges to the robustness of portfolio strategies that rely on optimization methods. It is well documented that portfolios based on optimization are highly sensitive to input data, often resulting in extreme portfolio weights and poor out of sample performance.

This paper studies the large dimensional minimum variance portfolio (MVP), inspired by the real-world investment scenarios that fund managers frequently encounter, where they must manage a vast universe of assets, often exceeding the sample size. The focus of the MVP is primarily on the covariance structure, presenting a more straightforward task in such complex environments \citep{jagannathan2003risk, demiguel2009generalized}. Specifically, we propose a robust minimum variance portfolio (R-MVP) to deal with outliers (or heavy tails) in the financial return data. By ``\emph{robust portfolio}'', specifically, it means (1), the portfolio constructed according to the presumed optimization rule (here it is the minimum variance) is robust to the features of real data deviating from classic assumptions such as i.i.d., the existence of second moment, and sub-Gaussian tails; (2), the portfolio is data-adaptive: it closely aligns with the MVP based on POET under regular conditions, and shifts towards a robust version when encountering outliers; and (3), its out of sample performance is robust. The robustness is guaranteed by our asymptotic theory in Section \ref{sec_theory}. When applying to the real data, the R-MVP is shown to be immune to highly sensitive weights due to potential financial turmoils, including both market and idiosyncratic shocks. For instance, the 2007-2008 financial crisis, triggered by the banking sector, spread across the entire market. During this period, we observed various shocks affecting the entire market as well as specific individual firms. We achieve the robust portfolio by proposing a robustified version of the factor model described below.

In asset pricing studies, the factor model is the main workhorse \citep{fama1992crosssection,bai2002determining} that evolved from the capital asset pricing model (CAPM). The theoretical work of arbitrage pricing theory (APT, \citealp{ross1976arbitrage}) laid the foundation for approximate factor models, which conveniently decompose the return data covariance matrix into a low-rank component and a sparse component \citep{Chamberlain83}. Following the literature, we estimate the high-dimensional covariance matrix using the popular approximate factor model. The celebrated POET \citep{fan2013large} method employs the factor model for dimension reduction and an adaptive thresholding method to regularize the error covariance matrix estimation.  However, the POET method  does not deal with heavy tails. The subsequent methods, e.g. \cite{fan2018aos} and \cite{fan2019robust} further discuss  heavier tails in the data but still need the existence of the second moment. In practical applications, the occurrence of heavy tails is common in financial returns. From a modeling perspective, many financial asset returns might not even possess a finite second moment.

To deal with above issues in the data, our procedure employs a robust PCA method for the estimation of the factor model \citep{M2005}, and a simple thresholding method for the residual covariance \citep{cai2011adaptive}. Our method enhances the robustness of the MVP by incorporating a weight function into the estimation procedure, which effectively mitigates the impact of outliers. The weight in the robust PCA is automatically derived from the data, serving as a generalization of classical PCA. This approach is adept at handling data that may contain outliers, without the need to prespecify their effects. 
We develop the desired theoretical properties of the robust factor model and R-MVP, including factor loading and common factors estimation consistency, error covariance estimation consistency, and oracle risk and Sharpe ratio consistency. 
Moreover, the proposed R-MVP is robust to the global (homogeneous) and idiosyncratic (heterogeneous) outliers, representing the shocks to the whole market and individual assets, respectively. In our model, we allow for either one or both of the two shocks without the need to specify them in the data. Additionally, we generalize the POET method and provide a weighted PCA approach which is of independent interest.

\subsection{MVP Preliminaries}
MVP has a surging appearance in recent studies on portfolio management \citep{ demiguel2009generalized, fan2012vast, ding2021high,caner2020sharpe}. Specifically, the classic MVP problem is:
	\begin{equation}
	\label{eq gmvp}
	\begin{split} \min \ W^{\top}\Sigma_r W,\ \ \ \ \ \ \\ \text{s.t.} \ W^{\top} 1_p = 1. \ \ \ \ \ 
	\end{split}
	\
	\end{equation}
	where $\Sigma_r$ is the  $p$-dimensional population covariance matrix of asset returns; $W$ is a $p\times1$ vector of asset weights in the portfolio, and $1_p$ is a $p\times1$ vector of 1's. 
	The analytical solution for MVP weight is:
	\begin{equation}
	\label{analytic solution for gmvp}
	W^* = \frac{\Sigma_r^{-1} 1_p}{1_p^{\top}\Sigma_r^{-1} 1_p}.
	\end{equation}
	The variance of MVP and the corresponding Sharpe ratio (SR) are:
	\begin{equation}
	\label{analytic solution for minimum variance}
	R_{\min}= {W^{*}}^{\top} \Sigma_r W^* = \frac{1}{1_p^{\top}\Sigma_r^{-1} 1_p},
	\end{equation}
	\begin{equation}
	\label{analytic solution for SR}
	{SR} = \frac{{W^{\ast}}^{\top}\mu}{ \sqrt{ {W^{*}}^{\top} \Sigma_r W^*}}=\frac{1_p^{\top}\Sigma_r^{-1} \mu }{\sqrt{1_p^{\top} \Sigma_r^{-1} 1_p}}.
	\end{equation}
	respectively, and $\mu$ is the expected excess returns (over a risk-free rate) of $p$ stocks.

\subsection{A Brief Literature Review}
In the literature, many papers propose robust methods for portfolio allocation. 
Robust estimation has been studied in the classic literature \citep{huber1964,huber1973} and it has recently drawn much interest in portfolio-related studies. 
\cite{demiguel2009or} propose robust portfolios based on M- and S-estimation techniques and show their performance in mostly low-dimension cases.
\cite{delage2010distributionally} form
the robust portfolio while considering a moment-based  uncertainty set.
\cite{plachel2019}  introduce a joint method for covariance regularization and robust optimization within the framework of minimum variance problem. \cite{Blanchet2022} propose a distributionally robust mean-variance portfolio, where the model uncertainty is imposed on the distribution of asset return. However, these methods usually deal with low-dimension situation. \cite{Hardle2022} propose a robust (minimum variance) Markowitz portfolio utilizing a projected gradient descent technique while avoiding the estimation of the covariance as a whole. The theoretical results of \cite{Hardle2022} are based on i.i.d. assumption of the return vectors.

A strand of literature deals with high-dimensionality issues in 
portfolios allocation. 
\cite{ledoit2017nonlinear} consider the nonlinear shrinkage method in portfolio selection. 
\cite{ao2019approaching} propose the MAXSER method that can achieve optimality in mean-variance portfolios with a large number of assets. \cite{li2022synthetic} further extend their studies by imposing factor structure.
\cite{ding2021high} developed a unified minimum variance portfolio under statistical factor models in high-dimensional situation. \cite{caner2020sharpe} use a residual based nodewise regression to estimate the large covariance matrix of assets and provide the out of sample rates for the Sharpe ratio. 
However, these methods do not consider the outliers and robust portfolio.

Our paper also connects to factor models and robust estimation. The factor model \citep{fan2013large,ait17,ding2021high} is a main workhorse in portfolio studies. \cite{giglio22} provides an excellent survey on recent developments in factor models, machine learning, and asset pricing. Drawing on the previous literature on robust estimation \citep{M2005}, we propose a robust PCA procedure for potential extreme returns in the covariance estimation.  \cite{fan2018aos} summarize  a unified framework for applying POET to various potentially heavy-tailed distributions and 
propose a Kendall's tau based  robust estimator of a large covariance matrix. \cite{fan2019robust} propose a robust covariance matrix estimator for factor models based on Huber loss function. The main difference between \cite{fan2019robust} and ours is that they assume observable factors, while we work on both observable and unobservable factors.  Our work is most closely related to \cite{fan2013large} and is an adaptively robustified version of the POET method.
	\subsection{Our Contributions}
The contributions of the paper are summarized in the following. 
\begin{enumerate}
    \item We develop a robust minimum variance portfolio, which is shown to have desirable theoretical and numerical results when the outliers can be either at the global, idiosyncratic, or both levels. We show the optimal risk consistency and optimal Sharpe ratio consistency with this approach. 
    \item We develop an adaptive robust estimation procedure for factor models that can handle high-dimensional return data with outliers. These outliers can originate from both factors (global) and error terms (idiosyncratic). Our robust estimation procedure can automatically adapt to regular or outlier data without the need to predefine the types of outliers. It operates by diminishing the influence of extreme observations on the portfolio weights. We derive the theoretical properties of the robust factor model estimation, allowing for the variance of common factors to diverge to infinity, which is new to the literature. Our theoretical derivation differs from that of POET, as the identification conditions of the two methods are different. While our primary focus is on robust portfolios, this paper also makes contribution to the literature on robust factor models. 
    \item Our robust investment strategy demonstrates superior out of sample performance compared to existing methods, achieving higher Sharpe ratios and better risk measures across various simulation settings and empirical studies. Additionally, our unified estimation approach is straightforward and easy to implement in practical scenarios.
\end{enumerate}

\subsection{Organization}
	
The rest of the paper is organized as follows. Section \ref{model} states the basic factor model. Section \ref{estimation} describes the estimation procedure.  Section \ref{sec_theory} provides the main theoretical results. Section \ref{simu} provides Monte Carlo simulation results. An empirical application using the S\&P 500 index and Russell 2000 component stock returns follows in Section \ref{empirical}. Section \ref{concl} concludes the paper. Supplementary material collects the proofs of the main theoretical results.

Throughout the paper,  $\lambda_{max}$ and $\lambda_{min}$ denote the maximum and minimum eigenvalue respectively, $\probP(A)$ is the probability of event A occurs.   We denote by $\|A\|_F, \|A\|$, and $\|A\|_1,$ the Frobenius norm, spectral norm and $L_1$-norm of a matrix A, defined respectively by $\|A\|_F = tr^{1/2}(A^\top A)$, $\|A\| = \lambda^{1/2}_{max}(A^\top A)$ and $\|A\|_1 = max_j \sum_i |a_{ij}|$,
respectively. If $A$ is a vector, $\|A\|_F$ and $\|A\|$ are equal to the Euclidean norm.
 
\section{Factor Model on Asset Returns}\label{model}
The investment universe in this study is composed of $p$ assets with observable return data. Following the classic arbitrage pricing theory \citep{ross1976arbitrage} and the approximate	factor model of \cite{Chamberlain83}, the expected return of financial assets is assumed to be driven by some finite number of common factors. Specifically, the assets return can be modeled as
\begin{eqnarray}\label{yr(1)}	r_{it}=b_{i}^{\top}F_t+e_{it}, \ \ \ i=1, 2, \ldots, p; \ \ t=1, 2, \ldots, T,
\end{eqnarray}
where $r_{it}$ is the excess of the risk-free return (hereafter also referred to as return for simplicity) of asset $i$ at time $t$, $F_t$ is the $m$-dimensional vector of common factors at time $t$, $b_{i}$ is the $m \times 1$ factor loading, which captures the relationship between common factors and financial asset $i$, {$m$ is the number of common factors}, and $e_{it}$ is the idiosyncratic error. 
 In practice, factors can be either observed or not\footnote{Some of those factors are commonly used by practitioners \citep{fama1992crosssection, fama1993common, fama2015five}, while others might not be universally accepted. \cite{harvey2021factor}, \cite{feng2020taming} have good discussions on this ``zoo of factors''.}. In this paper, we assume unobserved factors as the default setting. 
	
	The vector form of model (\ref{yr(1)}) is 
	\begin{eqnarray}\label{vecform}
	r_t=BF_t+e_t, \ \ \ t=1, 2, \ldots, T, 
	\end{eqnarray}
	where $r_t=\left(r_{1t}, r_{2t}, \ldots, r_{pt}\right)^{\top}$ is the vector of returns for $p$ assets, $B=\left(b_{1}, b_{2}, \ldots, b_{p}\right)^{\top}$ is the factor loading matrix, and $e_t=\left(e_{1t}, e_{2t}, \ldots, e_{pt}\right)^{\top}$. 
 Without loss of generality, we assume that the common factor $F_t$ is regularized such that its covariance matrix $\Sigma_F$ is the identity matrix. After 
  imposing a factor structure \eqref{vecform} on financial return data, we can decompose its covariance matrix as follows: 
	\begin{eqnarray}\label{yr(2)}
	\Sigma_r =BB^{\top}+\Sigma_e, \ \ \ t=1, 2, \ldots, T, 
	\end{eqnarray}
 Based on the factor-structured covariance matrix in (\ref{yr(2)}), the proposed MVP strategy has the representation
	\begin{equation}
	\label{tvmvp}
	W^{*} = \frac{(BB^\top + \Sigma_e)^{-1} 1_p}{1_p^{\top}(BB^\top + \Sigma_e)^{-1} 1_p},
	\end{equation}
which is the $p\times1$ vector of optimal asset weights.

	
	\section{Robust Estimation Approach}\label{estimation}
We begin by obtaining the factor loading estimator  $\hat{B}$ and common factor estimator $\hat{F}_t$. Consider the following optimization problem, 
	\begin{align}
	\mathop{min}_{B,F_t} \frac{1}{T}\sum^{T}_{t=1}\rho_\tau \left(\left\| r_t -  {B}{F}_t\right\|^2 \right),
	\label{PCA loss}
	\end{align}
	where $\rho_\tau(\cdot)$ is the Huber loss function, given by 
 \begin{align}
	\rho_\tau(x) = \left\{\begin{array}{cc}
	\frac{1}{2}x & \sqrt{x} \leq \tau
	\\ \tau\sqrt{x} -\frac{\tau^2}{2} & \sqrt{x} >\tau .
	\end{array} \right.	
	\end{align}
  By taking the derivative of (\ref{PCA loss}) with respect to $F_t$ and $B$ together with identification condition $p^{-1}B^\top B = I_m$, it can be shown that the factor loading estimator $\hat{B}$ is $\sqrt{p}$ times the corresponding (to its $m$ largest eigenvalues)
	eigenvectors of $\hat{\mathbf{V}}$ and $\hat{F}_t = p^{-1}\hat{B}^\top r_t$  where 
	\begin{eqnarray}
	&&\hat{\mathbf{V}}=\frac{1}{T}\sum^{T}_{t=1}\omega_tr_tr_t^{\top}, \label{weighted covariance}
	\end{eqnarray}
	and 
	\begin{eqnarray}
	\label{weights omegat}
	\omega_t = \left\{\begin{array}{cc}
	\frac{1}{2} & ||r_t - BF_t|| \leq \tau
	\\ \frac{\tau}{2}\frac{1}{\sqrt{r_t^\top r_t - r_t^\top B B^\top r_t/p}} &  ||r_t - BF_t|| > \tau
	\end{array} \right.
	\end{eqnarray}
 is the weight  function. The weighting sequence {$\{ \omega_t\}_{t=1}^T$} is completely data-driven {and is shock-adaptive.} 
 The weight function $\omega_t$ automatically assigns the lower values to those periods in which the shocks are large so that the covariance matrix \eqref{weighted covariance} is less affected by the outliers. Therefore, our portfolio can be robust. Note that  the classic principal component estimation procedure applied in \cite{fan2013large} is a special case of our model by setting $\omega_t 
 = 0.5$ for all $t$.
 
 \begin{remark}
      For robust loss function $\rho_\tau(\cdot)$, we simply choose the commonly-used Huber loss function. Other types of robust loss function could be also applied, e.g. the bisquare loss function $min\{1, 1-(1-x/\tau^2)^3\}$ used by \cite{M2005}.
 \end{remark}

\subsection{Relation with POET}
The estimation robustness is achieved by regulating both the common factors and error terms. Recall that $\hat{B}$ is given by  $\sqrt{p}$ times the corresponding eigenvectors of $\hat{\mathbf{V}}$, and consider the following minimization problem using the transformed return (hence the transformed factors and errors)
\begin{align}
    \argmin_{F,B} \| \tilde{R}^\top - FB^\top \|_F^2, \label{eq: transform PCA obj}
\end{align}
subject to the identification condition $p^{-1}B^\top B = I_m$, where $\tilde{R} = \left( \tilde{r}_1,\ldots\tilde{r}_T \right)$, $\tilde{r}_t = B\tilde{F}_t + \tilde{e}_t$, $\tilde{F}_t = \omega_t^{1/2}F_t$ and $\tilde{e}_t = \omega_t^{1/2}e_t$. For the minimization problem \eqref{eq: transform PCA obj}, it is shown  \citep{bai2002determining, bai2003} that the columns of estimated factor loading from \eqref{eq: transform PCA obj} are $\sqrt{p}$ times the eigenvectors corresponding to the $m$ largest eigenvalues of the matrix $\tilde{R}\tilde{R}^\top = \sum_{t=1}^T \tilde{r}_t \tilde{r}_t^\top = \sum_{t=1}^T \omega_t {r}_t {r}_t^\top$, which is $\hat{B}$. As a result, the estimated factor loading from solving the problem \eqref{PCA loss} enjoys the same property as that from problem \eqref{eq: transform PCA obj}. Note that $\tilde{r}_t$ is the regularized asset returns, which consist of regularized factors $\tilde{F}_t$ and regularized error terms $\tilde{e}_t$, where $\tilde{F}_t$ and $\tilde{e}_t$ are the robustified version of $F_t$ and $e_t$, respectively. 
Therefore, it is equivalent to get the robust estimation of \eqref{PCA loss} from the PCA of a transformed factor model whose outliers in both factors and residuals are taken care of. Conceptually,  compared to the estimation procedure in  \cite{fan2013large}, which is developed based on the original return $R$ itself, our estimation achieves robustness using the robustified return $\tilde{R}$.
 
\subsection{Computation Algorithm}
 Since $\omega_t$ depends on the unknown factor loading, in practice,  we can obtain numerical solutions for $\hat{B}$ and $\hat{F}_t$ in  \eqref{PCA loss} using Algorithm \ref{algonew} shown below.  

\begin{algorithm}[H]		\caption{\label{algonew} {Procedure for robust PCA.}}                          	\begin{algorithmic}
			\footnotesize
			\STATE 1. Set $i \leftarrow 1$. 
			\STATE 2. Compute initial values for optimizing variables by using conventional PCA. In details, we calculate ${B}^{(0)} =\sqrt{p}\times eig_m(RR^\top)$ and ${F}_{t}^{(0)} = B^{(0)}r_t/p$
where $eig_{m}(A)$ takes corresponding (to its m largest eigenvalues) eigvectors of matrix A.
			
			\REPEAT %
			\STATE $\quad$  (1). Compute weighting sequence $\omega_t$ by (\ref{weights omegat}) based on $B^{(i-1)}$,  and further obtain $\hat{\mathbf{V}} = \frac{1}{T}\sum_{t=1}^{T}\omega_t r_t r_t^\top$. 
			\STATE $\quad$ (2). Update ${B}^{(i)}$ by $\sqrt{p} \times eig_m(\hat{\mathbf{V}})$,  and ${F}_t^{(i)} = {B^{(i)}}^\top r_t/p$ for $t=1,\ldots,T$. 
			\STATE $\quad$ (3). $\Delta = abs(||r_t  - {B}^{(i-1)}{F}_{t}^{(i-1)}||^2-||r_t  - {B}^{(i)}{F}_{t}^{(i)}||^2)$, $i \leftarrow i + 1$.
			\UNTIL{$i > MaxI$ or $\Delta < tol$.} 
		\end{algorithmic}
	\end{algorithm}	
	
 Algorithm \ref{algonew} requires five inputs: the $T \times p$ return data matrix $R$, the number of common factors $m$, threshold parameter $\tau$, the maximum steps $MaxI$ and the tolerance level $tol$. {Based on this iterative calculation, we can choose the value of $\tau$ as the empirical upper quantile of $\|r_t - B^{(i-1)}F_t^{(i-1)}\|$. In practice, we recommend using the 0.9-th quantile, which has shown good performance in our numerical studies. }
 The proposed algorithm is efficient and  it typically stops within a few steps.  

\begin{remark}
The practitioners can apply the standard approach proposed by \cite{bai2002determining} to determine the number of common factors $m$, which is given by 
\begin{align}
	\hat{m} = \mathop{arg min}_{0 \leq m_1 \leq M } log\left\{\frac{1}{pT}\left\|\tilde{R} - \hat{B}^{m_1}\hat{F}^{m_1 \top}  \right\|^2_F \right\} + m_1 g(T,p), 
 \label{eq: select factor number}
	\end{align}
	where $M$ is a predetermined upper bound for the number of factors, $\tilde{R}$ is transformed returns based on classic factor loading estimation, $\hat{B}^{m_1}$ and $\hat{F}^{m_1}$ are estimated factor loadings and common factors conditional on factor number $m_1$, and $g(T,p)$ is a penalty function of $p$ and $T$, e.g., $g(T,p) ={(p+T)}log({pT}/{(p+T)})/{pT}$. In addition, the proposed algorithm can be further modified to iteratively update the estimation of the number of common factors. In more details, one can re-estimate the factor numbers via criterion \eqref{eq: select factor number}  after updating the weight sequence in step (1) of Algorithm \ref{algonew}.
\end{remark}

After obtaining $\hat{B}$ and $\hat{F}_t$, we can compute the estimated residuals $\hat{e}_t = r_t  - \hat{B}\hat{F}_t$, and thus estimate the error (residual) covariance matrix $\Sigma_{e}$. In this paper, we follow the studies of \cite{cai2011adaptive} and \cite{fan2013large}, and impose approximate sparsity assumption on $\Sigma_{e} = (\sigma_{e,ij})_{p\times p}$: for some $q \in [0,1]$,
	\begin{align}
	\label{sparse kp}
	\kappa_q=\max_{i\leq p}\sum_{j\leq p}|\sigma_{e,ij}|^{q},
	\end{align}
	does not grow too fast as $p\rightarrow \infty$. In particular, $\kappa_q$ is the maximum number of non-zero elements in each row when $q = 0$.	 Then, we apply the shrinkage estimation to the off-diagonal elements in sample covariance matrix. Specifically, the error covariance estimator is given by
	\begin{align}
	\hat{\Sigma}_e = (\hat{\sigma}_{\hat{e},ij})_{p\times p}, \quad \hat{\sigma}_{\hat{e},ij} = \left\{ \begin{array}{ll} s_{\hat{e},ii} & i=j \\ \chi_{ij}(s_{\hat{e},ij}) & i\neq j \end{array} \right.
 \label{eq: adpative thresholding estimator}
	\end{align}
	where $s_{\hat{e},ij} = (1/T)\sum_{t=1}^{T}\hat{e}_{it}\hat{e}_{jt}$ is the $(i,j)$-th element of sample covariance matrix of $\hat{e}_t$, $\chi_{ij}(\cdot)$ is a shrinkage function satisfying
	$\chi_{ij}(z)=0$ if $|z| \leq \tau_{ij}$, and $|\chi_{ij}(z)-z|\leq \tau_{ij}$, where $\tau_{ij}$ is a positive threshold. The general $\chi_{ij}(\cdot)$ includes many commonly-used thresholding functions such as soft thresholding ($\chi_{ij}(z) = \sgn(z)(|z|-\tau_{ij})^{+}$, $(z)^+ = max\{z,0\}$) and hard thresholding ($\chi_{ij}(z) = z\mathbb{I}(|z|\geq \tau_{ij})$,  where $\mathbb{I}(\cdot)$ is an indicator function). 
	In this paper, we use the adaptive thresholding method developed by \cite{cai2011adaptive}, with entry-adaptive $\tau_{ij} = c_\tau \tilde{\varsigma}_T\sqrt{\hat{\theta}_{ij}}$, where $c_\tau > 0$ is a large constant, {$\tilde{\varsigma}_T = 1/\sqrt{p} + \sqrt{logp /T}$}, 
	and $\hat{\theta}_{ij} = (1/T)\sum_{t=1}^{T}(\hat{e}_{it}\hat{e}_{jt}-{s}_{\hat{e},ij})^2$.  It is worth noting that the shrinkage estimator in the residual covariance matrix also helps the portfolio to achieve certain robustness, since it regulates the extremely large values in the sample covariance matrix.

With estimated factor loading $\hat{B}$ and error covariance estimator $\hat{\Sigma}_e$, we can obtain the return covariance estimator $\hat{\Sigma}_r=\hat{B}\hat{B}^{\top}+\hat{\Sigma}_e$.  Then, the proposed R-MVP is given by
	\begin{equation}
	\label{rmvp}
	\hat{W}_R = \frac{\hat{\Sigma}_r^{-1} 1_p}{1_p^{\top}\hat{\Sigma}_r^{-1} 1_p}.
	\end{equation}

	\section{Assumptions and Asymptotic Theory}\label{sec_theory}
	We first present the assumptions to facilitate the theoretical derivations.
 \subsection{Basic Assumptions}\label{assumptions}

 \begin{assumption}\label{assu: factor loading}
 \begin{enumerate}
     \item[(a)] $B^\top B$ is diagonal and 
     $\|p^{-1}B^\top B - I_m\|_F =O_p(p^{-1/2})$.
     \item[(b)] There exists constant $M >0$ such that for all $i\leq p$, $E\|b_i\|^2 < M$.
 \end{enumerate}
\end{assumption}
 
Assumption \ref{assu: factor loading}(a) is similar to Assumption 3.1 in \cite{fan2013large}. It is one of the most common assumptions in the literature of approximate factor models. It implies that the first $m$ eigenvalues of $B^\top B$ grow at rate $O(p)$ and requires the factors to be pervasive, i.e., to impact a non-vanishing proportion of individual time series. Under this assumption and other regularity conditions, the factor loadings and common factors can be consistently estimated (up to some rotations).

 \begin{assumption}
\label{assum: error terms}
    \begin{itemize}
        \item[(a)] 
        $\{e_t\}$ is strictly stationary and ergodic with zero mean and finite covariance matrix $\Sigma_e$.
\item[(b)]
 There exists  constants $c_1, c_2>0$ such that $\lambda_{\min}\left(\Sigma_{e}\right)>c_1$, $\lambda_{max}(\Sigma_{e})< c_2$, and $\min_{i\leq p ,j\leq p } var({e}_{it}{e}_{jt}) > c_1$.
    \end{itemize}
\end{assumption}

Assumption \ref{assum: error terms}(a) is general for the error term, requiring that the error term has a zero mean and a finite covariance matrix.
Assumption \ref{assum: error terms}(b) requires that $\Sigma_{e}$ be well-conditioned and {ensures that the largest eigenvalue of $\Sigma_r$ grow at rate $O(p)$}.

Define $E(F_{it}^2) = C \delta$ for $i\leq p$, $t\leq T$, and $\max_{t\leq T}E \| F_t\| = C\Delta$, where $C$ is some positive constant. In classic factor model setups, $\delta$ and $\Delta$ are often assumed to be of constant order, as in \cite{bai2002determining,bai2003,fan2013large}.
In this paper, we derive asymptotic results while allowing for situations where the factor may not have a finite second order moment or even a finite maximum first-order moment, and that both $\delta$ and $\Delta$ can diverge to positive infinity at some rate.  If $F_t$ has exponential tails such that for any $j\leq m$, $\probP(|F_{jt}| > s_1)\leq exp(-(s_1/s_2)^{s_3})$, it is clear that $\delta$ is finite, and it can be shown $\Delta = O_p\left((logT)^{1/s_3}\right)$ by Bonferroni's method. If $F_t$ is independent and identically distributed, it is clear that $\delta^{1/2}$ has the same order of magnitude of $\Delta$.  We note that our theoretical derivation can also be extended to situations where the covariance matrix of the error term does not exist.

 The following Assumption \ref{assum: error terms factor trans} regulates the behavior  of  transformed factors $\tilde{F}_t$ and transformed error terms $\tilde{e}_t$. 

 \begin{assumption}
\label{assum: error terms factor trans}
    \begin{itemize}
\item[(a)] 
$\{\tilde{e}_t, \tilde{F}_t\}_{t\geq 1}$ is strictly stationary, and $E(\tilde{e}_{it}) = E(\tilde{e}_{it}\tilde{F}_{jt}) = 0$ for all $i \leq p, j \leq m$ and $t\leq T$.
\item[(b)] There exists $\varphi_1, \varphi_2 > 0$ and $d_1, d_2 > 0$, such that for any $s>0$, $i\leq p$ and $j\leq m$,
            $$
            \probP(|{e}_{it}| > s) \leq  exp(-(s/d_1)^{\varphi_1}), \probP(|\tilde{e}_{it}| > s) \leq  exp(-(s/d_1)^{\varphi_1}),  \probP(|\tilde{F}_{jt}| > s) \leq  exp(-(s/d_2)^{\varphi_2}).
            $$
            \item[(c)] $cov(\tilde{F}_t) = I_m$, $\|T^{-1}\tilde{F}^\top \tilde{F} - I_m\| = o_p(1)$, where $\tilde{F} = (\tilde{F}_1,\ldots,\tilde{F}_T)^\top$.
            \item[(d)] There exists a  positive constant $C$ such that  $\|\tilde{\Sigma}_e\|_1 \leq C$, where $\tilde{\Sigma}_e$ is the covariance matrix of transformed error term $\tilde{e}_t$.
    \end{itemize}
\end{assumption}

\begin{assumption}
    \label{assu: mixing condition}
    There exists $\varphi_3 >0$  such that $(log p)^{2/\varphi - 1} = o(T)$ and $(logp)^{6/\tilde{\varphi}-1} = o(T)$ where 
$\varphi = 1.5\varphi_1^{-1} + 1.5\varphi_2^{-1} + \varphi_3^{-1}$, $\tilde{\varphi}^{-1}= 3\varphi_1^{-1} + \varphi_3^{-1} > 1$ and $3\varphi_2^{-1} + \varphi_3^{-1} >1$, and $C>0$ satisfying: for all $t\in \mathbb{Z}^+$(the set of positive integers),
$$
\alpha(t) \leq exp(-Ct^{\varphi_3})
$$
where $\alpha$ is $\alpha$-mixing coefficient defined based on $\sigma$-algebras generated by $\{e_t, \tilde{e}_t,\tilde{F}_t \}$.
\end{assumption}

 Assumption \ref{assum: error terms factor trans}(b) and Assumption \ref{assu: mixing condition} specify the exponential-type tails and mixing dependence for $e_t$, $\tilde{e}_t$, and $\tilde{F}_t$, respectively. These conditions allow us to apply the Bernstein type inequality for the weakly dependent data and thus help to analyze the terms such as $({1}/{T})\sum_{t=1}^T{e}_{it}{e}_{jt}$ and $\frac{1}{T}\sum_{t=1}^T\tilde{F}_t\tilde{e}_{it}$. Similar conditions are also imposed in \cite{fan2011high} and \cite{fan2013large}. Assumption \ref{assum: error terms factor trans}(c) requires that the transformed  factor $\tilde{F}_t$ is regularized with covariance matrix being an identity matrix, which is often assumed for simplicity in literature, e.g., \cite{fan2013large, li2022integrative}. 
Assumption \ref{assum: error terms factor trans}(d) ensures that the largest $m$ eigenvalues of  transformed sample covariance matrix $T^{-1}\tilde{R}\tilde{R}^\top$ diverges to infinity at order $p$ where $\tilde{R} = (\tilde{r}_1,\ldots,\tilde{r}_T)$, and guarantees the consistency of factor number estimation, also refer to  Assumption 3.2 of \cite{fan2013large}. To sum up, Assumptions \ref{assum: error terms factor trans} and \ref{assu: mixing condition} are regularization condition on the transformed common factors and error terms. As we consider heavy-tailed data, such conditions specify that the robustified data $\tilde{r}_t$ has certain good behavior such that the traditional PCA is applicable.

Let  $\varepsilon_i = (e_{i1},\ldots,e_{iT})^\top$ and $\tilde{\varepsilon}_i = (\tilde{e}_{i1},\ldots,\tilde{e}_{iT})^\top, i = 1,\ldots, p$,
then we  additionally impose the following regularity conditions.

\begin{assumption}
\label{assum: regularization}
    \begin{itemize}
        \item[(a)] $\max_i \sum_{s=1}^p|E(\tilde{\varepsilon}_s^\top\tilde{\varepsilon}_i)|/T = O(1) $. 
        \item[(b)]  For all $s,i\leq p$, $$E\left(\tilde{\varepsilon}_s^\top \tilde{\varepsilon}_i - E(\tilde{\varepsilon}_s^\top \tilde{\varepsilon}_i ) \right)^4  = O(T^2).$$ 
    \item[(c)]  For $j \leq p$ and $t\leq T$, we have $E\left\| \sum_{j=1}^p b_j {e}_{jt}\right\|^4 = O(p^2)$.
    \item[(d)] $T = o(p^2), \delta = o(p), \Delta^2 = o(p), \delta \sqrt{p} /T =o(1), \Delta^2 \sqrt{p}/ T =o(1)$, $\kappa_q \zeta_T^{1-q} = o(1)$ where $\zeta_T = \frac{\delta^{1/2}+T^{1/4}+\Delta}{\sqrt{p}} + \frac{\delta^{1/2}p^{1/4}}{\sqrt{T}}$.
    \end{itemize}
\end{assumption}

Assumption \ref{assum: regularization}(a)-(c) are analogous to condition 3.4 in \cite{fan2013large}. Since $\tilde{e}_{st}$ is strictly stationary over $t$ for all $s$,  Assumption \ref{assum: regularization}(a) is equivalent to require that $\max_i \sum_{s=1}^p |\tilde{\sigma}_{e,si}| = O(1)$ where $\tilde{\sigma}_{e,si} = E(\tilde{e}_{st}\tilde{e}_{it})$, which is a type of sparsity condition. Similarly, Assumption \ref{assum: regularization}(b) can be rewritten as 
    $
   E\left[\frac{1}{T}\sum_{t=1}^T(\tilde{\varepsilon}_{st} \tilde{\varepsilon}_{it} ) - E(\tilde{\varepsilon}_{st} \tilde{\varepsilon}_{it})  \right]^4 = O(T^2).
    $
 Assumption \ref{assum: regularization}(d) ensures the convergence of $\hat{b}_i$, $\hat{F}_t$, and $\hat{\Sigma}_e$ to corresponding population version, respectively.

The following two assumptions are necessary for the consistency of the optimal minimum variance estimator and the Sharpe ratio estimator, which are obtained using the plug-in method based on the analytical solutions given in equations \eqref{analytic solution for minimum variance} and \eqref{analytic solution for SR}.

 \begin{assumption}\label{assu: risk}
		The minimum risk $R_{\min}=\frac{1}{1_p^\top\Sigma_{r}^{-1}1_p}\asymp p^{1-\eta}$, where $\eta$ is a  constant satisfying $p^{\eta}\zeta_T^{1-q}\kappa_q=o(1)$. 
	\end{assumption}
	
	\begin{assumption}\label{assu: sr}
		Suppose the term $1_p^\top\Sigma_{r}^{-1}\mu\asymp p^{1-\phi}$, where $\phi$ is a  constant satisfying $p^{\phi}\zeta_T^{1-q}\kappa_q =o(1)$.
  \end{assumption}
	
Assumptions \ref{assu: risk} and \ref{assu: sr} guarantee the convergence of minimum risk estimator and Sharpe ratio estimator. Similar assumptions can be referred to  \cite{ding2021high} and \cite{fan2022}.
\begin{remark}
A simple example in \cite{ding2021high} demonstrates that it is reasonable to assume $R_{min}$ and ${\mathbf{1}_p^\top \Sigma_r^{-1} \mu}$ are of the order of powers of $p$.
    Suppose that $r_i = \beta_i f + e_i$, $(\beta_i)_{1\leq i\leq p}$ are i.i.d. with mean 1 and standard deviations $\sigma_\beta$, $E(f) = \sigma(f) = 1$, and $\Sigma_e =   I $. Under such a model, $\Sigma_r =\tilde{\beta}\tilde{\beta}^\top + I$ where $\tilde{\beta} = (\beta_1,\ldots,\beta_p)^\top$, and $\mu = 1_p$.  If $\sigma_\beta > 0$, by Proposition 2.2 of \cite{ding2021high}, this model corresponds to a well-diversifiable case with $R_{min}$ and $1/{1}_p^\top\Sigma_r^{-1}\mu$ converging to zero at rate $O(1/p)$, and thus $\eta = 2$, and $\phi = 0$. If $\sigma_\beta = 0$, it is easy to see that the minimum variance portfolio is equal allocation, leading to a minimum risk of $1 + 1/p$, and thus $\eta = \phi = 1$.
\end{remark}

\subsection{Asymptotic Theory}

Define $\tilde{V} = \mathrm{diag}(\tilde{\lambda}_{1},\ldots, \tilde{\lambda}_{m})$, where $\tilde{\lambda}_i, i =1,\ldots,m$ is the i-th largest eigenvalues of $p^{-1}\tilde{R}^\top \tilde{R}$ in descending order, and then define $\tilde{H} = p^{-1}\tilde{V}^{-1}\hat{B}^\top B \tilde{F}^\top \tilde{F}$. The following lemma  shows the asymptotic properties of  the estimated factor loading $\hat{b}_i$ and common factors $\hat{F}_t$.

\begin{lemma}
    \label{lemma B and F}
    Suppose Assumptions \ref{assu: factor loading}- \ref{assu: mixing condition}, \ref{assum: regularization}(a)-(c), let $\varpi_T = \frac{1}{\sqrt{p}}+\frac{p^{1/4}}{\sqrt{T}}$. Then we have 
    \begin{align*}
      &\max_{i\leq p}\left\vert\left\vert  \hat{b}_i - \tilde{H}b_i\right\vert \right\vert = O_p(\varpi_T),
      \\ &
      \max_t \| \hat{F}_t - \tilde{H}F_t \| = O_p\left( \frac{\delta^{1/2}+ T^{1/4} + \Delta}{\sqrt{p}}  + \sqrt{\frac{\delta log p}{{T}}}\right).
    \end{align*}
\end{lemma}

The following Theorems \ref{theorem sigma eps} and \ref{Theorem sigma inverse convergence} demonstrate the asymptotic results of estimated covariance matrices.
\begin{theorem}
    \label{theorem sigma eps}
    Suppose  Assumptions \ref{assu: factor loading} - \ref{assum: regularization} hold true. Let  $\tau_{ij} = C \varsigma_T\sqrt{\hat{\theta}_{ij}}$ where $\varsigma_T = \sqrt{\frac{logp}{T}}+\sqrt{\frac{(\delta^{1/2}+T^{1/4}+\Delta)^2}{p}+\frac{\delta\sqrt{p}}{T}}$.
    Then there is a constant $C > 0$ 
     such that
$$
\left\vert\left\vert \hat{\Sigma}_e - \Sigma_e \right\vert\right\vert = O_p(\zeta_{T}^{1-q}\kappa_q) = o_p(1),
$$
where $
    \zeta_T = \frac{\delta^{1/2}+ T^{1/4} + \Delta}{\sqrt{p}} + \frac{\delta^{1/2}p^{1/4}}{\sqrt{T}}$. 
The eigenvalues of $\hat{\Sigma}_e$ are all bounded away from 0 with probability approaching 1, and 
$$
\left\vert\left\vert \hat{\Sigma}_e^{-1} - \Sigma_e^{-1} \right\vert\right\vert = O_p(\zeta_{T}^{1-q}\kappa_q) = o_p(1).
$$
 \end{theorem}

	\begin{theorem}
    \label{Theorem sigma inverse convergence}
    Under the conditions of Theorem \ref{theorem sigma eps}, we have 
    $$
    \left\| \hat{\Sigma}_r^{-1} -\Sigma_r^{-1} \right\| = O_p\left(\zeta_T^{1-q}\kappa_q \right) = o_p(1),
    $$
where $\Sigma_r = BB^\top + \Sigma_e$, and $\zeta_T$ is defined in Theorem \ref{theorem sigma eps}.
\end{theorem}

The following implications stem from Theorems \ref{theorem sigma eps} and \ref{Theorem sigma inverse convergence}: 
(1) We utilize the identification condition $p^{-1}B^\top B = I_m$, which is different from the condition $T^{-1}F^\top F = I_m$ in \cite{fan2013large}. This difference implies different convergence rates for the estimated factor loading and, consequently, the covariance matrix. Specifically, the convergence rate {derived from \cite{fan2013large}} of $\hat{\Sigma}_e$ and $\hat{\Sigma}_r^{-1}$ to the corresponding population version is $\kappa_q (\sqrt{\frac{logp}{T}}+\frac{1}{\sqrt{p}})^{1-q}$. For comparison, the convergence rate derived by us is $\kappa_q \left(  \frac{\delta^{1/2}+ T^{1/4} + \Delta}{\sqrt{p}} + \frac{\delta^{1/2}p^{1/4}}{\sqrt{T}}
\right)^{1-q}$. It is clear that our matrix estimator converge at a slower rate than the corresponding matrix estimator of \cite{fan2013large} even if the data satisfied some assumptions that $\delta$ and $\Delta$ are all bounded. (2) The orders of magnitude  of $\delta$ and $\Delta$ affect the convergence rate of estimated quantities. To achieve the convergence shown in Theorem \ref{theorem sigma eps} and \ref{Theorem sigma inverse convergence},  the maximum first order moment of common factors is allowed to diverge to infinity at a rate of at most $\sqrt{p}$, and the second moment of factors is allowed to diverge to infinity at a rate of at most $min\{T/p^{1/2},p\}$. (3) The asymptotic results are unchanged  if we replace $\varsigma_T$ by $\tilde{\varsigma}_T$ in threshold $\tau_{ij}$, where $\tilde{\varsigma}_T/\varsigma_T = o(1)$, e.g., $\tilde{\varsigma}_T = 1/\sqrt{p} + \sqrt{logp /T}$ used in \cite{fan2013large} and our empirical applications. (4) The convergence rate of estimated inverse covariance matrix of return is the same as that of estimated error covariance matrix.

Next, we turn to  consider the rates of convergence for the minimum risk estimator and Sharpe ratio estimator, respectively. Recall that, the minimum risk estimator and Sharpe ratio estimator  are respectively given by 
	\begin{eqnarray}\label{mr_y}
\hat{R}_{min}=\frac{1}{1_p^\top \hat{\Sigma}_r^{-1}1_p} 	\end{eqnarray}
and
\begin{eqnarray}\label{sr_y}
\hat{SR}=\frac{1_p^{\top}\hat{\Sigma}_{r}^{-1}\hat{\mu}}{\sqrt{1_p^{\top}\hat{\Sigma}^{-1}_{r} 1_p}},
\end{eqnarray} 
 where $\hat{\mu}=\hat{B}\frac{1}{T}\sum^{T}_{t=1}\hat{F}_t$.
The asymptotic behavior of the minimum risk estimator with respect to the oracle risk defined in (\ref{analytic solution for minimum variance}) is demonstrated as follows: 
	\begin{theorem}
    \label{theorem risk convergence}
    Under conditions of Theorem \ref{Theorem sigma inverse convergence}, suppose Assumption \ref{assu: risk} holds. Then we have the convergence rate for minimum risk estimator:
    $$
\left|\frac{\hat{R}_{min}}{R_{min} } -1 \right| = O_p\left(p^\eta \zeta_T^{1-q}\kappa_q \right) = o_p(1).    
    $$
\end{theorem}

With a similar ratio criterion, we further evaluate the Sharpe ratio estimation. 
	\begin{theorem}
    \label{theorem SR convergence}
    Under conditions of Theorem \ref{Theorem sigma inverse convergence}, suppose Assumptions \ref{assu: risk} and \ref{assu: sr} hold. Then we have the convergence rate for Sharpe ratio estimator:
    $$
\left|\frac{\hat{SR}}{SR} -1 \right| = O_p\left(p^{(\phi+\eta)}\zeta_T^{1-q}\kappa_q  \right) = o_p(1).    
    $$
\end{theorem}

Theorems \ref{theorem risk convergence} and \ref{theorem SR convergence} demonstrate the consistency of the minimum risk and Sharpe ratio of the proposed robust portfolio, which are desirable in portfolio allocation. We note that the consistency orders are primarily affected by the estimation of the sparse error covariance matrix. The similar theoretical results for the convergence of minimum risk and Sharpe ratio based on factor model can be referred to the work of  \cite{ding2021high} and \cite{fan2022}. 


	
\section{Simulation}\label{simu}

	In this section, we conduct simulations to evaluate the finite sample performance of the proposed portfolio R-MVP. 
	
	\subsection{Model Set-up}
	To mimic the real world scenarios, we design the Monte Carlo simulations such that investors make decisions based on historical data available at current time $t$, and then hold portfolios for a period of time $T$. We evaluate the portfolios based on their out of sample performances. 
 
 We first introduce the data generating process (DGP) without shocks. This is named \textbf{DGP 1} which is the baseline setting.
	We generate the $p\times (2T)$ return data $\{r_t\}_{t=1}^{2T}$ based on the factor model defined in (\ref{yr(1)}) with $m = 2$ common factors. 	The first $T$ data $\{r_t\}_{t=1}^T$ is the training set used for model estimation and constructing portfolios, and the remaining data $\{r_t\}_{t=T+1}^{2T}$ is used for out of sample evaluation. 
   We generate two common factors from the following  AR(1) processes :
	$$
	f_{1,t} = 0.01 + 0.6f_{1,t-1} + u_{1,t}
	$$
	$$
	f_{2,t} = 0.01 + 0.95f_{2,t-1} + u_{2,t}
	$$
	where $u_{1,t}\sim N(0,1-0.6^2)$ and $u_{2,t}\sim N(0,1-0.95^2)$ which indicate $\Sigma_F = I_2$, and $f_{1,0}=f_{2,0}=0$. Both two factors are stationary processes but the second factor is more commonly referred to as the near-unit root process\footnote{Simulations based on common factors with autocorrelation coefficients close to 0 are also investigated and exhibit similar results.}.
	The similar setting for autoregressive common factors can be found in \cite{su2017time}, \cite{fan2022}. 
	For factor loading $b_i = \left( b_{i,1}, b_{i,2} \right)^\top$, $i = 1,2,\cdots,p$ and idiosyncratic error $e_{t} = (e_{1t},\ldots,e_{pt})^\top$, we draw them from normal distribution, that $b_{i,j} \sim N(\mu_{b,j},\sigma_{b,j}^2)$, $j = 1, 2$ and $e_{t} \sim N(0,\Sigma_{e})$.

To illustrate the performance of R-MVP, we further define three DGPs by  incorporating homogeneous and heterogeneous outliers, respectively, and both two shocks simultaneously:

\textbf{DGP 2, heterogeneous outliers:} We impose shocks $S$ at fixed frequency $v$ into $e_{t}$, where $S \sim N(\mu_S,\Sigma_{e})$ with $\mu_S = (5\Sigma_{e,11}^{1/2},\ldots,5\Sigma_{e,pp}^{1/2})$  and $v = 50$ where $\Sigma_{e,ii}$ is i-th diagonal element of $\Sigma_e$. These shocks are individual-specific. It can mimic the shocks to individual firms such as product recalls for a specific firm. 
	
\textbf{DGP 3, homogeneous outliers:} We impose fixed-size shocks $H_j$ at fixed frequency $v$ on $u_{j,t}, j = 1,2$, where $H_1 = 5\sqrt{1-0.6^2}, H_2 = 5\sqrt{1-0.95^2}$ and $v = 40$. These shocks are global. It can mimic the shocks to all the firms such as the global economic crisis.
	
	\textbf{DGP 4, homogeneous + heterogeneous outliers:} We impose aforementioned shocks $H_j$ and $S$ for $u_{j,t}$ and $e_{t}$ simultaneously. 
	
	DGPs 2 and 3 introduce shocks into the return data to produce outliers, which are measured (in this section and for the convenience of study) by a threshold of five times corresponding standard deviations. The outliers take place at discrete time points $[v,2v,\ldots]$ and we set different values of $v$ for homogeneous and heterogeneous shocks such that outliers of  $u_{j,t}$ and $e_{it}$ occur at the different dates in DGP 4. 
    We note that the homogeneous and heterogeneous shocks imposed in DGPs 2-4 do not affect population covariance matrix of return data $\{r_t\}$.  Under DGP 2 where $r_t = BF_t + e_t + S\mathbb{I}(t=kv)$, we have
$$
\Sigma_r = BB^\top + \Sigma_e + cov(S\mathbb{I}(t=kv))
$$
where the interactive term is 0, and $cov(S\mathbb{I}(t=kv)) = 0$ since integration of finite point is 0. For DGP 3, constant shocks are required to keep  the variances of common factors unchanged. Specifically, under DGP 3, $f_{j,t} = \mu_{f,j} + \alpha_j f_{j,t-1}+u_{j,t} + H_j\mathbb{I}(t = kv), 1 \leq kv \leq T, j = 1,2$. 
Assume there are $K$ shocks before current time, such that $Kv \leq t$. then for calculating the variance of $f_{j,t}$, we first make following transformation
\begin{align*}
f_{j,t}& = \sum_{i=1}^{K}\alpha_j^{t-iv}H_j + (\mu_{f,j} + u_{j,t}) + \alpha_j\left(\mu_{f,j} + u_{j,t-1} \right) + \alpha_j^2\left(\mu_{f,j} + u_{j,t-2} \right)+\cdots
\\ &
= \sum_{i=1}^{K}\alpha_j^{t-iv}H_j  + \frac{\mu_{f,j}}{1-\alpha_j} + u_{j,t} + \alpha_ju_{j,t-1} + \alpha_j^2u_{j,t-2} + \cdots
\end{align*}
since $\{u_{j,t}\}$ is i.i.d. random variable with mean 0 and finite variance, we then have $E(f_{j,t}) = \sum_{i=1}^{K}\alpha_j^{t-iv}H_j  + {\mu_{f,j}}/{(1-\alpha_j)}$, hence the variance of $f_{j,t}$ is 
    $$
E(f_{j,t}-E(f_{j,t}))^2 = E\left( u_{j,t} + \alpha_ju_{j,t-1} + \alpha_j^2u_{j,t-2} + \cdots \right)
$$ 
which is fixed and the same as that in DGP 1. Hence, $\Sigma_F$ and $\Sigma_e$ are the same as the baseline model in DGPs 2, 3 and 4, and true population covariance matrix $\Sigma_{r} = BB^\top + \Sigma_e$ is fixed.

For values of parameters in the DGPs, we calibrate them from real data. In details, we select largest 100 stocks measured by market values in S\&P 500 index at January 2006, then we apply POET (with thresholding parameter 0.5) to daily excess  return data from Jan 2006 to Dec 2009 which contains 1008 observations  to obtain estimation of factor loading $\hat{B}$ and sparse residual covariance matrix $\hat{\Sigma}_{e}$. Then, we set $\mu_{b,1}  = 0.018, \mu_{b,2} = -0.001, \sigma_{b,1}= 0.0072, \sigma_{b,2}= 0.0084$ according to the mean and variance of $\hat{B}$.  For $\Sigma_{e}$, we set it to be the residual covariance matrix estimator $\hat{\Sigma}_e$ using $p$ largest stocks.	 The results of the following four combinations of sample size and portfolio dimension will be reported: $\{p = 50, T  = 100\}$, $\{p = 50, T =150\}$, $\{p = 80, T  = 100\}$, $\{p = 80, T =150\}$. Each portfolio is held for time $T$ (from $T+1$ to $2T$).

	\subsection{Performance Measures}
	To assess the performance of the proposed estimator, we apply the following out of sample  evaluation statistics for portfolios and covariance matrices:
	
	(1) Out of sample risk: the standard deviation of out of sample portfolio excess returns $\{r_{t}^p\}_{t=1}^{T}$ where $r_{t}^p = \hat{W}^\prime r_{t+T}, t \in [1,T]$ and $\hat{W}$  is the vector of estimated portfolio weights, specifically, 
	\begin{equation}
	\label{eq risk error}
	\text{SD} = \sqrt{\frac{1}{T-1}\sum_{t=1}^{T}\left(r_t^p - \hat{\mu}^p  \right) }, 
	\  \hat{\mu}^p = \frac{1}{T}\sum_{t=1}^{T} r_t^p 
	\end{equation}
	
	(2) Out of sample Sharpe ratio error: the absolute difference between out of sample Sharpe ratio of estimated and infeasible oracle portfolio, where out of sample Sharpe ratio is given by  $SR = \hat{\mu}^p/SD$.
	
	(3) Maximum drawdown (MDD): the largest difference among cumulative returns throughout the entire sample,
	$$
	MDD = \mathop{max}_{1\leq t_1 \leq t_2 \leq T}(\gamma_{t_1}-\gamma_{t_2})
	$$
	where $\gamma_{t_1}$  and $\gamma_{t_2}$ refer to the cumulative portfolio return from the beginning of investment to $t_1$ and $t_2$ respectively, and is defined as $\gamma_{t} = \sum_{i=1}^{t}r_i^p$ for $t= t_1, t_2$.
	
	(4) Weight error:  the  $\ell_2$ norm of the difference between the estimated and oracle weights.
	\begin{equation}
	\label{eq weight error}
	\text{Weight  error} = \mathop{\Vert\hat{W}-W\Vert}\nolimits_{{2}}=\sqrt{\sum\limits_{{i}}\mathop{\mathop{(\mathop{\mathop{\hat{W}-W}})}\nolimits_{{i}}}\nolimits^{{2}}} 
	\end{equation}
	where $W$ is true minimum variance weights calculated by using true $\Sigma_{r}$.

	(5) Covariance matrix error: the relative error matrix measured by Frobenius norm of estimated covariance matrix of assets returns. This measure is also used in \cite{fan2013large,fan2019robust,WANG202153}. 
	\begin{equation}
	\label{eq loss measure}
	\text{Covariance error} = \mathop{\Vert\Sigma_{r}^{-1/2}\hat{\Sigma}_{r}\Sigma_{r}^{-1/2} -  I_p \Vert}\nolimits_{{F}} =  \mathop{\Vert\hat{\Sigma}_{r} - \Sigma_{r} \Vert}\nolimits_{{\Sigma}}
	\end{equation}
	
	All the above statistics are reported based on the simulation replications of 200 times. 
	
	\subsection{Comparison Portfolios}
	For other comparison strategies, we consider the following  portfolios: 
	(1) Linear: portfolio using the linear shrinkage estimator proposed by \cite{ledoit2003improved}. (2) Nonlinear: portfolio when the nonlinear shrinkage estimator proposed by \cite{ledoit2017nonlinear} is used. 
	(3)  POET: portfolio whose covariance matrix is estimated using \cite{fan2013large}. The comparison with POET is informative to assess the robustness of our proposed R-MVP.
	The thresholding parameter ($c_\tau$ in our paper) for R-MVP and POET is set  to 0.5 following \cite{fan2013large} and we choose the soft threshold shrinkage function.  Finally, for  threshold parameter $\tau$ of Huber loss function in our proposal, we  set it as the 0.9-th quantile of $\{||r_t  - {B}^{(i-1)}{F}_{t}^{(i-1)}||, t=1,\ldots,T\}$ at each step i during iterative calculation, which implies that about 90\% of data is weighted by 0.5.
	\subsection{Results}\label{sim:res}
	
	\begin{table}[htbp!]
		\centering
		\caption{Simulation results under DGP 1-4 with $p= 50, T= 100$ and $p =50, T =150$.}
		\label{table DGP1 results1 p = 50}
		\begin{threeparttable}
			{\begin{tabular}{cccccccccccc}
					\noalign{\global\arrayrulewidth1pt}\hline \noalign{\global\arrayrulewidth0.4pt}
					& \multicolumn{5}{c}{p=50, T = 100}                                       & \multicolumn{1}{c}{}     & \multicolumn{5}{c}{p=50, T = 150}                                             \\ \cline{2-6} \cline{8-12} 
					& Risk & MDD    & SR     & Weight & \multicolumn{2}{l}{COV} & Risk & MDD    & SR    & Weight & COV \\ \hline
					& \multicolumn{11}{c}{DGP1, Benchmark}                                                                               \\
					Oracle     & 0.637   & 8.215   & 0       & 0       & 0                         &      & 0.639   & 10.270   & 0      & 0       & 0          \\
Linear     & 0.782   & 11.348  & 8.228   & 1.833   & 499.2                     &      & 0.740   & 13.863   & 5.061  & 1.603   & 410.5      \\
Non-linear & 0.756   & 12.147  & 9.569   & 1.569   & 548.5                     &      & 0.722   & 14.523   & 5.501  & 1.416   & 443.3      \\
POET       & 0.704   & 11.207  & 8.716   & 1.167   & 297.7                     &      & 0.688   & 13.601   & 5.068  & 1.037   & 262.5      \\
R-MVP      & 0.698   & 10.446  & 7.184   & 1.145   & 291.2                     &      & 0.686   & 13.114   & 4.476  & 1.022   & 259.6 \\
					   \hline
					& \multicolumn{11}{c}{DGP2, 5   times standard deviation heterogeneous shocks}                                       \\
					Oracle     & 0.637   & 8.215   & 0       & 0       & 0                         &      & 0.639   & 10.270   & 0      & 0       & 0          \\
Linear     & 0.987   & 16.024  & 10.525  & 2.789   & 868.1                     &      & 0.963   & 18.734   & 7.065  & 2.805   & 791.4      \\
Non-linear & 0.955   & 16.284  & 11.540  & 2.516   & 837.5                     &      & 0.934   & 19.092   & 7.410  & 2.546   & 771.5      \\
POET       & 0.885   & 15.530  & 12.926  & 2.344   & 594.8                     &      & 0.864   & 18.294   & 8.810  & 2.244   & 547.6      \\
R-MVP      & 0.793   & 12.832  & 10.035  & 1.826   & 443.4                     &      & 0.777   & 15.261   & 6.842  & 1.706   & 408.6  \\ \hline
					& \multicolumn{11}{c}{DGP3, 5 times standard deviation  homogeneous shocks}                                          \\
					Oracle     & 0.637   & 8.376   & 0       & 0       & 0                         &      & 0.639   & 10.378   & 0      & 0       & 0          \\
Linear     & 0.766   & 11.821  & 8.069   & 1.733   & 523.6                     &      & 0.733   & 13.965   & 4.993  & 1.544   & 427.4      \\
Non-linear & 0.744   & 12.882  & 9.638   & 1.551   & 557.3                     &      & 0.716   & 14.712   & 5.682  & 1.398   & 451.4      \\
POET       & 0.694   & 11.899  & 8.980   & 1.124   & 310.5                     &      & 0.682   & 13.703   & 5.241  & 1.000   & 274.5      \\
R-MVP      & 0.688   & 10.937  & 7.235   & 1.102   & 302.4                     &      & 0.680   & 13.124   & 4.296  & 0.984   & 269.1   \\ \hline
					& \multicolumn{11}{c}{DGP4, 5   times standard deviation homogeneous shocks and heterogeneous shocks} \\
					Oracle     & 0.637   & 8.376   & 0       & 0       & 0                         &      & 0.639   & 10.378   & 0      & 0       & 0          \\
Linear     & 0.982   & 17.212  & 11.276  & 2.836   & 876.5                     &      & 0.961   & 18.604   & 6.652  & 2.848   & 798.8      \\
Non-linear & 0.940   & 17.752  & 12.398  & 2.476   & 844.6                     &      & 0.922   & 18.910   & 7.221  & 2.509   & 776.4      \\
POET       & 0.847   & 16.417  & 13.136  & 2.105   & 549.8                     &      & 0.828   & 16.515   & 7.845  & 1.997   & 502.3      \\
R-MVP      & 0.761   & 13.142  & 9.935   & 1.626   & 427.8                     &      & 0.751   & 14.414   & 6.003  & 1.539   & 399.7   \\ \noalign{\global\arrayrulewidth1pt}\hline \noalign{\global\arrayrulewidth0.4pt}
			\end{tabular}}
			\begin{tablenotes}
				\item[a] The reported statistics are out of sample risk, maximum drawdown, out of sample Sharpe ratio error, 
				weight error, and covariance error.  
                \item[b] The values are all enlarged by 100 for ease of presentation. 
				\item[c] The number of factors in POET and R-MVP is set to 2, and thresholding parameter $c_\tau$ for R-MVP and POET are set to 0.5. 
			\end{tablenotes}
		\end{threeparttable}
	\end{table}
	
	\begin{table}[htbp!]
		\centering
		\caption{Simulation results under DGP 1-4 with $p= 80, T= 100$ and $p =80, T =150$.}
		\label{table DGP1 results1 p = 80}
		\begin{threeparttable}
			{\begin{tabular}{cccccccccccc}
					\noalign{\global\arrayrulewidth1pt}\hline \noalign{\global\arrayrulewidth0.4pt}
					& \multicolumn{5}{c}{p=80, T = 100}                                       & \multicolumn{1}{c}{}     & \multicolumn{5}{c}{p=80, T = 150}                                             \\ \cline{2-6} \cline{8-12} 
					& Risk & MDD    & SR     & Weight & \multicolumn{2}{l}{COV} & Risk & MDD    & SR    & Weight & COV \\ \hline
					& \multicolumn{11}{c}{DGP1, Benchmark}                                                                \\
					Oracle     & 0.492     & 5.683     & 0       & 0        & 0        &    & 0.487   & 7.317    & 0       & 0        & 0        \\
Linear     & 0.671     & 8.806     & 8.109    & 1.902    & 778.6    &    & 0.618   & 10.236   & 5.439   & 1.666    & 645.0    \\
Non-linear & 0.623     & 9.330     & 9.254    & 1.608    & 839.8    &    & 0.578   & 10.511   & 5.401   & 1.428    & 699.7    \\
POET       & 0.557     & 8.250     & 8.739    & 1.008    & 416.1    &    & 0.531   & 9.655    & 5.344   & 0.836    & 359.0    \\
R-MVP      & 0.550     & 7.714     & 7.016    & 0.986    & 400.6    &    & 0.528   & 9.016    & 4.556   & 0.822    & 350.5  \\ \hline
					& \multicolumn{11}{c}{DGP2, 5   times standard deviation heterogeneous shocks}                        \\
					Oracle     & 0.492     & 5.683     & 0        & 0        & 0        &    & 0.487   & 7.317    & 0       & 0       & 0        \\
Linear     & 0.903     & 12.671    & 10.749   & 2.742    & 1323.3   &    & 0.884   & 14.298   & 6.969   & 2.773    & 1205.1   \\
Non-linear & 0.844     & 11.943    & 10.559   & 2.297    & 1304.3   &    & 0.826   & 13.812   & 7.004   & 2.299    & 1199.8   \\
POET       & 0.777     & 11.232    & 11.027   & 2.315    & 989.2    &    & 0.751   & 12.283   & 7.757   & 2.229    & 883.2    \\
R-MVP      & 0.665     & 9.281     & 9.285    & 1.849    & 694.7    &    & 0.644   & 10.201   & 6.455   & 1.733    & 616.1 \\ \hline
					& \multicolumn{11}{c}{DGP3, 5 times standard deviation  homogeneous shocks}                           \\
					Oracle     & 0.492     & 5.715     & 0        & 0        & 0        &    & 0.488   & 7.338    & 0      & 0        & 0        \\
Linear     & 0.650     & 8.891     & 8.280    & 1.798    & 803.2    &    & 0.607   & 9.939    & 5.351   & 1.598    & 661.6    \\
Non-linear & 0.605     & 9.601     & 9.334    & 1.563    & 849.5    &    & 0.570   & 10.360   & 5.680   & 1.406    & 706.4    \\
POET       & 0.545     & 8.358     & 8.175    & 0.955    & 422.6    &    & 0.523   & 9.334    & 4.962   & 0.794    & 366.5    \\
R-MVP      & 0.537     & 7.752     & 6.641    & 0.937    & 408.1    &    & 0.521   & 8.864    & 4.155   & 0.785    & 356.9   \\ \hline
					& \multicolumn{11}{c}{DGP4, 5   times standard deviation homogeneous shocks and heterogeneous shocks} \\
					Oracle     & 0.492     & 5.715     & 0       & 0        & 0        &    & 0.488   & 7.338    & 0       &0        &0        \\
Linear     & 0.909     & 13.440    & 11.220   & 2.828    & 1341.8   &    & 0.891   & 14.190   & 6.884   & 2.854    & 1218.7   \\
Non-linear & 0.831     & 12.701    & 11.106   & 2.274    & 1310.3   &    & 0.816   & 13.626   & 7.055   & 2.279    & 1205.5   \\
POET       & 0.745     & 11.186    & 10.901   & 2.229    & 920.9    &    & 0.726   & 11.352   & 7.763   & 2.134    & 824.2    \\
R-MVP      & 0.628     & 9.169     & 9.301    & 1.689    & 647.0    &    & 0.612   & 9.753    & 5.974   & 1.592    & 582.2    \\ \noalign{\global\arrayrulewidth1pt}\hline \noalign{\global\arrayrulewidth0.4pt}
			\end{tabular}}
			\begin{tablenotes}
				\item[a] see notes in Table \ref{table DGP1 results1 p = 50}
			\end{tablenotes}
		\end{threeparttable}
	\end{table}

	Tables \ref{table DGP1 results1 p = 50}-\ref{table DGP1 results1 p = 80}  tabulate the simulation results under DGPs 1-4 with $p=50$ and $p = 80$, respectively. The reported values are all enlarged by 100 for ease of presentation.  	The main findings  can be summarized as follows: (1) In the baseline setting,  R-MVP is the best non-oracle portfolio in terms of all considered measures. The POET  enjoys the second lowest values (close to those of R-MVP) of out of sample risk, weight error and relative covariance error. 
 Non-linear shrinkage based portfolio has lower out of sample risk than linear shrinkage based portfolio, but it loses to linear based portfolio in terms of maximum drawdown and Sharpe ratio error.
(2) In the presence of heterogeneous outliers,  the superiority of R-MVP over other portfolios become more visible. E.g., in the case of $p  = 80, T =100$, the gaps of R-MVP over the second-best portfolio in terms of  out of sample risk, maximum drawdown, Sharpe ratio error, weight error and covariance error are {1.12e-3, 0.0195, 0.0127, 4.48e-3 and 2.95}, respectively, which are significantly higher than those of {7e-5, 0.0054, 0.0109, 2.2e-4 and 0.155} under baseline.
(3) The findings in presence of homogeneous outliers are interesting. Firstly, it can be observed that the risk and weight error decrease. Secondly, the changes of  maximum drawdown and Sharpe ratio error are uncertain, which depends on two forces : (i) With more accurate weight estimator, the portfolio's maximum drawdown and Sharpe ratio should be closer to the oracle. Thus, the MDD and Sharpe ratio error values will decrease. (ii) With positive shocks on common factors, the levels of factors rise, and it leads to higher oracle maximum drawdown  and Sharpe ratio (hence the magnitude of its error). Even though the weights remain unchanged from those in the baseline setting, the maximum drawdown and Sharpe ratio of all portfolios increase, {e.g., the oracle maximum drawdown increases to 0.08376 from 0.08215 in the case of $p=50$ and $T = 100$}.  
We note that R-MVP is still the best non-oracle portfolio under DGP 3, compared to other competitors. 
(4) In the presence of both homogeneous outliers and heterogeneous outliers, the advantages of R-MVP are quite visible.

	\begin{figure*}[htbp]
		\centering
		\includegraphics[width=16cm]{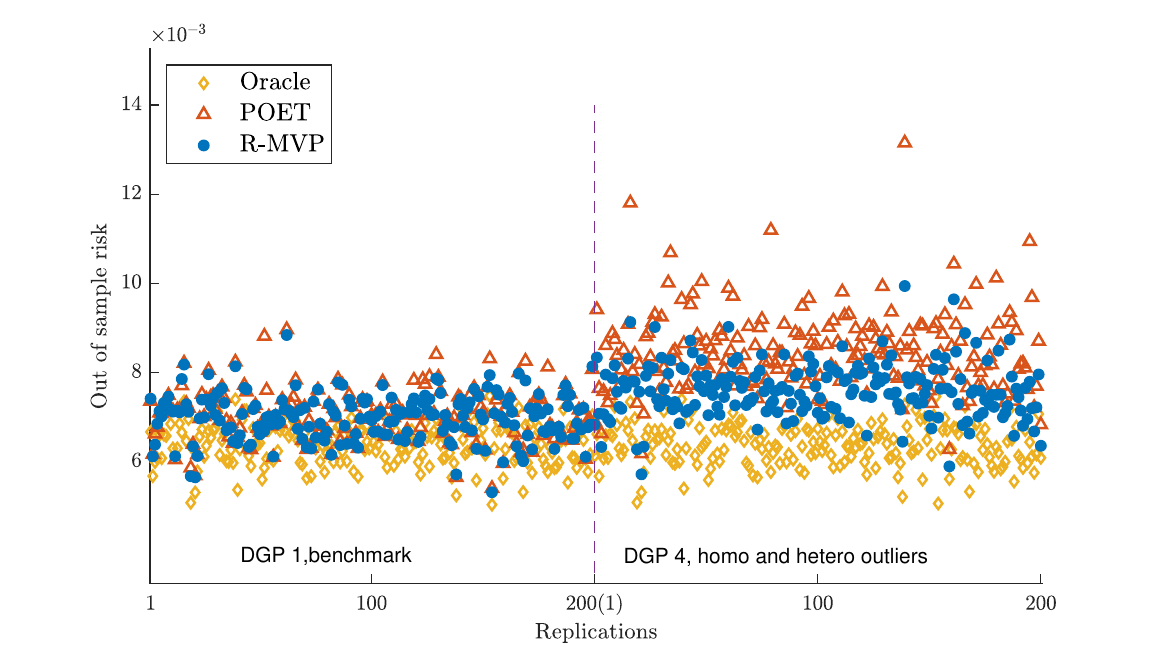}
		\caption{The out of sample risks of POET-based portfolio, R-MVP and corresponding oracle results over all replications under DGP 1 and 4 with $p = 50, T= 100$. }
		\label{Figure std_plot}
	\end{figure*}
	Figure \ref{Figure std_plot} plots the out of sample risks of POET-based portfolio, R-MVP and corresponding oracle results over 200 simulation replications under DGP 1 and 4 with $p = 50, T= 100$ (the results under other [p,T] combinations and near unit root situations are similar). 
	To make it an easy and direct comparison, the results for each replication in DGP 1 and DGP 4 are put in the same figure. And they are shown in the left and right halves of the x-axis, respectively. 
	Firstly, despite the risks of POET-based strategy are  slight higher than those of R-MVP under baseline DGP 1, the gaps are quite small and their risks coincide in many replications. Secondly, the risks of R-MVP are close to the oracle levels, the yellow diamonds and blue dots in the figure can not be separated simply. Thirdly, when the data includes large shocks, R-MVP has better risk control ability. It can be observed that the risk of POET-based portfolio under all simulation replications is much higher than that of R-MVP.
	
	\begin{table}[htbp!]
		\centering
		\caption{Simulation results under DGP 5 and 6 with $p= 50, T= 100, 150$ and $p =80, T =100, 150$.}
		\label{table DGP5-6 results p = 50,80}
		\begin{threeparttable}
			{\begin{tabular}{cccccccccccc}
					\noalign{\global\arrayrulewidth1pt}\hline \noalign{\global\arrayrulewidth0.4pt}
					& \multicolumn{5}{c}{p=50, T = 100}                                       & \multicolumn{1}{c}{}     & \multicolumn{5}{c}{p=50, T = 150}                                             \\ \cline{2-6} \cline{8-12} 
					& Risk & MDD    & SR     & Weight & \multicolumn{2}{l}{COV} & Risk & MDD    & SR    & Weight & COV \\ \hline
					& \multicolumn{11}{c}{DGP 5, 3 times   standard deviation   heterogeneous   shocks}                          \\
					Oracle     & 0.637  & 8.215  & 0      & 0      & 0                        &     & 0.639  & 0       & 0     & 0      & NaN        \\
Linear     & 0.902  & 13.695 & 9.383  & 2.492  & 597.6                    &     & 0.862  & 16.240  & 6.211 & 2.371  & 509.7      \\
Non-linear & 0.830  & 14.053 & 10.722 & 1.945  & 617.6                    &     & 0.802  & 16.561  & 6.711 & 1.877  & 520.3      \\
POET       & 0.754  & 12.769 & 11.021 & 1.550  & 359.5                    &     & 0.735  & 15.046  & 7.005 & 1.415  & 323.9      \\
R-MVP      & 0.728  & 11.307 & 8.331  & 1.406  & 336.1                    &     & 0.714  & 13.861  & 5.571 & 1.282  & 306.3  \\ \hline
					& \multicolumn{11}{c}{DGP 6, 2 times covariance matrix of shocks}                                            \\
					Oracle     & 0.637  & 8.215  & 0      & 0      & 0                        &     & 0.639  & 10.270  & 0     & 0      & 0          \\
Linear     & 0.971  & 15.705 & 10.252 & 2.702  & 912.8                    &     & 0.951  & 18.560  & 6.923 & 2.738  & 827.8      \\
Non-linear & 0.947  & 15.985 & 11.332 & 2.485  & 877.4                    &     & 0.927  & 18.926  & 7.206 & 2.522  & 805.5      \\
POET       & 0.882  & 15.426 & 12.667 & 2.327  & 636.0                    &     & 0.861  & 18.360  & 8.764 & 2.229  & 574.9      \\
R-MVP      & 0.784  & 12.637 & 9.743  & 1.776  & 450.9                    &     & 0.769  & 15.148  & 6.741 & 1.657  & 415.7  \\ \hline
					& \multicolumn{4}{c}{p=80, T = 100} & \multicolumn{1}{c}{} &  & \multicolumn{4}{c}{p=80, T = 150} &        \\ \cline{2-6} \cline{8-12}
					& \multicolumn{11}{c}{DGP 5, 3 times standard deviation   heterogeneous shocks}                              \\
					Oracle     & 0.492  & 5.683  & 0      & 0      & 0                        &     & 0.487  & 7.317   & 0     & 0      & 0          \\
Linear     & 0.824  & 11.494 & 10.131 & 2.606  & 927.9                    &     & 0.786  & 12.636  & 6.507 & 2.527  & 789.5      \\
Non-linear & 0.712  & 10.392 & 9.950  & 1.904  & 949.9                    &     & 0.682  & 11.839  & 6.663 & 1.845  & 811.3      \\
POET       & 0.619  & 9.065  & 10.281 & 1.518  & 533.6                    &     & 0.589  & 10.149  & 6.788 & 1.355  & 462.5      \\
R-MVP      & 0.585  & 8.256  & 8.440  & 1.360  & 479.3                    &     & 0.562  & 9.331   & 5.548 & 1.222  & 425.8 \\ \hline
					& \multicolumn{11}{c}{DGP 6, 2 times covariance matrix of shocks}                                            \\
					Oracle     & 0.492  & 5.683  & 0      & 0      & 0                        &     & 0.487  & 7.317   & 0     & 0      & 0          \\
Linear     & 0.874  & 12.343 & 10.528 & 2.637  & 1392.2                   &     & 0.863  & 14.042  & 6.946 & 2.701  & 1256.6     \\
Non-linear & 0.827  & 11.825 & 10.513 & 2.266  & 1367.1                   &     & 0.812  & 13.651  & 6.998 & 2.290  & 1246.8     \\
POET       & 0.767  & 11.161 & 11.179 & 2.269  & 1060.5                   &     & 0.743  & 12.333  & 7.725 & 2.205  & 926.6      \\
R-MVP      & 0.653  & 9.094  & 9.228  & 1.797  & 707.4                    &     & 0.631  & 10.016  & 6.366 & 1.686  & 622.1    \\ \noalign{\global\arrayrulewidth1pt}\hline \noalign{\global\arrayrulewidth0.4pt}
			\end{tabular}}
			\begin{tablenotes}
				\item[a] see notes in Table \ref{table DGP1 results1 p = 50}
			\end{tablenotes}
		\end{threeparttable}
	\end{table}

	Since portfolio performance is shown to be affected mainly by individual outliers, we try  two other DGP settings as follows: (1) $S_2 \sim N({\mu}_{S,2},\Sigma_{e})$ with fix frequency $v=50$ where ${\mu}_{S,2} = (3\Sigma_{e,11}^{1/2},\ldots, 3\Sigma_{e,pp}^{1/2})$; (2) $S_3 \sim N({\mu}_S,2\Sigma_{e})$ with fixed frequency $v= 50$.
	We denote these two data generating processes as  \textbf{DGP 5}, and \textbf{6}.
	DGP 5 imposes weaker shocks compared to DGP 2. On the other hand, DGP 6 considers shocks with a larger covariance matrix.   
	The results of four combinations of $p$ and $T$  are reported in Table \ref{table DGP5-6 results p = 50,80}. Some conclusions are further drawn: (1) Evidently, we can still observe that R-MVP manifests its superiority under new DGPs.  
	(2) In more details, when shocks become larger (with larger variance), R-MVP gets even larger advantages compared to the other portfolios. E.g., in the case of DGP 2 with $p = 80, T =150$, the gaps of R-MVP from POET are about 1.07e-3, 0.0208 and 0.0130, 4.96e-03 and 2.67 in terms of  risk, maximum drawdown, Sharp ratio error, weight error and covariance error, respectively,  
 and in DGP 6 that gaps increase to 1.12e-3, 0.0232, 0.0136, 5.19e-03, 3.05.

	\begin{figure*}[htbp]
		\centering
		\includegraphics[width=\textwidth]{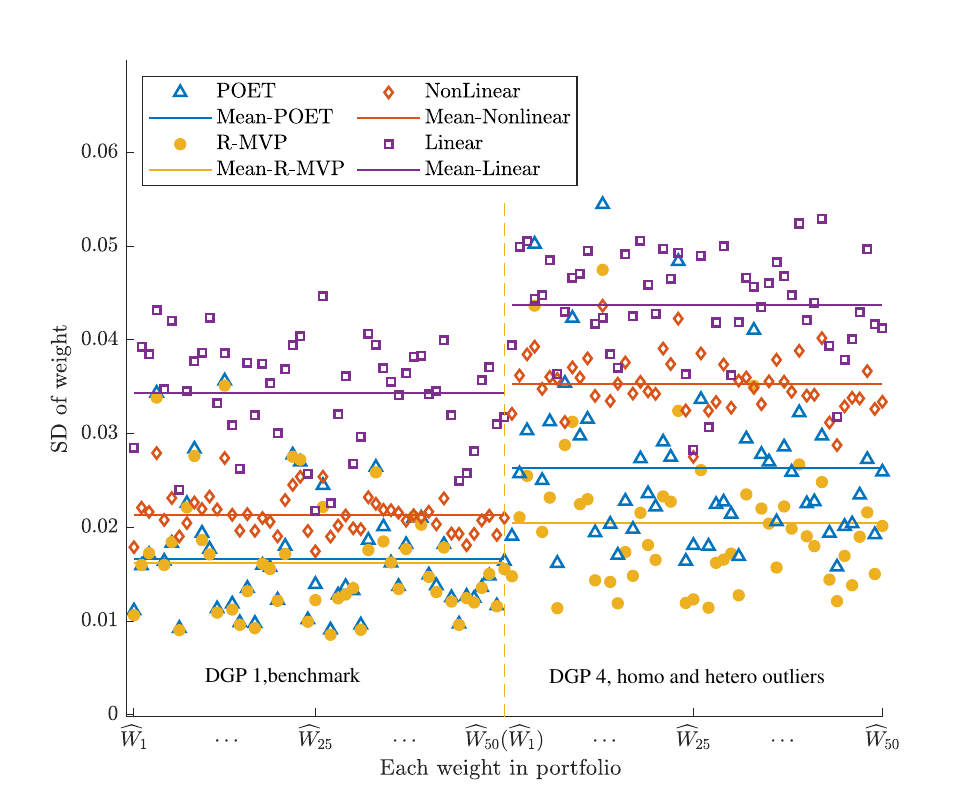}
		\caption{The standard deviation of each element of $p$-dimensional portfolio weights
			under DGP 1 and 4 with $p = 50, T= 100$. Left half of panel : results over 200 replications under DGP 1. Right half of panel : results over 200 replications under DGP 2.}
		\label{Figure std_weight}
	\end{figure*}
	Furthermore, we analyze the stability of portfolio weights of considered portfolios. The simulation data are randomly generated from the same population distribution in each setting. A robust portfolio should have small variances for the weights of each asset. Figure \ref{Figure std_weight} shows empirical standard deviation of estimated weights by different portfolios under DGP 1 (left panel) and DGP 4 (right panel) with $p=50, T = 100$. The results under other combinations of $[p, T]$ are similar and are not reported. In each replicated experiment, the sizes of imposed shocks follow a random distribution and are different. Therefore, the higher standard deviation of portfolio weights means the corresponding strategy is more  sensitive to the outliers of the return data. The oracle portfolio weight (not plotted) has zero standard deviation since the true covariance matrix maintains unchanged over replications. In DGP 1, the weights' standard deviations of POET-based portfolio and R-MVP are lowest (the average lines nearly coincide but R-MVP still show a slight lower average of all weights' standard deviations). The standard deviation of weights of nonlinear shrinkage based portfolio weight is slightly higher than R-MVP. The linear shrinkage based portfolio  has large variations in the weights. Also, it can be observed that the variation patterns of POET-based portfolio and R-MVP are similar. It is worth noting that standard deviations of each weight in nonlinear shrinkage based portfolio are close to each other and have the lowest variation across the weights. In the presence of both homogeneous and heterogeneous outliers (DGP 4, right half of the figure), it is clear that R-MVP is the most robust to outliers. The weights' standard deviations of R-MVP are still close to the levels under DGP 1. However, weight standard deviations of other comparison strategies increase sharply. E.g., although POET-based portfolio  still enjoys the second lowest weight standard deviation, its values are nearly 1.3 times those of R-MVP and 1.6 times the levels of DGP 1,  the variation of weights of nonlinear shrinkage based portfolio increases sharply and is also about 1.7 times the levels of DGP 1, but it is still better than linear shrinkage based portfolio. 
	\begin{figure}[htbp]
		\centering
		\includegraphics[width=14 cm]{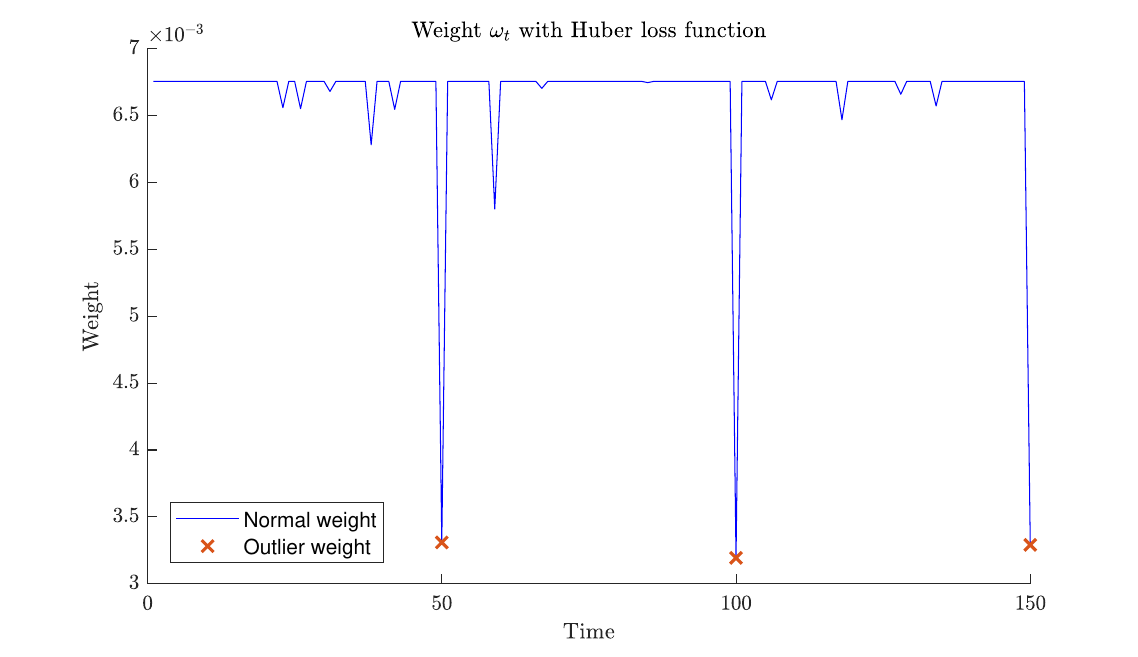}
		\caption{ Weighting sequence $\omega_t$ calculated by Eq. (\ref{weights omegat}) with bi-square loss function and $p = 50, T= 150$. Outlier weight: $\omega_t$ at time points hit by shocks, e.g., T= 50, 100, 150. Normal weight: $\omega_t$ other than the outlier weights.}
		\label{Figure weight}
	\end{figure}
	
	To get a better insight into the superiority of R-MVP in the presence of  outliers,  we plot weighting sequence $\omega_t$ calculated by (\ref{weights omegat})  under DGP 2 with $p = 50, T= 150$ in Figure \ref{Figure weight}. For easy comparison, we transform the weights into percentage form, and note that $\omega_t = 6.67\times 10^{-3}$ of POET which treats all time points equally. It is clear that R-MVP assigns much low weights  to time points where shocks occur, hence the influences of outliers are reduced largely, and as a consequence, R-MVP achieves the desired robust performance.
	
	\begin{figure}[htbp]
		\centering
		\includegraphics[width=18 cm]{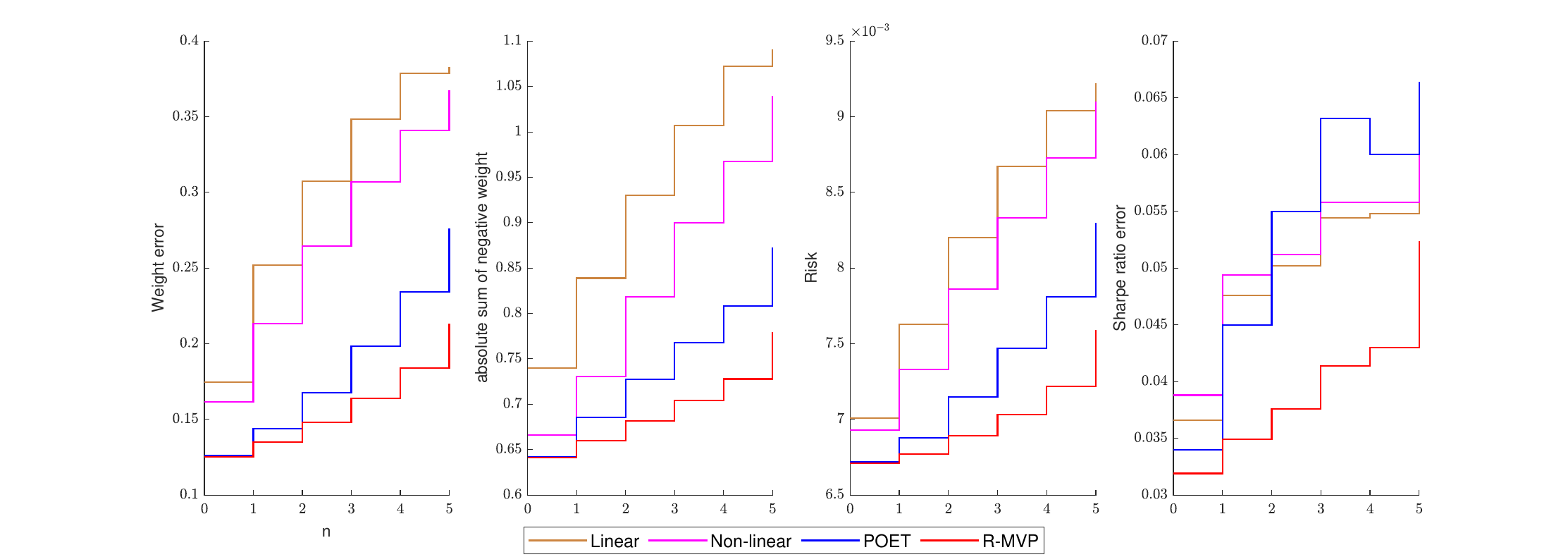}
		\caption{Weight error, negative weight error, out of sample risk and out of sample Sharpe ratio error of various portfolios with respect to the number of outliers in the training set.}
		\label{Figure mechan}
\end{figure}

At last, we show the robustness of the proposed R-MVP by observing the marginal effects of both types shocks in the training set. We generate the return data from DGP 4 but impose the homogeneous and heterogeneous shocks step by step. Specifically, let $n$ denote the number of shocks we impose, then homogeneous and heterogeneous shocks are imposed at time points $40\times i$ and $50\times i$, respectively, $ 1 \leq i\leq n$. Obviously when  $n = 0$, the data is the same as if generated from DGP 1. We consider the case with $p =50$ and $T =250$, which indicates the maximum value can be taken by $n$ is 5.
	Then with respect to $n$, we plot curves for following measures: (1) out of sample risk, (2) out of sample Sharpe ratio error, (3) weight error and (4) negative weight error which is defined as $\sum_{i=1}|\hat{w}_i\mathbb{I}(\hat{w}_i<0)|$. The curves are shown in Figure \ref{Figure mechan}. It is evidence that the average slopes of R-MVP are the smallest in terms of four measures, which indicates the  robustness of R-MVP that it is less affected by the large hits. Additionally, R-MVP has significantly lower amount of short positions (measured by summation of  negative weights) than other competitors when data is contaminated by shocks, which means  R-MVP possesses lower transaction cost in practical applications.

	\section{Real Data Study}\label{empirical}
		
	This section demonstrates the out of sample performance of the proposed R-MVP in real data.  We focus on four portfolio performance measures: out of sample Sharpe ratio, out of sample risk, MDD, which are all introduced in the simulation section,  and out of sample cumulative excess return, which is the sum of out of sample realized portfolio returns in the investment period.  
	The competing investment strategies include the ones in simulation studies, plus (1) the benchmark ``1/N'', (2) minimum variance portfolio using the sample covariance matrix, (3) the unified minimum variance problem estimator (MVP-UF) proposed by \cite{ding2021high}, (4) M-portfolio proposed by \cite{demiguel2009or} which also considers robust estimation in the minimum variance problem, (5) minimum variance portfolio based on robust covariance estimation (RCOV) developed by \cite{fan2019robust}, and (6) minimum variance portfolio based on robust covariance estimation under elliptical factor model (eFM-COV) developed by \cite{fan2018aos}. 
	
At the beginning of the investment decision period (denoted as $T_0$), we use the historical excess return data (the training data set) with length $T$ to compute the covariance matrix estimators and the weights of a portfolio. The holding period for all actively managed portfolios is $HT$. To mimic real-world mutual fund, we use the holding time of one week (5 trading days) and one month (21 trading days). During the holding period, we use excess return data $\mathop{r}\nolimits_{{t}}$ to calculate the returns of each portfolio. After each holding period, we rebalance the portfolios using historical data of the same length $T$ in a rolling window manner. We repeat this process until the last period of the testing sample. We first assume the same holding period for all strategies and no transaction fees to make a fair comparison. In the following we also consider transaction cost.
	
In practice, transaction cost is important for investors. Following \cite{demiguel2009optimal} and  \cite{ao2019approaching}, the excess portfolio return net of transaction cost is computed as, 
	\begin{equation}
	r_{t}^{net} = \left(1 - \sum_{i=1}^pc\left|\hat{W}_{t+1,i}-\hat{W}_{t,i}^+\right| \right)(1+ r_{t}) -1
	\label{eq transaction cost} 
	\end{equation}
	where $\hat{W}_{t+1,i}$($\hat{W}_{t,i}^+$) is the $i$-th element of  portfolio weight after(before) rebalancing, and	$r_{t}$ is the excess return of the portfolio without transaction cost.
The parameter $c$ controls the level of the transaction cost. In this study,  $c$ is set to 10 basis points following \cite{ao2019approaching}. Furthermore,  the total portfolio turnover  is defined as  
	\begin{equation}
	TO = \frac{1}{RT}\sum_{l =1 }^{RT}\sum_{i=1}^p|\hat{W}_{l+1,i}-\hat{W}_{l,i}^+|
	\label{eq turnover}
	\end{equation} 
	where $\hat W_{l+1,i}$ is the desired portfolio weight at $(l+1)$-th rebalancing, and $\hat W_{l,i}^+$ is portfolio weight before the $(l+1)$-th rebalancing.	
	\subsection{Data}
	We conduct empirical analysis based on  daily return data of S\&P 500 index and Russell 2000 index component stocks downloaded from CRSP. The stocks selected in S\&P 500 index are large cap stocks, and  Russell 2000 index consists of 2000 smallest stocks in Russell 3000 index and is often used to represent small-cap company stocks. We consider both indices' constituent stocks to represent a good mixture of both large- and small-cap stocks. Specifically, we consider the following investing pools (henceforth referred to as scenarios (1) and (2), respectively): (1) largest 200 stocks in S\&P 500 index measured by market value. (2) largest 200 stocks in Russell 2000 index measured by market value.  
	Above two asset pools can represent large-cap-specialized mutual fund, and small-cap-specialized mutual fund, respectively. 
	We consider out of sample period from 01/03/2011 to 12/31/2013 with sample size $T = 400$, which contains a total of three years. The historical data used at most decision nodes include many individual outliers; see the number of heterogeneous outliers during considered periods from the upper panel of Figure \ref{Figure cumulative plot}. For some anecdotal examples, during the studied sample period, JPMorgan Chase suffered a loss of around 2 billion dollars in April and May 2012 due to the ``errors'', and ``bad judgment'' of trading strategy.  
	The share price of Netflix, an American online streaming company, plummeted by more than 30\% on Oct 25, 2011 after it announced the mass loss of users in the previous quarter. There are many other cases of these individual firm-specific shocks. 
	We remove those stocks that do not have entire  history data when doing the first calculation under the rolling window scheme. Finally, we take risk-free return data from the Fama-French data library. And the results based on weekly and monthly holding frequencies are both considered. 
	
We first make a preliminary analysis of considered return data taking scenario (2) where the investment pool is composed of the largest 200 stocks in Russell 2000 index as an example (the results are similar for scenario (1)). Figure \ref{Figure qq plot} provides Q-Q plots of return residual calculated by applying principal component analysis on historical data at the first decision node, against Gaussian distribution and t-distribution with degrees of freedom 3, 4 and 6. It can be observed that the residual series (randomly selected) is well fitted by t-distribution with the degree of freedom 4. Generally speaking,   t-distribution is more suitable than Gaussian distribution for fitting daily return residuals, and the behaviors of most residuals can be depicted by t-distribution with degree of freedom varying from 3 to 6. It indicates that our robust method may be  more appropriate given the feature of heavy-tails in the data.
	Moreover,  we visualize the individual (residual) outliers in Figure \ref{Figure outlier plot}, where the upper panel provides the time plot of the number of outliers based on historical data at 01/2011 and the lower panel further draws allocation of outliers in the form of a sparse matrix that outliers take 1 and others take 0. For the criterion of judging outliers, we simply use 1.96 standard deviations around mean return as the critical boundary of outlier returns. From the upper panel, it seems that the gray background region has more outliers than the blank background region (with 5\% confidence level, if we roughly assume 200 residuals are normally distributed, the total outliers should be close to 10). The lower panel shows that there are plenty of outliers in estimated residuals, especially between 05/2009 and 08/2009. These findings further imply the suitability of our proposal.

	\begin{figure}[h]
		\centering
		\includegraphics[width=16cm]{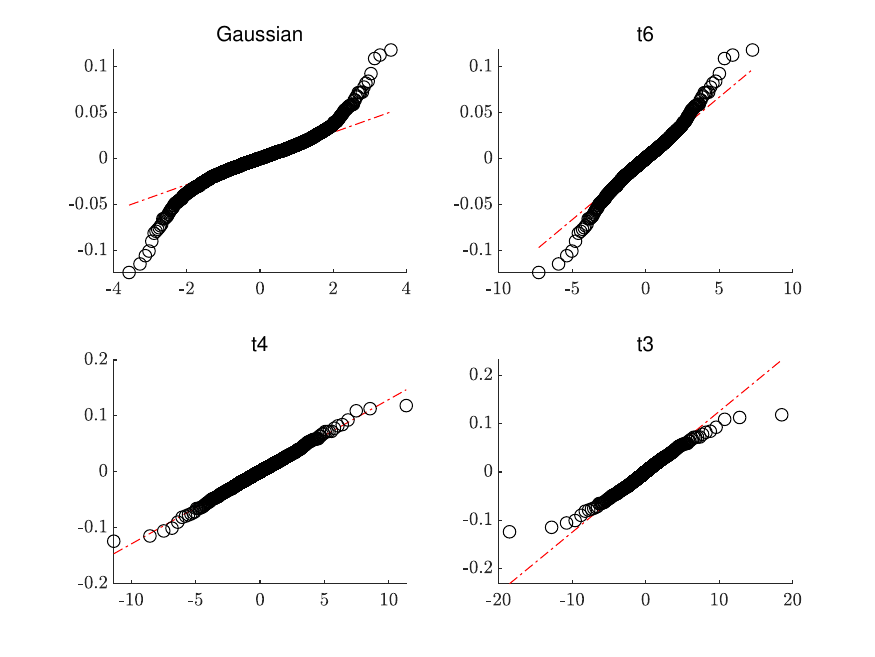}
		\caption{ Q-Q plot of excess return residual of stocks in Russell 2000 index against Gaussian distribution and t-distribution with degree of freedom 3, 4 and 6.}
		\label{Figure qq plot}
	\end{figure}
	
	\begin{figure}[h]
		\centering
		\includegraphics[width=16cm]{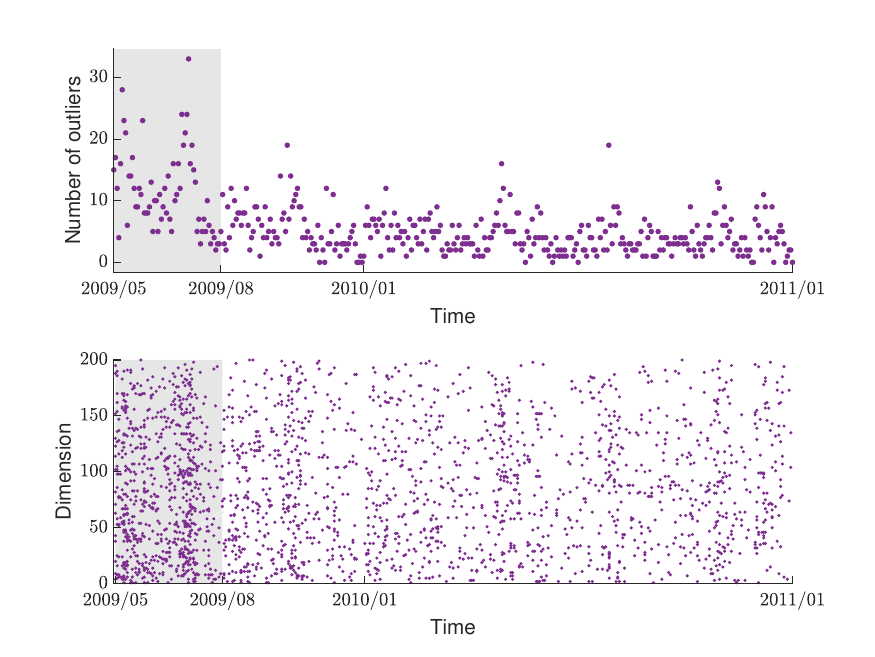}
		\caption{Time plot of the number of outliers (upper panel) of largest 200 Russell 2000 component residuals and corresponding allocated patterns (lower panel).}
		\label{Figure outlier plot}
	\end{figure}

	\subsection{Choice of Thresholding Parameter}
	In practice, the value of $c_\tau$ can be data driven, that we use $K$-fold  cross-validation method to determine it.  With the residual $\hat{e}_t$ from our robust PCA estimation procedure, we randomly divide them into two subsets, denoted as $\{\hat{e}_t\}_{t\in A} $ and $\{\hat{e}_t\}_{t\in B} $ respectively. The sample sizes of each subset are $T(A)$ and $T(B)$, with $T(A) + T(B) = T$. Subset $A$ is used for training and subset $B$ is used for validation. Therefore, we choose the threshold $c_\tau$ by minimizing the following objective function over a compact interval:
	\begin{align*}
	c_\tau^* = \mathop{argmin}_{\underline{c}_{\tau}<\tilde{c}_\tau \leq \overline{c}_\tau} \frac{1}{K}\sum_{j=1}^{K} \left\|\hat{\Sigma}_e^{A,j}(\tilde{c}_\tau) - S_e^{B,j} \right\|^2_F
	\end{align*}
where $\underline{c}_{\tau}$ is the minimum constant that guarantees the positive definiteness of $\hat{\Sigma}_e^{A,j}(\tilde{c}_\tau)$, $\overline{c}_{\tau}$ is large constant such that $\hat{\Sigma}_e^{A,j}(\tilde{c}_\tau)$ is diagonal, $\hat{\Sigma}_e^{A,j}(\tilde{c}_\tau)$ is the threshold residual covariance estimator by using subset $A$ in j-th loop with threshold $\tilde{c}_\tau$, and $S_e^{B,j} $ is sample covariance matrix by using subset $B$ in j-th loop. In this paper,  we  simply set $T(A)= [2T/3]$ in our empirical studies.

	\subsection{Results}
	
	\begin{table}[htbp!]
		\centering
		\caption{Empirical daily application under rolling window scheme, out of sample period is from 01/03/2011 to 12/31/2013, sample size T = 400, monthly holding period. }
		\begin{threeparttable}
			{
				\begin{tabular}{cccccccccccc}
					\noalign{\global\arrayrulewidth1pt}\hline \noalign{\global\arrayrulewidth0.4pt}
					& \multicolumn{5}{c}{max SP 200}   &         & \multicolumn{5}{c}{max Russell 200}      \\ \cline{3-5} \cline{8-11}
					& CR     & Risk    & SR      & MDD   & TR  &  & CR    & Risk    & SR     & MDD   & TR     \\ \hline
					\multicolumn{12}{c}{Without transaction   cost}                                                                                                        \\           
					R-MVP                    & 0.395  & 0.00576 & 0.0934  & 0.080 & -  &   & 0.363 & 0.00693 & 0.0712 & 0.106 & -    \\
					Sample                   & -0.009 & 0.00691 & -0.0018 & 0.225 & -  &   & 0.192 & 0.00864 & 0.0302 & 0.206 & -    \\
					1/N                      & 0.483  & 0.01205 & 0.0545  & 0.256 & -  &   & 0.379 & 0.01615 & 0.0320 & 0.396 & -    \\
					Linear                   & 0.100  & 0.00626 & 0.0217  & 0.141 & -  &   & 0.274 & 0.00780 & 0.0478 & 0.160 & -    \\
					Nonlinear                & 0.257  & 0.00573 & 0.0609  & 0.084 & -  &   & 0.310 & 0.00754 & 0.0560 & 0.145 & -    \\
					POET                & 0.386  & 0.00576 & 0.0911  & 0.083 & -  &   & 0.354 & 0.00694 & 0.0693 & 0.112 & -    \\
					MVP-UF                   & 0.386  & 0.00576 & 0.0911  & 0.083 & -  &   & 0.354 & 0.00694 & 0.0693 & 0.112 & -    \\
					M-portfolio              & 0.034  & 0.00739 & 0.0062  & 0.280 & -  &   & 0.141 & 0.00967 & 0.0199 & 0.244 & -    \\
					RCOV                      & 0.505  & 0.02080 & 0.0330  & 0.355 & -  &   & 0.417 & 0.00762 & 0.0745 & 0.161 & - \\ 
                    eFM-COV &-2.279&	0.08717&	-0.0356&	2.367&	-&		&2.602&	0.21674&	0.0163&	4.616&	-
                     \\\hline
					\multicolumn{12}{c}{With transaction   cost}                                                                                                    \\
					R-MVP                    & 0.355  & 0.00577 & 0.0838  & 0.082 & 0.84 & & 0.348 & 0.00688 & 0.0687 & 0.109 & 0.37 \\
					Sample                   & -0.097 & 0.00694 & -0.0189 & 0.260 & 1.87 & & 0.116 & 0.00863 & 0.0184 & 0.214 & 1.16 \\
					1/N                      & 0.482  & 0.01205 & 0.0544  & 0.256 & 0.01 & & 0.378 & 0.01612 & 0.0319 & 0.396 & 0.02 \\
					Linear                   & 0.043  & 0.00628 & 0.0092  & 0.163 & 1.22 & & 0.220 & 0.00779 & 0.0384 & 0.163 & 0.72 \\
					Nonlinear                & 0.229  & 0.00574 & 0.0543  & 0.089 & 0.59 & & 0.266 & 0.00751 & 0.0481 & 0.147 & 0.60 \\
					POET                & 0.346  & 0.00577 & 0.0816  & 0.085 & 0.85 & & 0.333 & 0.00691 & 0.0655 & 0.115 & 0.41 \\
					MVP-UF                   & 0.346  & 0.00577 & 0.0816  & 0.085 & 0.85 & & 0.333 & 0.00691 & 0.0655 & 0.115 & 0.41 \\
					M-portfolio              & -0.105 & 0.00745 & -0.0191 & 0.340 & 2.97 & & 0.020 & 0.00971 & 0.0029 & 0.296 & 1.97 \\
					RCOV                      & -0.193 & 0.02148 & -0.0122 & 0.527 & 15.12& & 0.383 & 0.00763 & 0.0684 & 0.168 & 0.63
					\\
                        eFM-COV &-4.627&	0.09446&	-0.0666&	4.701&	65.05&	&	0.114&	0.21625&	0.0007&	4.822&	43.09\\
					\noalign{\global\arrayrulewidth1pt}\hline \noalign{\global\arrayrulewidth0.4pt}    
				\end{tabular}	
			}
			\begin{tablenotes}
				\item[a] ``max SP 200'' (``max Russell 200'') means we use largest 200 stocks in S\&P 500 index (Russell 2000 index) measured by market values as asset pool. 
				\item[b] The reported measures CR, Risk, SR, MDD and TR are out of sample cumulative portfolio return, out of sample risk, out of sample Sharpe ratio, maximum drawdown and portfolio turnover ratio, respectively. 
			\end{tablenotes}
		\end{threeparttable}
		\label{Table empirical monthly}
	\end{table}

	\begin{table}[htbp!]
		\centering
		\caption{Empirical daily application under rolling window scheme, out of sample period is from 01/03/2011 to 12/31/2013, sample size T = 400, weekly holding period.}
		\begin{threeparttable}
			{
				\begin{tabular}{cccccccccccc}
					\noalign{\global\arrayrulewidth1pt}\hline \noalign{\global\arrayrulewidth0.4pt}
					& \multicolumn{5}{c}{max SP 200}   &         & \multicolumn{5}{c}{max Russell 200}      \\ \cline{3-5} \cline{8-11}
					& CR     & Risk    & SR      & MDD   & TR  &   & CR    & Risk    & SR     & MDD   & TR     \\ \hline
					\multicolumn{12}{c}{Without transaction   cost}                                                                                                        \\           
					R-MVP                 & 0.400  & 0.00566 & 0.0942  & 0.074 & -  &  & 0.386  & 0.00688 & 0.0747  & 0.094 & -    \\
					Sample                & 0.021  & 0.00682 & 0.0041  & 0.223 & -  &  & 0.185  & 0.00857 & 0.0288  & 0.202 & -    \\
					1/N                   & 0.508  & 0.01196 & 0.0567  & 0.256 & -  &  & 0.409  & 0.01606 & 0.0339  & 0.395 & -    \\
					Linear                & 0.134  & 0.00617 & 0.0290  & 0.140 & -  &  & 0.263  & 0.00774 & 0.0453  & 0.152 & -    \\
					Nonlinear             & 0.267  & 0.00567 & 0.0629  & 0.080 & -  &  & 0.319  & 0.00753 & 0.0564  & 0.146 & -    \\
					POET             & 0.396  & 0.00567 & 0.0932  & 0.075 & -  &  & 0.370  & 0.00694 & 0.0711  & 0.098 & -    \\
					MVP-UF                & 0.396  & 0.00567 & 0.0932  & 0.075 & -  &  & 0.370  & 0.00694 & 0.0711  & 0.098 & -    \\
					M-portfolio           & 0.043  & 0.00738 & 0.0078  & 0.236 & -  &  & 0.226  & 0.00946 & 0.0319  & 0.204 & -    \\
					RCOV                   & 1.299  & 0.02250 & 0.0770  & 0.299 & -  &  & 0.589  & 0.01167 & 0.0674  & 0.156 & -    \\ 
                    eFM-COV &	33.291&	1.15642&	0.0392&	7.220	&-	&	&2.125&	0.24704&	0.0117&	3.415&	- \\ \hline
					\multicolumn{12}{c}{With transaction   cost}                                                                                                    \\
					R-MVP                 & 0.320  & 0.00567 & 0.0753  & 0.078 & 0.40 & & 0.361  & 0.00687 & 0.0701  & 0.096 & 0.16 \\
					Sample                & -0.158 & 0.00687 & -0.0306 & 0.307 & 0.90 & & 0.064  & 0.00855 & 0.0100  & 0.216 & 0.54 \\
					1/N                   & 0.507  & 0.01196 & 0.0565  & 0.256 & 0.01 & & 0.408  & 0.01605 & 0.0345  & 0.395 & 0.01 \\
					Linear                & 0.019  & 0.00620 & 0.0040  & 0.178 & 0.58 & & 0.190  & 0.00773 & 0.0327  & 0.161 & 0.33 \\
					Nonlinear             & 0.205  & 0.00568 & 0.0481  & 0.089 & 0.31 & & 0.242  & 0.00751 & 0.0430  & 0.150 & 0.36 \\
					POET             & 0.316  & 0.00567 & 0.0743  & 0.079 & 0.40 & & 0.341  & 0.00693 & 0.0656  & 0.101 & 0.18 \\
					MVP-UF                & 0.316  & 0.00567 & 0.0743  & 0.079 & 0.40 & & 0.341  & 0.00693 & 0.0656  & 0.101 & 0.18 \\
					M-portfolio           & -0.330 & 0.00749 & -0.0587 & 0.443 & 1.86 & & -0.017 & 0.00945 & -0.0024 & 0.270 & 1.15 \\
					RCOV                   & -1.419 & 0.02678 & -0.0707 & 1.631 & 13.63 & & 0.459  & 0.01196 & 0.0512  & 0.169 & 0.86
					\\
                        eFM-COV	&23.264&	1.29198&	0.0240&	22.238&	112.16&	&	-1.440&	0.24937&	-0.0077&	4.395&	38.60 \\
					\noalign{\global\arrayrulewidth1pt}\hline \noalign{\global\arrayrulewidth0.4pt}    
				\end{tabular}	
			}
			\begin{tablenotes}
				\item[a] see notes in Table \ref{Table empirical monthly}.
			\end{tablenotes}
		\end{threeparttable}
		\label{Table empirical weekly}
	\end{table}
	
	The results with monthly and weekly holding periods are reported in Tables \ref{Table empirical monthly} and \ref{Table empirical weekly}, respectively.  The number of common factors and thresholding parameter values are computed at the first decision node and are kept fixed. For threshold value $\tau$ of Huber loss function, we take 0.9-th quantile of sequence $||r_t - B^{(i-1)}F_t^{(i-1)}||$ at each step of algorithm for R-MVP.  The soft threshold function is applied for  R-MVP and POET with threshold from cross-validation method. For others, threshold in M-portfolio is set to 0.0001 as recommended by \cite{demiguel2009or}, the tuning parameter of MVP-UF is  selected by cross-validation procedure proposed by \cite{ding2021high}, the parameter values of Huber loss threshold and soft  threshold function for RCOV and eFM-COV are specified according to their simulation settings.

From both tables, one can see that R-MVP has robust results both with and without transaction costs during the studied periods. Specifically, R-MVP achieves the highest Sharpe ratio in most cases, and its risks are also very close to the lowest level. For cumulative returns, R-MVP has a considerable advantage over other methods except for the ``1/N'' and RCOV strategies, which both undertake much higher risk than R-MVP. RCOV does not seem to have a stable performance, it is very good when asset pool consists of Russell 200, but very bad for the 200 largest stocks in S\&P 500 index. E.g. its risk is approximately 0.02 which is the twice as much as that of the ``1/N'' strategy, and its turnover ratio is quite substantial at 15.12. Furthermore, R-MVP enjoys the lowest maximum drawdown. The ``1/N'' strategy, on the other hand, has very large MDD. The turnover ratio of R-MVP and POET-based portfolios are close to each other, but R-MVP has overall smaller turnover. E.g., in the cases of 200 largest Russell 2000 firms with monthly rebalancing, the turnover ratio of R-MVP is 0.37, which is lower than that of POET-based portfolio 0.41. The M-portfolio, which is also established on robust method, does not perform well in high-dimensional problems. POET and MVP-UF have good out of sample performances in terms of Sharpe ratio and risks, although their performances are topped by R-MVP.

\subsection{Comparison of R-MVP and POET}
To get direct insight on the robustness performance of the proposed R-MVP,  we provide time plots of the number of individual outliers (upper panel) and the difference between cumulative returns of  R-MVP and POET-based portfolio (lower panel, weekly rebalancing, $p = 200$) in Figure \ref{Figure cumulative plot}. 
We use 144 stocks with complete observations during 05/2009 - 12/2013 from the largest 200 Russell 2000 stocks (we do the same for 200 largest stocks in S\&P 500 index,  which has 196 stocks with complete data) to count the number of heterogeneous outliers, which is shown  by upper panel of Figure \ref{Figure cumulative plot}. 
It shows that the return data during 07/2011-01/2012 (gray shadow region of the upper panel) has many outliers. Furthermore, we highlight the time period of decision nodes where the historical data contains periods 07/2011-01/2012. One can observe that the gap in cumulative return tends to increase in the investment period. A similar phenomenon can be observed from the results with the monthly holding period. The increment part can be regarded as the contributions of robust PCA. Moreover, it seems that the increasing performance gap of small-cap stock constituted portfolio is larger than that of large-cap based portfolio. This can be potentially explained by the fact that small-cap stocks returns are more volatile than large-cap stocks. We leave future studies for a more thorough investigation.
	\begin{figure}[h!]
		\centering
		\includegraphics[width=16cm]{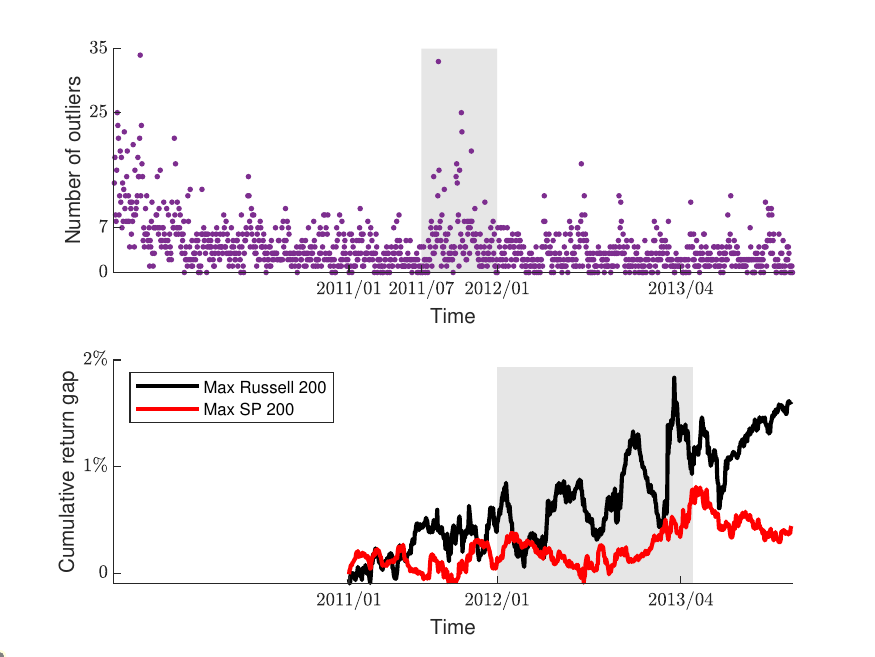}
		\caption{ Time plots of the number of heterogeneous outliers  and cumulative excess return gaps between R-MVP based portfolio and POET based portfolio.}
		\label{Figure cumulative plot}
	\end{figure}	
	
	\subsection{Long-term Performance}
	In this subsection, we further exhibit the out-of-sample performance of the proposed R-MVP over a  longer period of time. To be specific, we continue to examine the same investment universe as in the previous subsections, but over a longer out-of-sample periods: (1) largest 200 stocks in S\&P index, out of sample period is from 01/03/2006 to 12/31/2014; (2) largest 200 stocks in Russell 2000 index, out of sample period is from 01/03/2011 to 12/31/2019. The length of used historical returns is 400, and portfolios are rebalanced monthly. We report the results of cumulative return, risk, Sharpe ratio, maximum drawdown and turnover ratio with and without transaction costs in Table \ref{Table empirical monthly longrun}. It shows our strategy has good performance in the long term. 
	The R-MVP is the best overall performer in Sharpe ratio. It achieves a good balance between satisfactory returns while the risk is controlled at the near-minimum level. For example, R-MVP has the largest SR 0.0806 and CR 1.19, the lowest risk 0.00657 when investing in scenario (2). ``1/N'' strategy gets higher cumulative return than R-MVP, but undertaking a much higher risk and maximum drawdown when investing in S\&P components.
	And those metrics of RCOV and eFM-COV are not on par with other portfolios due to large risk exposures and hence the drawdown. Also one can observe that R-MVP outperforms POET-based portfolio.
	
	\begin{table}[htbp!]
		\centering
		\caption{Empirical daily application under rolling window scheme, (1) largest 200 stocks in S\&P 500 index, out of sample period is from 01/03/2006 to 12/31/2015; (2) largest 200 stocks in Russell 2000 index, out of sample period is from 01/03/2011 to 12/31/2019; sample size T = 400 for both cases, monthly holding period.}
		\begin{threeparttable}
			{
				\begin{tabular}{cccccccccccc}
					\noalign{\global\arrayrulewidth1pt}\hline \noalign{\global\arrayrulewidth0.4pt}
					& \multicolumn{5}{c}{max SP 200}   &         & \multicolumn{5}{c}{max Russell 200}      \\ \cline{3-5} \cline{8-11}
					& CR     & Risk    & SR      & MDD   & TR  &  & CR    & Risk    & SR     & MDD   & TR     \\ \hline
					\multicolumn{12}{c}{Without transaction   cost}                                                                                                        \\           
					R-MVP       & 0.722   & 0.00805 & 0.0399  & 0.341  & -     &  & 1.191  & 0.00657 & 0.0806  & 0.118  &       \\
Sample      & 0.441   & 0.00897 & 0.0219  & 0.443  & -     &  & 0.942  & 0.00778 & 0.0539  & 0.211  & -     \\
1/N         & 1.008   & 0.01544 & 0.0290  & 0.802  & -     &  & 0.768  & 0.01286 & 0.0266  & 0.396  & -     \\
linear      & 0.521   & 0.00813 & 0.0285  & 0.390  & -     &  & 1.079  & 0.00710 & 0.0676  & 0.160  & -     \\
nonlinear   & 0.627   & 0.00788 & 0.0354  & 0.384  & -     &  & 1.128  & 0.00703 & 0.0714  & 0.156  & -     \\
POET        & 0.693   & 0.00806 & 0.0382  & 0.374  & -     &  & 1.146  & 0.00660 & 0.0773  & 0.124  & -     \\
MVP-UF      & 0.693   & 0.00806 & 0.0382  & 0.374  & -     &  & 1.173  & 0.00662 & 0.0789  & 0.124  & -     \\
M-portfolio & -0.006  & 0.00951 & -0.0003 & 0.691  & -     &  & 1.122  & 0.00843 & 0.0592  & 0.263  & -     \\
RCOV        & -0.361  & 0.04232 & -0.0038 & 2.487  & -     &  & 1.156  & 0.00782 & 0.0658  & 0.195  & -     \\
eFM-COV     & -4.196  & 0.08466 & -0.0221 & 6.324  & -     &  & 2.382  & 0.48559 & 0.0022  & 33.073 & -  \\ \hline
     
					\multicolumn{12}{c}{With transaction   cost}                                                                                                    \\
					R-MVP       & 0.550   & 0.00803 & 0.0305  & 0.416  & 1.07  &  & 1.1039 & 0.00653 & 0.0752  & 0.119  & 0.48  \\
Sample      & 0.189   & 0.00894 & 0.0094  & 0.544  & 1.72  &  & 0.855  & 0.00842 & 0.0452  & 0.322  & 1.17  \\
1/N         & 1.030   & 0.01540 & 0.0298  & 0.780  & 0.01  &  & 0.813  & 0.01284 & 0.0282  & 0.396  & 0.02  \\
linear      & 0.341   & 0.00809 & 0.0188  & 0.471  & 1.12  &  & 0.967  & 0.00708 & 0.0608  & 0.163  & 0.65  \\
nonlinear   & 0.492   & 0.00783 & 0.0280  & 0.449  & 0.78  &  & 1.009  & 0.00700 & 0.0642  & 0.159  & 0.71  \\
POET        & 0.519   & 0.00804 & 0.0287  & 0.449  & 1.08  &  & 1.049  & 0.00657 & 0.0711  & 0.125  & 0.51  \\
MVP-UF      & 0.519   & 0.00804 & 0.0287  & 0.449  & 1.08  &  & 1.074  & 0.00658 & 0.0726  & 0.125  & 0.54  \\
M-portfolio & -0.382  & 0.00950 & -0.0179 & 0.890  & 2.96  &  & 0.855  & 0.00842 & 0.0452  & 0.322  & 1.98  \\
RCOV        & -2.321  & 0.04291 & -0.0241 & 3.511  & 16.46 &  & 0.917  & 0.00781 & 0.0523  & 0.269  & 1.98  \\
eFM-COV     & -10.568 & 0.08669 & -0.0543 & 11.284 & 38.16 &  & -3.071 & 0.47468 & -0.0029 & 32.148 & 85.62
					\\
					\noalign{\global\arrayrulewidth1pt}\hline \noalign{\global\arrayrulewidth0.4pt}    
				\end{tabular}	
			}
			\begin{tablenotes}
				\item[a] see notes in Table \ref{Table empirical monthly}.
			\end{tablenotes}
		\end{threeparttable}
		\label{Table empirical monthly longrun}
	\end{table}
	
	Furthermore, we visualize the out of sample cumulative excess returns in Figure \ref{Figure SP and Ru}. In the  plot, we stop drawing the line of cumulative returns once it  reaches -50\%,   due to the empirical observations that the average maximum drawdown of delisted funds is approximately 40\%. We do not plot the linear shrinkage based portfolio and the Kendall's tau covariance estimator based portfolio (eFM-COV) since they are dominated by nonlinear shrinkage based portfolio. The MVP-UF strategy in scenario (1) is exactly the same with POET based portfolio. Hence, we do not plot it as well. It is evident that the R-MVP exhibits superior overall performance in the long investment horizon. If we zoom in the period from 2008 to 2010, which includes the financial crisis, it can be observed that our R-MVP is much more robust than others and exhibits the lowest drawdown, while the cumulative returns of  ``1/N'' portfolio and M-portfolio exceed -50\%. 
	
	\begin{figure}[h!]
		\centering
		\includegraphics[width=18cm,height =18cm]{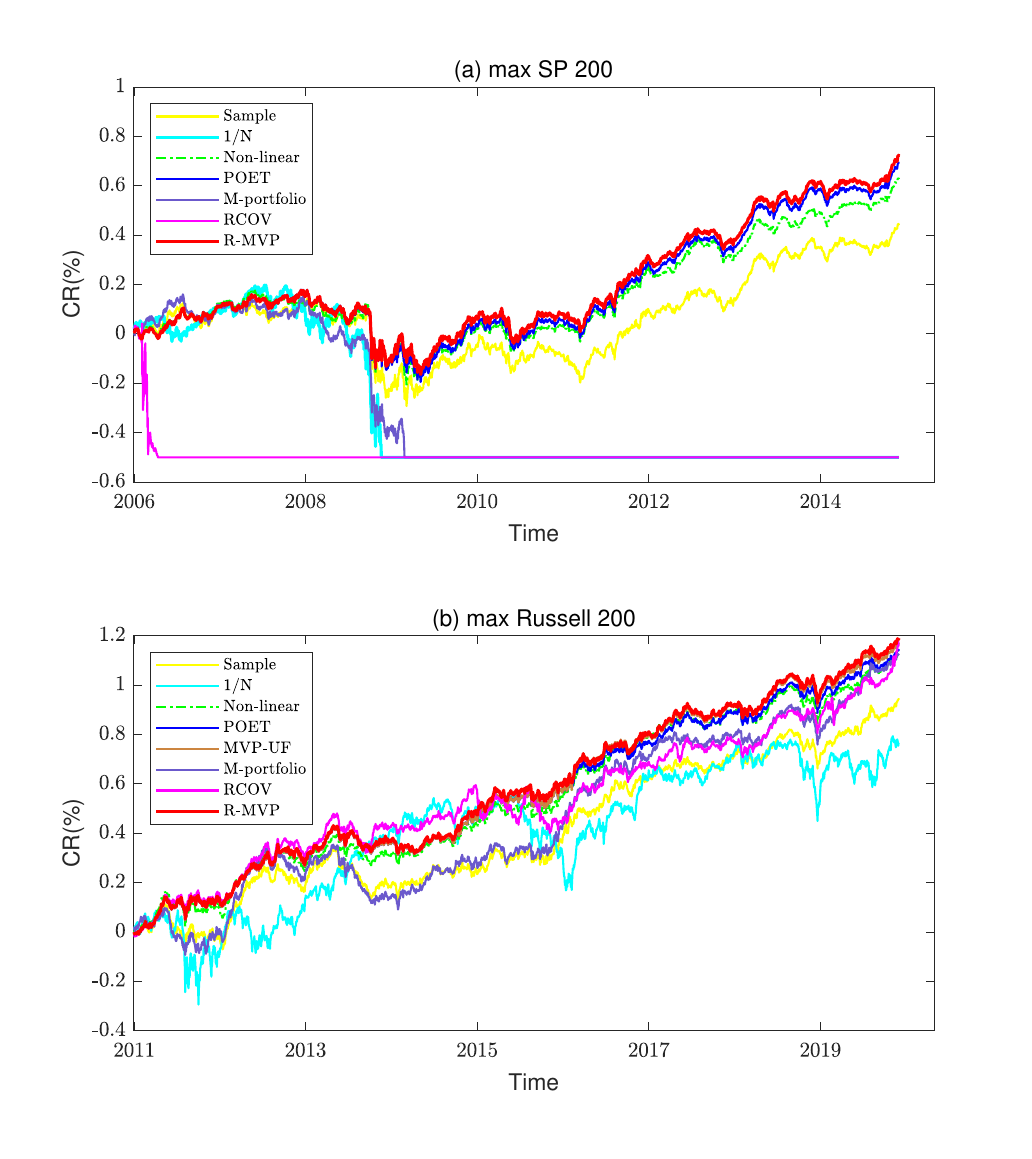}
  
		\caption{ Time plots of cumulative excess return of considered portfolios. Panel (a): largest 200 stocks in S\&P 500 index, time period 01/03/2006-12/31/2014; Panel (b): largest 200 stocks in Russell 2000 index, time period 01/03/2011-12/31/2019.}
		\label{Figure SP and Ru}
	\end{figure}

\section{Conclusion}\label{concl}
This paper proposes a robust minimum variance portfolio that allows for outliers in the form of different shocks. Our method is adaptive to global or idiocyncratic shocks data, utilizing the robust PCA for approximate factor model estimation and a thresholding method for the residual covariance matrix. We develop theorems of estimation consistency and the desired properties of the optimized portfolio. Simulation studies and real data analysis show the robust and superior performance of the proposed portfolio in various outlier settings.

\newpage
	
 \begin{appendices}
	
	\setcounter{page}{1}
	\counterwithin{figure}{section}
	\counterwithin{table}{section}
	\counterwithin{assumption}{section}
	\counterwithin{lemma}{section}
	
	{\Large \bf 
		\begin{center}
			Supplementary Material to ``Shocks-adaptive Robust Minimum Variance Portfolio for a Large Universe of Assets''
		\end{center}
	}

	This supplementary material consists of two parts. Appendix A offers some basic lemmas that are useful for proving the main results in the paper. Appendix B presents detailed proofs for Lemma 1 and Theorems 1-4 in the main paper.

	Throughout the appendix, $\lambda_{max}(A)$ and $\lambda_{min}(A)$ take the largest and smallest eigenvalue of matrix $A$, respectively. $\|A\|$ and $\|A\|_{F}$ are spectral norm and Frobenius norm  of matrix $A$, defined respectively by $\|A\| = \lambda_{max}^{1/2}(A^\top A)$ and $\|A\|_F = tr^{1/2}(A^\top A)$. And when $A$ is a vector, $\| A\|$ and $\|A\|_F$ are equal to Euclidean norm. $C$ is some positive constant that may change from line to line.

	\section{Basic Lemmas}
	In this section, we first derive the asymptotic results for the estimations from principal component analysis of return data. These results will then be applied to the transformed returns $\tilde{R}$ in the next section.
	
	Let us begin by considering the asset returns following the factor structure:
	\begin{eqnarray}
		r_t=BF_t+e_t, \ \ \ t=1, 2, \ldots, T, 
		\label{eq: vector form app}
	\end{eqnarray}
	where $r_t=\left(r_{1t}, r_{2t}, \ldots, r_{pt}\right)^{\top}$ is the vector of returns for $p$ assets, $B=\left(b_{1}, b_{2}, \ldots, b_{p}\right)^{\top}$ is the factor loading matrix, and $e_t=\left(e_{1t}, e_{2t}, \ldots, e_{pt}\right)^{\top}$. Further, let $\Sigma_e$ denote the covariance matrix of the error terms, with each element denoted by $\sigma_{e,ij}$. By vector form of factor structure \eqref{eq: vector form app}, we have the matrix form $$R = BF^\top + \mathcal{E},$$ where $R = (r_1,\ldots,r_T)$, $B^\top = (b_1,\ldots,b_p)$,  $F = (F_1,\ldots,F_T)^\top$ and $\mathcal{E} = (e_1,\ldots, e_T)$. Thus $R^\top = FB^\top + \mathcal{E}^\top$.  Let $Y = R^\top = (Y_1,\ldots,Y_p)$,   $\mathcal{E}^\top = (\varepsilon_1,\ldots,\varepsilon_p)$,  $\varepsilon_s = (\varepsilon_{s1},\ldots,\varepsilon_{sT})^\top$ and $\varepsilon_{st} = e_{st}$, we then have
	\begin{align}
		Y_i = F b_i + \varepsilon_i. 
		\label{eq: Appendix reform factor model}
	\end{align}

	We consider the following optimization problem: 
	\begin{align}
		\label{eq: PCA objective appendix}
		&argmin_{F,B} \| Y - FB^\top \|_F^2
		\\
		& \frac{1}{p}B^\top B = I_m, \  F^\top F \text{\ is diagonal.} \notag
	\end{align}
	It has been shown that the columns of estimated factor loading $\hat{B} = (\hat{b}_1,\ldots,\hat{b}_p)^\top$ are $\sqrt{p}$ times the eigenvectors corresponding to the $m$ largest eigenvalues of the $p \times p$ matrix $Y^\top Y$ where $\hat{b}_i = (\hat{b}_{1i},\ldots,\hat{b}_{mi})^\top$, and $\hat{F} = p^{-1}Y\hat{B}$ \citep{bai2002determining,bai2003}.

	By the same steps of (A.1) in \cite{bai2003}, we  have the following identity:
	\begin{align}
		\hat{b}_i - Hb_i = (V/T)^{-1}\left(\frac{1}{p}\sum_{s=1}^p\hat{b}_s E(\varepsilon_s^\top\varepsilon_i)/T + \frac{1}{p}\sum_{s=1}^p \hat{b}_s\zeta_{si} + \frac{1}{p}\sum_{s=1}^T\hat{b}_s\eta_{si} + \frac{1}{p}\sum_{s=1}^p\hat{b}_s\xi_{si}\right)
		\label{eq: appendix b-Hb}
	\end{align}
	where  $V$ is the $m \times m$ diagonal matrix of the first $m$ largest eigenvalues of $p^{-1}Y Y^\top$ in descending order, $H = \frac{1}{p}V^{-1}\hat{B}^\top BF^\top F$, $\zeta_{si} = \varepsilon_s^\top \varepsilon_i/T - E(\varepsilon_s^\top \varepsilon_i)/T$, $\eta_{si} = b_s^\top \sum_{t=1}^T F_t e_{it}/T$, and $\xi_{si} = b_i^\top \sum_{t=1}^T F_te_{st}/T$. 
	
	\begin{assumption}
		\label{assum: Appendix basic}
		\begin{itemize}
			\item[(i)] $\{e_t, F_t\}_{t\geq 1}$ is strictly stationary, and $E(e_{it}) = E(e_{it}F_{jt}) = 0$ for all $i \leq p, j \leq m$ and $t\leq T$. 
			\item[(ii)] $\max_i \sum_{s=1}^p|E(\varepsilon_s^\top\varepsilon_i)|/T = O(1) $.  
			\item[(iii)]  For all $s,i\leq p$, $$E\left(\varepsilon_s^\top \varepsilon_i - E(\varepsilon_s^\top \varepsilon_i ) \right)^4 = E\left[\frac{1}{T}\sum_{t=1}^T(\varepsilon_{st} \varepsilon_{it} ) - E(\varepsilon_{st} \varepsilon_{it})  \right]^4 =O(T^2).$$
			\item[(iv)] There exist $M >0$ such that for all $i\leq p$, $E\|b_i\|^2 < M$.
			
			\item[(v)] There exist $\varphi_1, \varphi_2 > 0$ and $d_1, d_2 > 0$, such that for any $s>0$, $i\leq p$ and $j\leq m$,
			$$
			\probP(|{e}_{it}| > s) \leq  exp(-(s/d_1)^{\varphi_1}),  \probP(|{F}_{jt}| > s) \leq  exp(-(s/d_2)^{\varphi_2}).
			$$
			
			\item[(vi)] $\{e_t, F_t\}$ is strong mixing process: there exists $\varphi_3 >0$ and $C>0$ such that for all $t\in \mathbb{Z}^+$(the set of positive integers),
			$$
			\alpha(t) \leq exp(-Ct^{\varphi_3})
			$$
			where $\alpha$ is $\alpha$-mixing coefficient defined based on $\sigma$-algebras generated by $\{e_t, F_t\}$.
			\item[(vii)] $\tilde{\varphi}^{-1} = 3\varphi_1^{-1} + \varphi_3^{-1} > 1$, $3\varphi_2^{-1} + \varphi_3^{-1} >1$. Let $\varphi = 1.5\varphi_1^{-1} + 1.5\varphi_2^{-1} + \varphi_3^{-1}$,  $(log p)^{2/\varphi - 1} = o(T)$ and $(logp)^{2/\tilde{\varphi}-1} = o(T)$.  
		\end{itemize}
	\end{assumption}

	\begin{lemma}
		\label{lemma basic C7} Suppose Assumption \ref{assum: Appendix basic} holds,
		for all $j \leq m$,
		\begin{enumerate}
			\item[(i)] $\frac{1}{p}\sum_{i = 1}^p\left( \frac{1}{p}\sum_{s=1}^p \hat{b}_{js}E(\varepsilon_s^\top\varepsilon_i)/T \right)^2 = O_p\left(\frac{1}{p}\right)$,
			\item[(ii)]
			$\frac{1}{p}\sum_{i = 1}^p\left(\frac{1}{p}\sum_{s=1}^p\hat{b}_{js}\zeta_{si} \right)^2 = O_p\left(\frac{1}{T}\right)$,
			\item[(iii)]
			$\frac{1}{p}\sum_{i = 1}^p\left(\frac{1}{p}\sum_{s=1}^p\hat{b}_{js}\eta_{si} \right)^2 = O_p\left(\frac{logp}{T}\right)$,
			\item[(iv)]
			$\frac{1}{p}\sum_{i = 1}^p\left(\frac{1}{p}\sum_{s=1}^p\hat{b}_{js}\xi_{si} \right)^2 = O_p\left(\frac{logp}{T}\right)$.
		\end{enumerate}
	\end{lemma}
	
	\textbf{Proof:} (i) First, we have $\forall  j$, $\sum_{s=1}^p \hat{b}_{js}^2= p$. By Cauchy-Schwarz inequality,
	\begin{align*}
		&\frac{1}{p}\sum_{i = 1}^p\left( \frac{1}{p}\sum_{s=1}^p \hat{b}_{js}E(\varepsilon_s^\top\varepsilon_i)/T \right)^2 \leq  \frac{1}{p}\sum_{i = 1}^p \frac{1}{p} \sum_{s=1}^p (E(\varepsilon_s^\top\varepsilon_i)/T)^2
		\\ &
		\leq \max_{i \leq p} \frac{1}{p} \sum_{s=1}^p (E(\varepsilon_s^\top\varepsilon_i)/T)^2 \leq  \max_{i,s} | E(\varepsilon_s^\top\varepsilon_i)/T|\max_i \frac{1}{p} \sum_{s=1}^p |E(\varepsilon_s^\top\varepsilon_i)/T|.
		\\ &
		= O(p^{-1}). 
	\end{align*}
	where  the last equality holds by Assumption \ref{assum: Appendix basic}(ii).
	
	(ii) By using Cauchy-Schwarz inequality, 
	\begin{align*}
		\frac{1}{p}\sum_{i = 1}^p\left(\frac{1}{p}\sum_{s=1}^p\hat{b}_{js}\zeta_{si} \right)^2
		&= \frac{1}{p^3}
		\sum_{s=1}^p \sum_{l=1}^p \hat{b}_{js} \hat{b}_{jl}
		\left(\sum_{i=1}^p \zeta_{si}\zeta_{li}\right) \leq  \frac{1}{p^3}\left(\sum_{sl}\left(\hat{b}_{js} \hat{b}_{jl}\right)^2 \sum_{sl}  \left(\sum_{i=1}^p \zeta_{si}\zeta_{li}\right)^2\right)^{1/2}
		\\ &
		\leq \frac{1}{p^2} \left(\sum_{s=1}^p\sum_{l=1}^p  \left(\sum_{i=1}^p \zeta_{si}\zeta_{li}\right)^2\right)^{1/2}.
	\end{align*}
	Note that $E\left(\sum_{s=1}^p\sum_{l=1}^p  \left(\sum_{i=1}^p \zeta_{si}\zeta_{li}\right)^2\right) = p^2 E\left( \left(\sum_{i=1}^p \zeta_{si}\zeta_{li}\right)^2\right) \leq p^4 \max_{s,i}E|\zeta_{si}|^4$. By Assumption \ref{assum: Appendix basic}(iii), $max_{s,i} E\zeta_{si}^4 = O(T^{-2})$, it indicates that $\left(\sum_{s=1}^p\sum_{l=1}^p  \left(\sum_{i=1}^p \zeta_{si}\zeta_{li}\right)^2\right) = O_p(p^4/T^2)$ and thus yields the result.
	
	(iii) Similarly, we have
	\begin{align*}
		\frac{1}{p}\sum_{i = 1}^p\left(\frac{1}{p}\sum_{s=1}^p\hat{b}_{js}\eta_{si} \right)^2 & \leq \left\| \frac{1}{p}\sum_{s=1}^p \hat{b}_{js}b_s^\top\right\|^2 \frac{1}{p}\sum_{i=1}^p \left\|\frac{1}{T} \sum_{t=1}^T F_t e_{it} \right\|^2
		\\ & \leq \frac{1}{p}\sum_{i=1}^p \left\|\frac{1}{T} \sum_{t=1}^T F_t e_{it} \right\|^2
		\left(\frac{1}{p}\sum_{s=1}^p \hat{b}_{js}^2 \frac{1}{p}\sum_{s=1}^p \| b_s\|^2 \right).
	\end{align*}
	By applying Assumption \ref{assum: Appendix basic}(iv),  first we have $\frac{1}{p}\sum_{s=1}^p \| b_s\|^2 = O_p(1)$. Then,  based on Assumption \ref{assum: Appendix basic}(i) and (v)-(vii),  following the Lemma B.1 of \cite{fan2011high}, we also have $\max_{i,j} \left|\frac{1}{T} \sum_{t=1}^T F_{jt} e_{it} \right| = O_p\left(\sqrt{\frac{logp}{T}}\right)$,
	and thus
	$max_i \left\|\frac{1}{T} \sum_{t=1}^T F_t e_{it} \right\|^2 \leq max_i \sum_{j=1}^m \left(\frac{1}{T}\sum_{t=1}^TF_{jt}e_{it} \right)^2 = O_p\left(\frac{logp}{T}\right)$.

	(iv) Similar to part (iii), we have 
	\begin{align*}
		\frac{1}{p}\sum_{i = 1}^p\left(\frac{1}{p}\sum_{s=1}^p\hat{b}_{js}\xi_{si} \right)^2 & = \frac{1}{p}\sum_{i = 1}^p\left| \frac{1}{p}\sum_{s=1}^p{b}_{i}^\top \sum_{t=1}^T F_t e_{st} \frac{1}{T}\hat{b}_{js} \right|^2 \leq  \frac{1}{p}\sum_{i = 1}^p \| b_i\|^2 \times \left\|\frac{1}{p}\sum_{s=1}^p\sum_{t=1}^T F_t e_{st} \frac{1}{T} \hat{b}_{js} \right\|^2
		\\ & 
		\leq O_p(1) \frac{1}{p}\sum_{s=1}^p \left\| \frac{1}{T}\sum_{t=1}^T F_t e_{st} \right\|^2\times \frac{1}{p}\sum_{s=1}^T \hat{b}_{js}^2 = O_p\left(\frac{logp}{T}\right).
	\end{align*}

	\begin{lemma}
		\label{lemma: appendix C8}
		\begin{enumerate}
			Suppose Assumption \ref{assum: Appendix basic} holds, 
			\item[(i)] $\max_{i\leq p} \left\| \frac{1}{Tp}\sum_{s=1}^p \hat{b}_s E(\varepsilon_s^\top \varepsilon_i) \right\| = O_p\left(\frac{1}{\sqrt{p}}\right)$,
			\item[(ii)] 
			$\max_{i\leq p} \left\| \frac{1}{p}\sum_{s=1}^p \hat{b}_s\zeta_{si} \right\| = O_p\left(\frac{p^{1/4}}{\sqrt{T}}\right)$,
			\item[(iii)] 
			$\max_{i\leq p} \left\| \frac{1}{p}\sum_{s=1}^p \hat{b}_s\eta_{si} \right\| = O_p\left(\sqrt{\frac{logp}{T}}\right)$,
			\item[(iv)] 
			$\max_{i\leq p} \left\| \frac{1}{p}\sum_{s=1}^p \hat{b}_s\xi_{si} \right\| = O_p\left(\sqrt{\frac{logp}{T}}\right)$.
		\end{enumerate}
		
	\end{lemma}

	\textbf{Proof:} (i) By the Cauchy-Schwarz inequality and the fact that $\frac{1}{p}\sum_{i=1}^p \| \hat{b}_i\|^2 = O_p(1)$, we have 
	\begin{align*}
		\max_{i\leq p} \left\| \frac{1}{Tp}\sum_{s=1}^p \hat{b}_s E(\varepsilon_s^\top \varepsilon_i) \right\| & \leq \max_{i\leq p}\left( \frac{1}{p}\sum_{s=1}^p\| \hat{b}_s\|^2\frac{1}{p} \sum_{s=1}^p E(\varepsilon_s^\top \varepsilon_i/T)^2\right)^{1/2}
		\\& \leq O_p(1) \max_{i\leq p}\left( \frac{1}{p} \sum_{s=1}^p E(\varepsilon_s^\top \varepsilon_i/T)^2\right)^{1/2}
		\\ & \leq O_p(1) \max_{s,i} \sqrt{|E(\varepsilon_s^\top \varepsilon_i/T)|} \max_{i} \left( \frac{1}{p} \sum_{s=1}^p |E(\varepsilon_s^\top \varepsilon_i/T)|\right)^{1/2} 
		\\ & = O_p(p^{-1/2}),
	\end{align*}
	where we apply Assumption \ref{assum: Appendix basic}(i) and (ii).
	
	(ii) Similarly, 
	\begin{align*}
		\max_{i\leq p} \left\| \frac{1}{p}\sum_{s=1}^p \hat{b}_s\zeta_{si} \right\| \leq  \max_{i\leq p} \frac{1}{p} \left( \sum_{s=1}^p \| \hat{b}_s\|^2 \sum_{s=1}^p \zeta_{si}^2 \right)^{1/2} \leq \left(O_p(1) \max_i \frac{1}{p}\sum_{s=1}^p \zeta_{si}^2 \right)^{1/2}.
	\end{align*}
	By Assumption \ref{assum: Appendix basic}(iii), $E\left( \frac{1}{p}\sum_{s=1}^p \zeta_{si}^2 \right)^2 \leq \max_{s,i}E(\zeta_{si}^4) =O(T^{-2})$. It then follows from the Chebyshev's inequality and Bonferroni's method that $\max_i \frac{1}{p}\sum_{s=1}^p \zeta_{si}^2 = O_p(\sqrt{p}/T)$.
	
	(iii) We have 
	\begin{align*}
		\max_{i\leq p} \left\| \frac{1}{p}\sum_{s=1}^p \hat{b}_s\eta_{si} \right\| \leq \left\| \frac{1}{p}\sum_{s=1}^p \hat{b}_sb_s^\top \right\| \max_i \left\| \frac{1}{T}\sum_{t=1}^T F_te_{it}  \right\| = O_p\left(\sqrt{\frac{logp}{T}}\right),  
	\end{align*}
	where the last equality holds by the fact that $\max_i\left\| \frac{1}{T}\sum_{t=1}^TF_te_{it}\right\| = O_p(\sqrt{\frac{logp}{T}})$ from the proof of Lemma \ref{lemma basic C7}(iii).

	(iv) Also, we have
	\begin{align*}
		\max_{i\leq p} \left\| \frac{1}{p}\sum_{s=1}^p \hat{b}_s\xi_{si} \right\| \leq \max_{i\leq p}\| b_i \|\times \left\|\frac{1}{p}\sum_{s=1}^p\sum_{t=1}^T F_t e_{st}\frac{1}{T}\hat{b}_s \right\| = O_p\left(\sqrt{\frac{logp}{T}}\right),
	\end{align*}
	where $\left\|\frac{1}{p}\sum_{s=1}^p\sum_{t=1}^T F_t e_{st}\frac{1}{T}\hat{b}_s \right\| = O_p\left(\sqrt{\frac{logp}{T}}\right)$ from the proof of Lemma \ref{lemma basic C7}(iv) and $\max_{i\leq p} E\|b_i \| = O_p(1)$ by Assumption \ref{assum: Appendix basic}(iv).

	\begin{assumption}
		\label{assum Appendix basic 2}
		\begin{enumerate}
			\item[(i)]  $B^\top B$ is diagonal and  $\| I_m -  p^{-1}B^\top B\|_F =O_p(p^{-1/2})$.
			\item[(ii)] $cov(F_t) = I_m$, and $\|T^{-1}F^\top F - I_m\| = o_p(1)$.
			\item[(iii)] There exists some positive constant $c_1$ such that $\|\Sigma_e\| \leq c_1$, where $\Sigma_e$ is covariance matrix of $e_t$.
		\end{enumerate}
	\end{assumption}

	\begin{lemma}
		\label{lemma appendix C4}
		Suppose Assumption \ref{assum Appendix basic 2} holds, 
		let $\hat{\lambda}_m$ denote the $m$-th largest eigenvalue of ${p}^{-1}YY^\top$, then $\hat{\lambda}_m > C_1T$ with probability approaching one for some constant $C_1 > 0$.
	\end{lemma}
	
	\textbf{Proof:} We first note that the eigenvalues of $T^{-1}YY^\top$ and $T^{-1}Y^\top Y$ only differ by $|T-p|$ zero eigenvalues. Since $Y= R^\top$,  $T^{-1}Y^\top Y = T^{-1}RR^\top$. Assumption \ref{assum Appendix basic 2}(i) indicates that all eigenvalues of the $m \times m$ matrix $p^{-1}B^\top B$ approach to 1  as $p \rightarrow \infty$, which are bounded away from both zero and infinity. As such, under Assumptions \ref{assum Appendix basic 2}, by applying Proposition 2.1 and Lemma C.4 of \cite{fan2013large},  we have that $\nu_m \geq C p $ for some $C > 0$ where $\nu_m$ is the $m$-th largest eigenvalue of $BB^\top + \Sigma_e$ for sufficiently large $p$. Using Weyl's theorem, if we show that $\| T^{-1}RR^\top - BB^\top -\Sigma_e\| = o_p(p)$, then we can conclude that $\hat{\nu}_m > C_1 p$ with probability approaching one for some $C_1 > 0$ where $\hat{\nu}_m$ is the $m$-th largest eigenvalue of $T^{-1}RR^\top$. As a result, $\hat{\lambda}_m = T\hat{\nu}_m/p > C_1 T$ with probability approaching one. 
	
	Based on factor structure, we have $$T^{-1}RR^\top = B\frac{1}{T}\sum_{t=1}^TF_tF_t^\top B^\top + \frac{1}{T}\sum_{t=1}^T e_te_t^\top + B\frac{1}{T}\sum_{t=1}^T F_te_t^\top + \left( B\frac{1}{T}\sum_{t=1}^T F_te_t^\top \right)^\top.$$
	Then, we have $\|B\left(\frac{1}{T}\sum_{t=1}^TF_tF_t^\top -I_m\right) B^\top\| = o_p(p)$ by Assumption \ref{assum Appendix basic 2}(ii). By Assumption \ref{assum: Appendix basic}(v)-(vii), the Lemma A.3 of \cite{fan2011high} implies that $\| \frac{1}{T}\sum_{t=1}^T(e_te_t^\top -\Sigma_e)\| = O_p(p\sqrt{logp/T}) = o_p(p)$. By the fact that $\max_i\left\| \frac{1}{T}\sum_{t=1}^TF_te_{it}\right\| = O_p(\sqrt{\frac{logp}{T}})$, it is easy to obtain that $\|B\frac{1}{T}\sum_{t=1}^T F_t e_t^\top \|= o_p(p)$. We finish the proof.

	\begin{lemma}
		\label{lemma Appendix C9}
		Suppose Assumptions \ref{assum: Appendix basic} and \ref{assum Appendix basic 2} hold, 
		\begin{enumerate}
			\item[(i)] $\max_{j\leq m} \frac{1}{p}\sum_{i=1}^p \left(\hat{b}_i - Hb_i \right)_j^2 = O_p\left(\frac{1}{p} + \frac{logp}{T}\right)$,
			\item[(ii)] $\frac{1}{p}\sum_{i=1}^p \left\| \hat{b}_i - Hb_i\right\|^2 = O_p\left(\frac{1}{p} + \frac{logp}{T}\right)$,
			\item[(iii)]
			$\max_{i\leq p} \left\| \hat{b}_i - Hb_i\right\|= O_p\left(\frac{1}{\sqrt{p}}+\frac{p^{1/4}}{\sqrt{T}}\right)$.
		\end{enumerate}
	\end{lemma}
	
	\textbf{Proof:} By Lemma \ref{lemma appendix C4}, all the eigenvalues of $V/T$ are bounded away from zero and infinity. Using the inequality $(a+b+c+d)^2 \leq 4(a^2 + b^2 +c^2+d^2)$ and \eqref{eq: appendix b-Hb}, for some constant $C>0$, we have
	\begin{align*}
		\max_{j\leq m} \frac{1}{p}\sum_{i=1}^p \left(\hat{b}_i - Hb_i \right)_j^2 &\leq C \max_{j}\frac{1}{p}\sum_{i = 1}^p\left( \frac{1}{p}\sum_{s=1}^p \hat{b}_{js}E(\varepsilon_s^\top\varepsilon_i)/T \right)^2 + C\max_{j} \frac{1}{p}\sum_{i = 1}^p\left(\frac{1}{p}\sum_{s=1}^p\hat{b}_{js}\zeta_{si} \right)^2 
		\\& + C \max_j \frac{1}{p}\sum_{i = 1}^p\left(\frac{1}{p}\sum_{s=1}^p\hat{b}_{js}\eta_{si} \right)^2 + C\max_j \frac{1}{p}\sum_{i = 1}^p\left(\frac{1}{p}\sum_{s=1}^p\hat{b}_{js}\xi_{si} \right)^2
	\end{align*}
	Each of the four terms on the right hand side above are bounded in Lemma \ref{lemma basic C7}. 
	
	(ii) The result follows that $\frac{1}{p}\sum_{i=1}^p \left\| \hat{b}_i - Hb_i\right\|^2 \leq m \max_j \frac{1}{p}\sum_{i=1}^p \left(\hat{b}_i - Hb_i \right)_j^2$.
	
	(iii) The result follows \eqref{eq: appendix b-Hb} and Lemma \ref{lemma: appendix C8} directly.

	\begin{lemma}
		\label{lemma Appendix C10}
		Suppose Assumptions \ref{assum: Appendix basic} and \ref{assum Appendix basic 2} hold, 
		\begin{enumerate}
			\item[(i)] $HH^\top = I_m + O_p\left(\sqrt{\frac{logp}{T}}+\frac{1}{\sqrt{p}}\right)$,
			\item[(ii)]$H^\top H = I_m + O_p\left(\sqrt{\frac{logp}{T}}+\frac{1}{\sqrt{p}}\right)$
		\end{enumerate}
	\end{lemma}
	
	\textbf{Proof:} First of all, by Lemma \ref{lemma appendix C4}, $\| V^{-1}\| = O_p(T^{-1})$. Furthermore, $\|\hat{B}\| = \sqrt{p}$ by $p^{-1}\hat{B}^\top \hat{B} = I_m$,  $\| B\| = \lambda_{max}^{1/2}(BB^\top) = O_p(\sqrt{p})$ by Assumption \ref{assum Appendix basic 2}(i), and $\|\frac{1}{T}F^\top F\| = \|\frac{1}{T}F F^\top\| = O_p(1)$ by Assumption \ref{assum Appendix basic 2}(ii). It then follows from the definition of $H$ that $\|H\| = O_p(1)$. Then
	\begin{align*}
		\| HH^\top - I_m\|_F \leq \| HH^\top - \frac{1}{p}\sum_{i=1}^p Hb_ib_i^\top H^\top \|_F + \| \frac{1}{p}\sum_{i=1}^p Hb_ib_i^\top H^\top - I_m\|_F.
	\end{align*}
	For the first term, it has 
	$$
	\| HH^\top - H\left(\frac{1}{p}\sum_{i=1}^p b_ib_i^\top\right) H^\top \|_F \leq \|H\|^2 \| I_m - \frac{1}{p}\sum_{i=1}^p b_ib_i^\top\|_F =O_p(p^{-1/2}),
	$$
	where we apply Assumption \ref{assum Appendix basic 2}(i). The second term, by Cauchy-Schwarz inequality and Lemma \ref{lemma Appendix C9}, can be bounde as follows:
	\begin{align*}
		\| \frac{1}{p}\sum_{i=1}^p Hb_i(Hb_i)^\top - \frac{1}{p}\sum_{i=1}^p\hat{b}_i \hat{b}_i^\top \|_F &\leq 
		\| \frac{1}{p}\sum_{i=1}^p (Hb_i-\hat{b}_i)(Hb_i)^\top \|_F +  \| \frac{1}{p}\sum_{i=1}^p \hat{b}_i(Hb_i -\hat{b}_i)^\top \|_F 
		\\ & \leq 
		\left( \frac{1}{p}\sum_{i=1}^p\| Hb_i -\hat{b}_i\|^2 \frac{1}{p}\sum_{i=1}^p \|Hb_i\|^2 \right)^{1/2} +  \left( \frac{1}{p}\sum_{i=1}^p\| Hb_i -\hat{b}_i\|^2 \frac{1}{p}\sum_{i=1}^p \|\hat{b}_i\|^2 \right)^{1/2}
		\\ & = O_p\left(\sqrt{\frac{logp}{T}}+\frac{1}{\sqrt{p}}\right).
	\end{align*}

	(ii) Since ${HH}^{\top}={I}_m+O_p\left(\sqrt{\frac{logp}{T}}+\frac{1}{\sqrt{p}}\right)$ and $\left\|H\right\|=O_p(1)$, right
	multiplying $H$ gives $HH^\top{H}= H +O_p\left(\sqrt{\frac{logp}{T}}+\frac{1}{\sqrt{p}}\right)$. Part (i) also gives  that $\|{H}^{-1}\|=O_p(1)$. Hence further left multiplying ${H}^{-1}$ yields ${H}^{\top}{H}={I}_m+O_p\left(\sqrt{\frac{logp}{T}}+\frac{1}{\sqrt{p}}\right).$

	\section{Proof of the Main Results}

	\begin{proof}[Proof of Lemma 1] Recall that the columns of estimated factor loading estimator $\hat{B}$ are $\sqrt{p}$ times the corresponding (to its $m$ largest eigenvalues)
		eigenvectors of $\hat{\mathbf{V}} = \frac{1}{T}\sum_{t=1}^T\omega_t r_tr_t^\top$.
		By model (2.1) in the main paper, we notice that
		\begin{align}
			\label{yr(3)}
			\underbrace{\omega_t^{1/2}r_t}_{\Tilde{r}_t} = B\underbrace{\omega_t^{1/2}F_t}_{\tilde{F}_t} + \underbrace{\omega_t^{1/2} e_t}_{\tilde{e}_t}, t= 1,\ldots, T.
		\end{align}
		where $\Tilde{r}_t = \omega_t^{1/2}r_t$, $ \tilde{F}_t = \omega_t^{1/2}F_t$ and $ \tilde{e}_t =\omega_t^{1/2} e_t $. Let $\tilde{R} = (\tilde{r}_1,\ldots,\tilde{r}_T)$, we have $\tilde{R}\tilde{R}^\top = \sum^{T}_{t=1}\omega_tr_tr_t^{\top}$. As such,  it is clear that the factor loading estimator $\hat{B}$ is also the solution of problem \ref{eq: PCA objective appendix} with $Y = \tilde{R}^\top$. Furthermore, the conditions assumed in Lemma 1 guarantee that the transformed factors $\tilde{F}_t$ and error terms $\tilde{e}_t$ satisfy the Assumptions \ref{assum: Appendix basic} and \ref{assum Appendix basic 2}. As a result, by Lemma \ref{lemma Appendix C9}, $\hat{b}_i$
		has the asymptotic property as follows  
		\begin{align}
			\label{convergence rate for B}
			\max_{i\leq p}\left\vert\left\vert  \hat{b}_i - \tilde{H}b_i\right\vert \right\vert = O_p(p^{-1/2}+p^{1/4}T^{-1/2}),
		\end{align}
		where $\tilde{H} = p^{-1}\tilde{V}^{-1}\hat{B}^\top B \tilde{F}^\top \tilde{F}$, $\tilde{V}$  is the $m \times m$ diagonal matrix of the first $m$ largest eigenvalues of $p^{-1}\tilde{R}^\top \tilde{R}$ in descending order, and $\tilde{F} = (\tilde{F}_1,\ldots,\tilde{F}_T)^\top$.
		
		Now, we further consider $\hat{F}_t = \frac{1}{p}\sum_{j=1}^p r_{jt}\hat{b}_j$. By the fact that  $\frac{1}{p}\sum_{j=1}^p \hat{b}_j\hat{b}_j^\top = I_m$, we can make decomposition as follow, 
		\begin{align}
			\hat{F}_t - \tilde{H}F_t = \frac{1}{p}\sum_{j=1}^p \tilde{H}{b}_j e_{jt} + \frac{1}{p}\sum_{j=1}^p r_{jt}\left(\hat{b}_j - \tilde{H}b_j \right) + \tilde{H }\left( \frac{1}{p}\sum_{j=1}^p b_j b_j^\top -I_m \right)F_t.
			\label{eq: F - HF}
		\end{align}

		For the first term on the right hand side of (\ref{eq: F - HF}), by Assumption 5(c), $E\|\sum_{j=1}^p b_j e_{jt} \|^4 =O(p^2)$ and using Chebyshev's inequality and Bonferroni's method, we  yield $\max_t \| \sum_{j=1}^p b_j e_{jt}\| = O_p(T^{1/4}\sqrt{p})$ with probability one. Thus, by Cauchy-Schwarz inequality, it follows that $$\max_t \left\| \frac{1}{p}\sum_{j=1}^p \tilde{H} {b}_j e_{jt} \right\|  \leq \|\tilde{H}\| \max_t \left\|\frac{1}{p}\sum_{j=1}^p {b}_j e_{jt} \right\| = O_p(T^{1/4}p^{-1/2}),
		$$
		where  $\| \tilde{H}\| = O_p(1)$ indicated by the proof of Lemma \ref{lemma Appendix C10}.
		For the second term, by Cauchy-Schwarz inequality,
		$$
		\max_t \left\| \frac{1}{p}\sum_{j=1}^p r_{jt}\left(\hat{b}_j - \tilde{H}b_j \right) \right\| \leq \max_t \left(\frac{1}{p}\sum_{j=1}^p r_{jt}^2 \frac{1}{p}\sum_{j=1}^p \| \hat{b}_j - \tilde{H}b_j  \|^2 \right)^{1/2} = O_p\left(\delta^{1/2}\left(\frac{1}{\sqrt{p}}+\sqrt{\frac{logp}{T}}\right)\right),
		$$
		where we apply Lemma \ref{lemma Appendix C9}(ii) and the fact that $E(r_{jt}^2) = O(\delta)$ indicated  by Assumption 1(b) and 2(a). Finally, for the third term, by Assumption 1, we can conclude that $$\max_t \left\| \tilde{H }\left( \frac{1}{p}\sum_{j=1}^p b_j b_j^\top -I_m \right)F_t \right\| = O_p\left( \Delta p^{-1/2} \right).$$ As a result, we obtain 
		\begin{align}
			\label{eq: F - HF rate}
			\max_t \| \hat{F}_t - \tilde{H}F_t \| = O_p\left( \frac{\delta^{1/2}+ T^{1/4} + \Delta}{\sqrt{p}} + \sqrt{\frac{\delta log p}{{T}}}\right).
		\end{align}
	\end{proof}
	
	\bigskip
	
	\bigskip

	\begin{lemma}
		\label{lemma: bf- bf}
		Under the assumptions of Lemma 1,
		$$
		\max_{i\leq p, t\leq T} \left\| \hat{b}_i^\top\hat{F}_t - b_i^\top F_t \right\| = O_p\left(\frac{\delta^{1/2}+ T^{1/4}+\Delta}{\sqrt{p}} + \frac{\sqrt{\delta logp }+\Delta p^{1/4}}{\sqrt{T}} \right)
		$$
	\end{lemma}
	\begin{proof}[Proof of Lemma \ref{lemma: bf- bf}]
		Uniformly in $i$ and $t$, we have
		\begin{align}
			\left\| \hat{b}_i^\top\hat{F}_t - b_i^\top F_t \right\| & \leq \left\| \hat{b}_i - \tilde{H}b_i\right\| \left\| \hat{F}_t-\tilde{H}F_t\right\| + \left\| \tilde{H}b_i\right\| \left\| \hat{F}_t - \tilde{H}F_t \right\| \notag
			\\  &+ \left\|\hat{b}_i - \tilde{H}b_i \right\|\|\tilde{H}F_t\| + \left\|b_i\right\| \left\|F_t \right\|\left\|\tilde{H}^\top \tilde{H} - I_m \right\| \notag
			\\ & =O_p\left(\frac{\delta^{1/2}+ T^{1/4}+\Delta}{\sqrt{p}} + \frac{\sqrt{\delta logp}+\Delta p^{1/4}}{\sqrt{T}} \right)
		\end{align}    
	\end{proof}

	\begin{lemma}
		Under the assumptions of Lemma 1, 
		$\max_{i} \frac{1}{T}\sum_{t=1}^T\left| e_{it} - \hat{e}_{it} \right|^2 = O_p\left(\frac{(\delta^{1/2}+ T^{1/4}+\Delta)^2}{{p}} + {\frac{\delta\sqrt{p}}{{T}}}\right)$. If $\frac{\delta^{1/2}+ T^{1/4}+\Delta}{\sqrt{p}} + \frac{\sqrt{\delta logp}+\Delta p^{1/4}}{\sqrt{T}}=o(1)$, then $\max_{i,t}\left| e_{it} - \hat{e}_{it} \right| = o_p(1) $.
		\label{lemma: e -hat e}
	\end{lemma}
	
	\begin{proof}[Proof of Lemma \ref{lemma: e -hat e}]
		We have $e_{it} - \hat{e}_{it} = b_i^\top \tilde{H}^\top(\hat{F}_t - \tilde{H}F_t) + (\hat{b}_i^\top  -b_i^\top \tilde{H}^\top)\hat{F}_t + b_i^\top(\tilde{H}^\top \tilde{H} - I_m)F_t$, using the inequality $(a+b+c)^2\leq 4a^2 + 4b^2+4c^2$, we have
		\begin{align}
			\frac{1}{T}\sum_{t=1}^T\left| e_{it} - \hat{e}_{it} \right|^2 
			&\leq 
			4 \max_i \left\| b_i^\top \tilde{H}^\top\right\|^2 \max_t \left\|\hat{F}_t - \tilde{H}F_t \right\|^2 \notag
			\\ &+ 4 \max_i \left\| \hat{b}_i^\top - b_i^\top \tilde{H}^\top\right\|^2 \frac{1}{T}\sum_{t=1}^T \left\| \hat{F}_t\right\|^2 + 4\max_i \left\| b_i\right\|^2 \frac{1}{T}\sum_{t=1}^T \left\|F_t \right\|^2 \left\| \tilde{H}^\top \tilde{H} - I_m\right\|_F^2 \notag
			\\ &=O_p\left(\frac{(\delta^{1/2}+ T^{1/4}+\Delta)^2}{{p}} + {\frac{\delta logp }{{T}}}\right) + O_p\left( \frac{\delta}{p} + \frac{\delta\sqrt{p}}{T}\right) + O_p\left(\delta\left(\frac{1}{p} +\frac{logp}{T} \right)\right) \notag
			\\ &= O_p\left(\frac{(\delta^{1/2}+ T^{1/4}+\Delta)^2}{{p}} + {\frac{\delta\sqrt{p}}{{T}}}\right).
		\end{align}
		where we use 
		Lemma \ref{lemma Appendix C10}, \eqref{convergence rate for B}, \eqref{eq: F - HF rate} and the fact that 
		$\frac{1}{T}\sum_{t=1}^T \left\| \hat{F}_t\right\|^2$ and $\frac{1}{T}\sum_{t=1}^T \left\| {F}_t\right\|^2$ are both $O_p(\delta)$. The second part follows from Lemma \ref{lemma: bf- bf} directly.
	\end{proof}
	
	\bigskip
	
	\bigskip

	Recall that the adaptive thresholding estimator for $\Sigma_e$ is given by 
	\begin{align}
		\hat{\Sigma}_e = (\hat{\sigma}_{\hat{e},ij})_{p\times p}, \quad \hat{\sigma}_{\hat{e},ij} = \left\{ \begin{array}{ll} s_{\hat{e},ii} & i=j \\ \chi_{ij}(s_{\hat{e},ij}) & i\neq j \end{array} \right.
		\label{eq: adpative thresholding estimator appen}
	\end{align}
	where $s_{\hat{e},ij} = (1/T)\sum_{t=1}^{T}\hat{e}_{it}\hat{e}_{jt}$ is the $(i,j)$-th element of sample covariance matrix based on  $\hat{e}_t$, $\chi_{ij}(\cdot)$ is a shrinkage function satisfying
	$\chi_{ij}(z)=0$ if $|z| \leq \tau_{ij}$, and $|\chi_{ij}(z)-z|\leq \tau_{ij}$, where $\tau_{ij} = C\sqrt{\hat{\theta}_{ij}}\varsigma_T$ is a positive threshold and and $\hat{\theta}_{ij} = \frac{1}{T}\sum_{t=1}^T\left( \hat{e}_{it}\hat{e}_{jt} - s_{\hat{e},ij} \right)^2$.

	\begin{proof}[Proof of Theorem 1] For notation simplicity, we abbreviate $s_{\hat{e},ij}$ as $s_{ij}$ in this proof. By the condition of threshold function, we have $\chi_{ij}(t) = \chi_{ij}(t)\mathbb{I}(|t|>C\varsigma_T \sqrt{\hat{\theta}_{ij}})$.  Under the assumptions of Theorem 1, by Lemma \ref{lemma: e -hat e}, we have $\max_{i,t}|e_{it}-\hat{e}_{it}| = o_p(1)$, $max_i \frac{1}{T}\sum_{t=1}^T |e_{it} - \hat{e}_{it} |^2 = O_p(a_T^2)$ and $a_T = \sqrt{{(\delta^{1/2}+ T^{1/4}+\Delta)^2}/{{p}} + {{\delta\sqrt{p}}/{{T}}}} = o(1)$. Let ${\varsigma}_T = \sqrt{\frac{logp}{T}} + a_T$.  Under the exponential tails and mixing dependence conditions, we can apply  the Lemmas A.3 and A.4 of \cite{fan2011high} and obtain that for any $x > 0$, there are positive constants  $M, \theta_1$ and $\theta_2$ such that each of the events 
		\begin{align*}
			&A_1 = \{\max_{i,j}|s_{\hat{e},ij} - \sigma_{ij} | \leq M {\varsigma}_T\}
			\\ &
			A_2 =  \{ \theta_2 < \sqrt{\hat{\theta}_{ij}} < \theta_1, \text{all} i\leq p, j \leq p\}
		\end{align*}
		occurs with probability at least $1-x$. 
		
		Now for $C = \theta_2^{-1}2M$, under the event $A_1 \cap A_2$, we have
		\begin{align}
			\left\|\hat{\Sigma}_e - \Sigma_e  \right\| &\leq \max_i \sum_{j=1}^p |\chi_{ij}(s_{ij}) - \sigma_{e,ij}| \notag
			\\ & 
			= \max_i \sum_{j=1}^p \left| \chi_{ij}(s_{ij})\mathbb{I}(|s_{ij}|> C\varsigma_T \sqrt{\hat{\theta}_{ij}}) - \sigma_{e,ij}\mathbb{I}(|s_{ij}|> C\varsigma_T \sqrt{\hat{\theta}_{ij}}) - \sigma_{e,ij}\mathbb{I}(|s_{ij}|\leq C\varsigma_T \sqrt{\hat{\theta}_{ij}}) \right| \notag 
			\\& \leq \max_i \sum_{j=1}^p|\chi_{ij}(s_{ij}) - s_{ij} |\mathbb{I}(|s_{ij}|> C\varsigma_T \sqrt{\hat{\theta}_{ij}}) + \sum_{j=1}^p |s_{ij}-\sigma_{e,ij}|\mathbb{I}(|s_{ij}|> C\varsigma_T \sqrt{\hat{\theta}_{ij}})
			\notag
			\\& 
			+ \sum_{j=1}^p|\sigma_{e,ij}|\mathbb{I}(|s_{ij}|\leq C\varsigma_T \sqrt{\hat{\theta}_{ij}}) \notag
			\\ &
			\leq \max_i \sum_{j=1}^p C\varsigma_T \sqrt{\hat{\theta}_{ij}}
			\mathbb{I}(|s_{ij}|> C\varsigma_T \theta_2) + M\varsigma_T \sum_{j=1}^p \mathbb{I}(|s_{ij}|> C\varsigma_T \theta_2) + \sum_{j=1}^p |\sigma_{e,ij}|\mathbb{I}(|s_{ij}|< C\varsigma_T \theta_1) \notag
			\\ & \leq (C\theta_1 + M)\varsigma_T \max_i \sum_{j=1}^p\mathbb{I}(|\sigma_{e,ij}|> M\varsigma_T ) + \max_i
			\sum_{j=1}^p |\sigma_{e,ij}|\mathbb{I}(|\sigma_{e,ij}|< C\varsigma_T \theta_1 + M\varsigma_T) \notag
			\\ & \leq (C\theta_1 + M)\varsigma_T \max_i \sum_{j=1}^p \frac{|\sigma_{e,ij}|^q}{M^q\varsigma_T^q}\mathbb{I}(|\sigma_{e,ij}|> M\varsigma_T ) \notag
			\\ & + \max_i
			\sum_{j=1}^p  |\sigma_{e,ij}| \frac{(C\theta_1 + M)^{1-q}\varsigma_T^{1-q}}{|\sigma_{e,ij}|^{1-q}}\mathbb{I}(|\sigma_{e,ij}|< (C\theta_1+M)\varsigma_T ) \notag
			\\ & \leq
			\frac{C\theta_1 + M}{M^q} \varsigma_T^{1-q}\max_i \sum_{j=1}^p |\sigma_{e,ij}|^q + (C\theta_1 + M)^{1-q}\varsigma_T^{1-q}\max_i \sum_{j=1}^p|\sigma_{e,ij}|^q \notag
			\\ & = \kappa_q\varsigma_T^{1-q}(C\theta_1 + M)\left(M^{-q} + (C\theta_1 + M)^{-q} \right).
		\end{align}
		where $\kappa_q = \max_{i}\sum_{j\leq p}|\sigma_{e,ij}|^{q}$.
		Let $M_1 = (C\theta_1 + M)\left(M^{-q} + (C\theta_1 + M)^{-q} \right)$, then with probability at least $1-2x$, $\left\|\hat{\Sigma}_e - \Sigma_e  \right\| \leq M_1\kappa_q\varsigma_T^{1-q}$. Since $x$ is arbitrary, we have $\left\|\hat{\Sigma}_e - \Sigma_e  \right\| = O_p(\kappa_q\varsigma_T^{1-q})$. If in addition $\varsigma_T^{1-q}\kappa_q = o(1)$, then the minimum eigenvalue of $\hat{\Sigma}_e$ is bounded away from zero with probability approaching one since $\lambda_{min}(\Sigma_e) > c_1$ by Assumption 2(b). This then implies $\left\|\hat{\Sigma}_e^{-1} - \Sigma_e^{-1}  \right\| = O_p(\kappa_q\varsigma_T^{1-q})$. Finally, note that 
		$$
		\varsigma_T \asymp \sqrt{\frac{logp}{T}} + \sqrt{\frac{(\delta^{1/2}+T^{1/4}+\Delta)^2}{p}} + \sqrt{\frac{\delta\sqrt{p}}{T}} \asymp \zeta_T :=  \frac{\delta^{1/2}+T^{1/4}+\Delta}{\sqrt{p}}+\frac{\delta^{1/2}p^{1/4}}{\sqrt{T}}. 
		$$
		We complete the proof of Theorem 1.
		
	\end{proof}

	\bigskip
	
	\bigskip
	
	Define $$G_T = \hat{B} - B\tilde{H}^\top$$ where $\hat{B} = \left(\hat{b}_1,\ldots, \hat{b}_p \right)^\top$.
	\begin{lemma}
		\label{lemma: GT related}
		Under the conditions of Theorem 1, 
		\begin{itemize}
			\item[(i)] $\left\|G_T \right\|_F^2 = O_p\left( 1 + \frac{p^{3/2}}{T} \right)$.
			\item[(ii)]      $\left\| \hat{B}^\top \hat{\Sigma}_e^{-1}\hat{B} - (B\tilde{H}^\top)^\top{\Sigma}_e^{-1} B\tilde{H}^\top \right\| = O_p\left(p\left(\varpi_T+ \zeta_T^{1-q}\kappa_q\right) \right) = o_p(p)$.
		\end{itemize}
	\end{lemma}
	
	\begin{proof}[Proof of Lemma \ref{lemma: GT related}]
		(i) We have $\left\|G_T \right\|_F^2 \leq \max_i p\left\|\hat{b}_i - \tilde{H}b_i \right\|^2 = O_p\left( 1 + \frac{p^{3/2}}{T} \right)$. 
		
		(ii)   By result (i) and Theorem 1,  we have
		\begin{align*}
			&\left\| \hat{B}^\top \hat{\Sigma}_e^{-1}\hat{B} - (B\tilde{H}^\top)^\top{\Sigma}_e^{-1} B\tilde{H}^\top \right\| \\& \leq \left\|G_T^\top\hat{\Sigma}_e^{-1}G_T \right\|
			+ 2 \left\|G_T^\top \hat{\Sigma}_e^{-1}B\tilde{H}^\top \right\| + \left\|B\tilde{H}^\top(\hat{\Sigma}_e^{-1} - {\Sigma}_e^{-1})B\tilde{H}^\top \right\| 
			\\ &= O_p\left(p\left(\varpi_T+ \zeta_T^{1-q}\kappa_q\right) \right),
		\end{align*}
		where $\varpi_T = \frac{1}{\sqrt{p}} + \frac{p^{1/4}}{\sqrt{T}}$.
	\end{proof}

	\begin{lemma}
		\label{lemma Omega}
		Under the conditions of Theorem 1, with probability approaching one, for some $c>0$,
		\begin{itemize}
			\item[(i)] $\lambda_{min}\left(I_m + (B\tilde{H}^\top)^\top{\Sigma}_e^{-1} B\tilde{H}^\top\right) \geq cp$.
			\item[(ii)] $\lambda_{min}\left(I_m + \hat{B}^\top \hat{\Sigma}_e^{-1}\hat{B} \right) \geq cp$.
			\item[(iii)]   $\lambda_{min}\left(I_m + B^\top{\Sigma}_e^{-1} B\right) \geq cp$.
			\item[(iv)]   $\lambda_{min}\left((\tilde{H}\tilde{H}^\top)^{-1}+ B^\top{\Sigma}_e^{-1} B\right) \geq cp$.
		\end{itemize}
	\end{lemma}
	
	\begin{proof}[Proof of Lemma \ref{lemma Omega}]
		(i) By Lemma \ref{lemma Appendix C10}, with probability approaching one, $\lambda_{min}(\tilde{H}\tilde{H}^\top )$ is bounded away from zero. Hence,
		\begin{align*}
			\lambda_{min}\left(I_m + (B\tilde{H}^\top)^\top{\Sigma}_e^{-1} B\tilde{H}^\top\right) &\geq \lambda_{min}\left( (B\tilde{H}^\top)^\top{\Sigma}_e^{-1} B\tilde{H}^\top\right)
			\\& \geq \lambda_{min}\left( {\Sigma}_e^{-1} \right) \lambda_{min}\left( (B\tilde{H}^\top)^\top B\tilde{H}^\top\right)
			\\ & \geq \lambda_{min}\left( {\Sigma}_e^{-1} \right) \lambda_{min}\left( B^\top B\right)\lambda_{min}\left( \tilde{H}\tilde{H}^\top \right) \\& \geq cp.
		\end{align*}
		
		(ii) The result follows from part (i) and Lemma \ref{lemma: GT related}(ii). Part (iii) and (iv) follow from a similar argument of part (i) and Lemma \ref{lemma Appendix C10}.
	\end{proof}
	
	\bigskip
	
	\bigskip

	\begin{proof}[Proof of Theorem 2.]
		Define $$\tilde{\Sigma}_r = B\tilde{H}^\top \tilde{H}B^\top + \Sigma_e.$$
		Note that $\hat{\Sigma}_r = \hat{B}\hat{B}^\top + \hat{\Sigma}_e$ and $\Sigma_r = BB^\top + \Sigma_e$. Then by triangular inequality, it follows that 
		$$
		\left\|\hat{\Sigma}_r^{-1} -\Sigma_r^{-1} \right\| \leq \left\|\hat{\Sigma}_r^{-1} - \tilde{\Sigma}_r^{-1} \right\| + \left\| \tilde{\Sigma}_r^{-1} - \Sigma_r^{-1} \right\|.
		$$
		Using the Sherman-Morrison-Woodbury formula, we have $\left\|\hat{\Sigma}_r^{-1} - \tilde{\Sigma}_r^{-1} \right\| \leq \sum_{i=1}^6 L_i$, where
		\begin{align}
			&L_1 = \left\|\hat{\Sigma}_e^{-1} - \tilde{\Sigma}_e^{-1} \right\| \notag
			\\
			&
			L_2 = \left\|(\hat{\Sigma}_e^{-1} - \tilde{\Sigma}_e^{-1})\hat{B}\left[ I_m + \hat{B}^\top\hat{\Sigma}_e^{-1}\hat{B} \right]^{-1} \hat{B}^{\top}\hat{\Sigma}_e^{-1} \right\| \notag
			\\ &
			L_3 = \left\|(\hat{\Sigma}_e^{-1} - \tilde{\Sigma}_e^{-1})\hat{B}\left[ I_m + \hat{B}^\top\hat{\Sigma}_e^{-1}\hat{B} \right]^{-1} \hat{B}^{\top}\Sigma_e^{-1}\right\| \notag
			\\ &
			L_4 = \left\| \Sigma_e^{-1}(\hat{B}-B\tilde{H}^\top)\left[ I_m + \hat{B}^\top\hat{\Sigma}_e^{-1}\hat{B} \right]^{-1} \hat{B}^{\top}\Sigma_e^{-1}\right\| \notag
			\\ &
			L_5 = \left\| \Sigma_e^{-1}(\hat{B}-B\tilde{H}^\top)\left[ I_m + \hat{B}^\top\hat{\Sigma}_e^{-1}\hat{B} \right]^{-1} \tilde{H}B^\top \Sigma_e^{-1}\right\| \notag
			\\ &
			L_6 = \left\| \Sigma_e^{-1}B\tilde{H}^\top\left( \left[ I_m + \hat{B}^\top\hat{\Sigma}_e^{-1}\hat{B} \right]^{-1} -\left[ I_m + \tilde{H}B^\top \Sigma_e^{-1}B\tilde{H}^\top \right]^{-1} \right)\tilde{H}B^\top\Sigma_e^{-1} \right\|
			\label{eq: L1-L6}
		\end{align}
		
		Now, we bound each of the six terms respectively. For $L_1$, it is bounded by Theorem 1. Let $\Omega = \left[ I_m + \hat{B}^\top\hat{\Sigma}_e^{-1}\hat{B} \right]^{-1} $, then
		$$
		L_2 \leq \left\| \hat{\Sigma}_e^{-1} - {\Sigma}_e^{-1} \right\| \left\|\hat{B}^\top \Omega \hat{B}^\top \right\| \left\|\hat{\Sigma}_e^{-1} \right\|. 
		$$
		Note that Theorem 1 implies that $\left\| \hat{\Sigma}_e^{-1} \right\| = O_p(1)$. Lemma \ref{lemma Omega} implies that $\|\Omega \| = O_p(p^{-1})$. This shows that $L_2 = O_p(L_1)$. Similarly, $L_3 = O_p(L_1)$. In addition, since $\|G_T \|_F^2 = O_p(1 + \frac{p^{3/2}}{T})$, $L_4 \leq \left\| \Sigma_e^{-1}(\hat{B}-B\tilde{H}^\top)\right\| \left\|\Omega \right\| \left\| \hat{B}^{\top}\Sigma_e^{-1}\right\| = O_p(\varpi_T)$. Similarly, $L_5 = O_p(L_4)$. 
		Finally, let $\Omega_1 = \left[ I_m + (B\tilde{H}^\top)^\top\Sigma_e^{-1}B\tilde{H}^\top\right]^{-1}$. By Lemma \ref{lemma Omega}, $\|\Omega_1 \| = O_p(p^{-1})$, then by Lemma \ref{lemma: GT related}(ii), 
		\begin{align*}
			\left\| \Omega-\Omega_1 \right\| & = \left\| \Omega(\Omega^{-1}-\Omega_1^{-1})\Omega_1 \right\| \leq O_p(p^{-2})\left\| (BH^\top)\Sigma_e^{-1}BH^\top - \hat{B}^\top\hat{\Sigma}_e^{-1}\hat{B}  \right\| = O_p\left(p^{-1}\zeta_T^{1-q}\kappa_q + p^{-1}\varpi_T\right).
		\end{align*}
		As a result, $L_6 \leq \left\| \Sigma_e^{-1}BH^\top\right\|^2\left\| \Omega-\Omega_1 \right\| = O_p\left(\zeta_T^{1-q}\kappa_q + \varpi_T\right)$. Adding up $L_1 - L_6$ gives 
		$$
		\left\|\hat{\Sigma}_r^{-1} - \tilde{\Sigma}_r^{-1} \right\|= O_p\left(\zeta_T^{1-q}\kappa_q + \varpi_T\right).
		$$
		Note that 
		$\varpi_T/\zeta_T = o(1)$. As a result,  $$
		\left\|\hat{\Sigma}_r^{-1} - \tilde{\Sigma}_r^{-1} \right\|= O_p\left(\zeta_T^{1-q}\kappa_q \right).
		$$

		On the other hand, using Sherman-Morrison-Woodbury formula, we have
		\begin{align*}
			\left\|\tilde{\Sigma}_r^{-1} - {\Sigma}_r^{-1} \right\| & \leq \left\| \Sigma_e^{-1}B\left( \left[(\tilde{H}^\top \tilde{H})^{-1} + B^\top\Sigma_e^{-1}B\right]^{-1} - \left[ I_m + {B}^\top{\Sigma}_e^{-1}{B} \right]^{-1} \right)B^\top\Sigma_e^{-1} \right\|
			\\ &
			\leq 
			O_p(p)\left\| \left[(\tilde{H}^\top \tilde{H})^{-1} + B^\top\Sigma_e^{-1}B\right]^{-1} - \left[ I_m + {B}^\top{\Sigma}_e^{-1}{B} \right]^{-1}\right\|
			\\& = O_p(p^{-1}) \left\|(\tilde{H}^\top \tilde{H})^{-1} - I_m \right\| = o_p(\zeta_T^{1-q}\kappa_q).
		\end{align*}
		
	\end{proof}
	
	\bigskip
	
	\bigskip

	\begin{proof}[Proof of Theorem 3]
		Recall that
		\begin{eqnarray*}
			&&\hat{R}_{\min}=\frac{1}{1_p^{\top}\hat{\Sigma}^{-1}_{r}1_{p}}, \ \ R_{\min}=\frac{1}{1_p^{\top}\Sigma_{r}^{-1}1_p}.
		\end{eqnarray*}
		Then, we have 
		\begin{eqnarray}\label{eq Rmin}
			\left|\frac{\hat{R}_{\min}}{R_{\min}}-1\right|
			=\left|\frac{1_p^{\top}\Sigma_{r}^{-1}1_p}{1_p^{\top}\hat{\Sigma}_{r}^{-1}1_p}-1\right|
			=\frac{\left|1_p^{\top}\Sigma_{r}^{-1}1_p-1_p^{\top}\hat{\Sigma}_{r}^{-1}1_p\right|}{1_p^{\top}\hat{\Sigma}_{r}^{-1}1_p}
			\leq\frac{p\cdot\left \Vert \Sigma_{r}^{-1}-\hat{\Sigma}_{r}^{-1}\right \Vert }{1_p^{\top}\hat{\Sigma}_{r}^{-1}1_p}. 
		\end{eqnarray}
		The bound for the numerator of (\ref{eq Rmin}) is $O_p\left(p\zeta_T^{1-q}\kappa_q\right)$ by  Theorem 2. The exact order of the denominator of (\ref{eq Rmin}) is $p^{1-\eta}$, which is guaranteed by Assumption 6 and the same argument as the proposition of \cite{ding2021high}. In view of these results, we have the result 
		\begin{eqnarray}\label{y2}
			\left|\frac{\hat{R}_{\min}}{R_{\min}}-1\right|=O_p\left(p^{\eta}\zeta_T^{1-q}\kappa_q\right). 
		\end{eqnarray}
	\end{proof}
	
	
	\bigskip
	
	\bigskip
	
	\begin{proof}[Proof of Theorem 4]
		We decompose the Sharpe ratio in the following way. 
		\begin{eqnarray*}
			\frac{\hat{SR}-SR}{SR}
			=\frac{Z_1+Z_2}{\sqrt{1_p^{\top}\hat{\Sigma}_{r}^{-1}1_p}\cdot1_p^{\top}\Sigma_{r}^{-1}\mu}, 
		\end{eqnarray*}
		where
		\begin{eqnarray*}
			&&Z_1=1_p^{\top}\left(\hat{\Sigma}_{r}^{-1}\hat{\mu}-\Sigma_{r}^{-1}\mu\right)\cdot\sqrt{1_p^{\top}\Sigma_{r}^{-1}1_p},\\
			&&Z_2=1_p^{\top}\Sigma_{r}^{-1}\mu\left(\sqrt{1_p^\top\Sigma_{r}^{-1}1_p}-\sqrt{1_p^\top\hat{\Sigma}_{r}^{-1}1_p}\right). 
		\end{eqnarray*}
		Now we consider the first term involving $Z_1$, 
		\begin{eqnarray}\label{yu90}
			&&\left|J_1\right|
			:=\frac{|Z_1|}{\sqrt{1_p^\top\hat{\Sigma}_{r}^{-1}1_p}\cdot\left|1_p^\top\Sigma_{r}^{-1}\mu\right|}\non
			&\leq&\frac{\left|1_p^\top\left(\hat{\Sigma}_{r}^{-1}-\Sigma_{r}^{-1}\right)\hat{\mu}\right|+\left|1_p^\top\Sigma_{r}^{-1}\left(\hat{\mu}-\mu\right)\right|}{\sqrt{1_p^\top\hat{\Sigma}_{r}^{-1}1_p}\cdot\left|1_p^\top\Sigma_{r}^{-1}\mu\right|}\cdot\sqrt{1_p^\top\Sigma_{r}^{-1}1_p}\non
			&\leq&C \left(p^{\phi}\left|\left|\hat{\Sigma}_{r}^{-1}-\Sigma_{r}^{-1}\right|\right|+p^{\phi-1/2}\left|\left|\hat{\mu}-\mu\right|\right|\right)\non
			&=&O_p\left(p^{\phi}\zeta_T^{1-q}\kappa_q+p^{\phi}\tilde{\zeta}_T\right), 
		\end{eqnarray}
		where we apply Assumption 7 to determine the order of magnitude of $1_p^\top \Sigma_r^{-1}\mu$,  the last equality uses the results that 
		\begin{align*}
			\left|\left|\hat{\mu}-\mu\right|\right| &= \left\| \hat{B}\hat{\mu}_f - B\mu_f \right\|
			\\& 
			\leq \left\| \hat{B}\hat{\mu}_f - B\tilde{H}^\top \hat{\mu}_f \right\| + \left\| B\frac{1}{T}\sum_{t=1}^T\left(\tilde{H}^\top\hat{F}_t - F_t  \right) \right\|
			\\& 
			\leq  \left\| \hat{B} - B\tilde{H}^\top \right\|\left\|\hat{\mu}_f \right\| + \left\| B \right\| \max_t\left\| \tilde{H}^\top\hat{F}_t - F_t   \right\|
			\\
			&=O_p\left(\sqrt{p}\varpi_T + \sqrt{p}\left(\frac{\delta^{1/2}+T^{1/4}+\Delta}{\sqrt{p}}+\sqrt{\frac{logp}{T}}\right)\right)
			\\& 
			= O_p\left(\sqrt{p}\tilde{\zeta}_T \right)
			, 
		\end{align*}
		where $\hat{\mu}_f = T^{-1}\sum_{t=1}^T \hat{F}_t$ and $\tilde{\zeta}_T = \frac{\delta^{1/2}+T^{1/4}+\Delta}{\sqrt{p}} + \frac{p^{1/4}}{\sqrt{T}}$. 
		
		Similarly, we can get
		\begin{eqnarray}\label{yu91}
			&&\left|J_2\right|:=\frac{|Z_2|}{\sqrt{1_p^\top\hat{\Sigma}_{r}^{-1}1_p}\cdot\left|1_p^\top\Sigma_{r}^{-1}\mu\right|}
			\leq\frac{\left|\sqrt{1_p^\top\Sigma_{r}^{-1}1_p}-\sqrt{1_p^\top\hat{\Sigma}_{r}^{-1}1_p}\right|}{\sqrt{1_p^\top\hat{\Sigma}_{r}^{-1}1_p}} \notag\\
			&=&\frac{\left|1_p^\top\left(\Sigma_{r}^{-1}-\hat{\Sigma}_{r}^{-1}\right)1_p\right|}{\sqrt{1_p^\top\hat{\Sigma}_{r}^{-1}1_p}\left(\sqrt{1_p^\top\Sigma_{r}^{-1}1_p}+\sqrt{1_p^\top\hat{\Sigma}_{r}^{-1}1_p}\right)}
			\leq {p^\eta}\left|\left|\hat{\Sigma}_{r}^{-1}-\Sigma_{r}^{-1}\right|\right|\non
			&=&O_p\left(p^\eta\zeta_T^{1-q}\kappa_q\right). 
		\end{eqnarray}
		Then the result is derived from (\ref{yu90}) and (\ref{yu91}). 
	\end{proof}	
	
\end{appendices}
	
	\bibliographystyle{chicago}
	\bibliography{reference}

\begin{thebibliography}{}

\bibitem[\protect\citeauthoryear{Ait-Sahalia and Xiu}{Ait-Sahalia and
  Xiu}{2017}]{ait17}
Ait-Sahalia, Y. and D.~Xiu (2017).
\newblock Using principal component analysis to estimate a high dimensional
  factor model with high-frequency data.
\newblock {\em Journal of Econometrics\/}~{\em 201\/}(2), 384 -- 399.

\bibitem[\protect\citeauthoryear{Ao, Li, and Zheng}{Ao
  et~al.}{2019}]{ao2019approaching}
Ao, M., Y.~Li, and X.~Zheng (2019).
\newblock Approaching mean-variance efficiency for large portfolios.
\newblock {\em The Review of Financial Studies\/}~{\em 32\/}(7), 2890--2919.

\bibitem[\protect\citeauthoryear{Bai}{Bai}{2003}]{bai2003}
Bai, J. (2003).
\newblock Inferential theory for factor models of large dimensions.
\newblock {\em Econometrica\/}~{\em 71\/}(1), 135--171.

\bibitem[\protect\citeauthoryear{Bai and Ng}{Bai and
  Ng}{2002}]{bai2002determining}
Bai, J. and S.~Ng (2002).
\newblock Determining the number of factors in approximate factor models.
\newblock {\em Econometrica\/}~{\em 70\/}(1), 191--221.

\bibitem[\protect\citeauthoryear{Blanchet, Chen, and Zhou}{Blanchet
  et~al.}{2022}]{Blanchet2022}
Blanchet, J., L.~Chen, and X.~Y. Zhou (2022).
\newblock Distributionally robust mean-variance portfolio selection with
  wasserstein distances.
\newblock {\em Management Science\/}~{\em 68\/}(9), 6382--6410.

\bibitem[\protect\citeauthoryear{Cai and Liu}{Cai and
  Liu}{2011}]{cai2011adaptive}
Cai, T. and W.~Liu (2011).
\newblock Adaptive thresholding for sparse covariance matrix estimation.
\newblock {\em Journal of the American Statistical Association\/}~{\em
  106\/}(494), 672--684.

\bibitem[\protect\citeauthoryear{Caner, Medeiros, and Vasconcelos}{Caner
  et~al.}{2023}]{caner2020sharpe}
Caner, M., M.~Medeiros, and G.~Vasconcelos (2023).
\newblock Sharpe ratio analysis in high dimensions: Residual-based nodewise
  regression in factor models.
\newblock {\em Journal of Econometrics\/}~{\em 235}, 393--417.

\bibitem[\protect\citeauthoryear{Chamberlain and Rothschild}{Chamberlain and
  Rothschild}{1983}]{Chamberlain83}
Chamberlain, G. and M.~Rothschild (1983).
\newblock Arbitrage, factor structure, and mean-variance analysis on large
  asset markets.
\newblock {\em Econometrica\/}~{\em 51}, 1281 -- 1304.

\bibitem[\protect\citeauthoryear{Delage and Ye}{Delage and
  Ye}{2010}]{delage2010distributionally}
Delage, E. and Y.~Ye (2010).
\newblock Distributionally robust optimization under moment uncertainty with
  application to data-driven problems.
\newblock {\em Operations Research\/}~{\em 58\/}(3), 595--612.

\bibitem[\protect\citeauthoryear{DeMiguel, Garlappi, Nogales, and
  Uppal}{DeMiguel et~al.}{2009}]{demiguel2009generalized}
DeMiguel, V., L.~Garlappi, F.~J. Nogales, and R.~Uppal (2009).
\newblock A generalized approach to portfolio optimization: Improving
  performance by constraining portfolio norms.
\newblock {\em Management Science\/}~{\em 55\/}(5), 798--812.

\bibitem[\protect\citeauthoryear{DeMiguel, Garlappi, and Uppal}{DeMiguel
  et~al.}{2009}]{demiguel2009optimal}
DeMiguel, V., L.~Garlappi, and R.~Uppal (2009).
\newblock Optimal versus naive diversification: How inefficient is the 1/n
  portfolio strategy?
\newblock {\em The Review of Financial Studies\/}~{\em 22\/}(5), 1915--1953.

\bibitem[\protect\citeauthoryear{DeMiguel and Nogales}{DeMiguel and
  Nogales}{2009}]{demiguel2009or}
DeMiguel, V. and F.~Nogales (2009).
\newblock Portfolio selection with robust estimation.
\newblock {\em Operations Research\/}~{\em 57\/}(3), 560--577.

\bibitem[\protect\citeauthoryear{Ding, Li, and Zheng}{Ding
  et~al.}{2021}]{ding2021high}
Ding, Y., Y.~Li, and X.~Zheng (2021).
\newblock High dimensional minimum variance portfolio estimation under
  statistical factor models.
\newblock {\em Journal of Econometrics\/}~{\em 222\/}(1B), 502--515.

\bibitem[\protect\citeauthoryear{Fama and French}{Fama and
  French}{1992}]{fama1992crosssection}
Fama, E.~F. and K.~R. French (1992).
\newblock The cross-section of expected stock returns.
\newblock {\em Journal of Finance\/}~{\em 47\/}(2), 427--465.

\bibitem[\protect\citeauthoryear{Fama and French}{Fama and
  French}{1993}]{fama1993common}
Fama, E.~F. and K.~R. French (1993).
\newblock Common risk factors in the returns on stocks and bonds.
\newblock {\em Journal of Financial Economics\/}~{\em 33\/}(1), 3--56.

\bibitem[\protect\citeauthoryear{Fama and French}{Fama and
  French}{2015}]{fama2015five}
Fama, E.~F. and K.~R. French (2015).
\newblock A five-factor asset pricing model.
\newblock {\em Journal of Financial Economics\/}~{\em 116\/}(1), 1--22.

\bibitem[\protect\citeauthoryear{Fan, Liao, and Mincheva}{Fan
  et~al.}{2011}]{fan2011high}
Fan, J., Y.~Liao, and M.~Mincheva (2011).
\newblock High dimensional covariance matrix estimation in approximate factor
  models.
\newblock {\em The Annals of Statistics\/}~{\em 39\/}(6), 3320.

\bibitem[\protect\citeauthoryear{Fan, Liao, and Mincheva}{Fan
  et~al.}{2013}]{fan2013large}
Fan, J., Y.~Liao, and M.~Mincheva (2013).
\newblock Large covariance estimation by thresholding principal orthogonal
  complements.
\newblock {\em Journal of the Royal Statistical Society: Series B (Statistical
  Methodology)\/}~{\em 75\/}(4), 603--680.

\bibitem[\protect\citeauthoryear{Fan, Liu, and Wang}{Fan
  et~al.}{2018}]{fan2018aos}
Fan, J., H.~Liu, and W.~Wang (2018).
\newblock Large covariance estimation through elliptical factor models.
\newblock {\em The Annals of Statistics\/}~{\em 46\/}(4), 1383 -- 1414.

\bibitem[\protect\citeauthoryear{Fan, Wang, and Zhong}{Fan
  et~al.}{2019}]{fan2019robust}
Fan, J., W.~Wang, and Y.~Zhong (2019).
\newblock Robust covariance estimation for approximate factor models.
\newblock {\em Journal of Econometrics\/}~{\em 208\/}(1), 5--22.

\bibitem[\protect\citeauthoryear{Fan, Zhang, and Yu}{Fan
  et~al.}{2012}]{fan2012vast}
Fan, J., J.~Zhang, and K.~Yu (2012).
\newblock Vast portfolio selection with gross-exposure constraints.
\newblock {\em Journal of the American Statistical Association\/}~{\em
  107\/}(498), 592--606.

\bibitem[\protect\citeauthoryear{Fan, Wu, Yang, and Zhong}{Fan
  et~al.}{2024}]{fan2022}
Fan, Q., R.~Wu, Y.~Yang, and W.~Zhong (2024).
\newblock Time-varying minimum variance portfolio.
\newblock {\em Journal of Econometrics\/}~{\em 239}, 105339.

\bibitem[\protect\citeauthoryear{Feng, Giglio, and Xiu}{Feng
  et~al.}{2020}]{feng2020taming}
Feng, G., S.~Giglio, and D.~Xiu (2020).
\newblock Taming the factor zoo: A test of new factors.
\newblock {\em The Journal of Finance\/}~{\em 75\/}(3), 1327--1370.

\bibitem[\protect\citeauthoryear{Giglio, Kelly, and Xiu}{Giglio
  et~al.}{2022}]{giglio22}
Giglio, S., B.~Kelly, and D.~Xiu (2022).
\newblock Factor models, machine learning, and asset pricing.
\newblock {\em Annual Review of Financial Economics\/}~{\em 14}, 337--368.

\bibitem[\protect\citeauthoryear{Harvey and Liu}{Harvey and
  Liu}{2021}]{harvey2021factor}
Harvey, C.~R. and Y.~Liu (2021).
\newblock Lucky factors.
\newblock {\em Journal of Financial Economics\/}~{\em 141\/}(2), 413--435.

\bibitem[\protect\citeauthoryear{Huber}{Huber}{1964}]{huber1964}
Huber, P.~J. (1964).
\newblock Robust estimation of a location parameter.
\newblock {\em The Annals of Statistics\/}~{\em 53}, 73--101.

\bibitem[\protect\citeauthoryear{Huber}{Huber}{1973}]{huber1973}
Huber, P.~J. (1973).
\newblock Robust regression: Asymptotics, conjectures and monte carlo.
\newblock {\em The Annals of Statistics\/}~{\em 1}, 799--821.

\bibitem[\protect\citeauthoryear{Jagannathan and Ma}{Jagannathan and
  Ma}{2003}]{jagannathan2003risk}
Jagannathan, R. and T.~Ma (2003).
\newblock Risk reduction in large portfolios: Why imposing the wrong
  constraints helps.
\newblock {\em The Journal of Finance\/}~{\em 58\/}(4), 1651--1683.

\bibitem[\protect\citeauthoryear{Ledoit and Wolf}{Ledoit and
  Wolf}{2003}]{ledoit2003improved}
Ledoit, O. and M.~Wolf (2003).
\newblock Improved estimation of the covariance matrix of stock returns with an
  application to portfolio selection.
\newblock {\em Journal of Empirical Finance\/}~{\em 10\/}(5), 603--621.

\bibitem[\protect\citeauthoryear{Ledoit and Wolf}{Ledoit and
  Wolf}{2017}]{ledoit2017nonlinear}
Ledoit, O. and M.~Wolf (2017).
\newblock Nonlinear shrinkage of the covariance matrix for portfolio selection:
  Markowitz meets goldilocks.
\newblock {\em The Review of Financial Studies\/}~{\em 30\/}(12), 4349--4388.

\bibitem[\protect\citeauthoryear{Li, Huang, Yang, and Zhang}{Li
  et~al.}{2022}]{li2022synthetic}
Li, G., L.~Huang, J.~Yang, and W.~Zhang (2022).
\newblock A synthetic regression model for large portfolio allocation.
\newblock {\em Journal of Business \& Economic Statistics\/}~{\em 40\/}(4),
  1665--1677.

\bibitem[\protect\citeauthoryear{Li and Li}{Li and
  Li}{2022}]{li2022integrative}
Li, Q. and L.~Li (2022).
\newblock Integrative factor regression and its inference for multimodal data
  analysis.
\newblock {\em Journal of the American Statistical Association\/}~{\em
  117\/}(540), 2207--2221.

\bibitem[\protect\citeauthoryear{Markowitz}{Markowitz}{1952}]{markowitz1952portfolio}
Markowitz, H. (1952).
\newblock Portfolio selection.
\newblock {\em Journal of Finance\/}, 77--91.

\bibitem[\protect\citeauthoryear{Maronna}{Maronna}{2005}]{M2005}
Maronna, R. (2005).
\newblock Principal components and orthogonal regression based on robust
  scales.
\newblock {\em Technometrics\/}~{\em 47\/}(3), 264--273.

\bibitem[\protect\citeauthoryear{Petukhina, Klochkov, H{\"a}rdle, and
  Zhivotovskiy}{Petukhina et~al.}{2024}]{Hardle2022}
Petukhina, A., Y.~Klochkov, W.~H{\"a}rdle, and N.~Zhivotovskiy (2024).
\newblock Robustifying {M}arkowitz.
\newblock {\em Journal of Econometrics\/}~{\em 239\/}(2), 105387.

\bibitem[\protect\citeauthoryear{Plachel}{Plachel}{2019}]{plachel2019}
Plachel, L. (2019).
\newblock A unified model for regularized and robust portfolio optimization.
\newblock {\em Journal of Economic Dynamics and Control\/}~{\em 109}, 103779.

\bibitem[\protect\citeauthoryear{Ross}{Ross}{1976}]{ross1976arbitrage}
Ross, S.~A. (1976).
\newblock The arbitrage theory of capital asset pricing.
\newblock {\em Journal of Economic Theory\/}~{\em 13\/}(3), 341--360.

\bibitem[\protect\citeauthoryear{Su and Wang}{Su and Wang}{2017}]{su2017time}
Su, L. and X.~Wang (2017).
\newblock On time-varying factor models: Estimation and testing.
\newblock {\em Journal of Econometrics\/}~{\em 198\/}(1), 84--101.

\bibitem[\protect\citeauthoryear{Wang, Peng, Li, and Leng}{Wang
  et~al.}{2021}]{WANG202153}
Wang, H., B.~Peng, D.~Li, and C.~Leng (2021).
\newblock Nonparametric estimation of large covariance matrices with
  conditional sparsity.
\newblock {\em Journal of Econometrics\/}~{\em 223\/}(1), 53--72.

\end{thebibliography}

\end{document}